\documentclass[11pt]{article}%

\usepackage[T1]{fontenc}
\usepackage[utf8]{inputenc}

\usepackage[margin=1.25in]{geometry}
\usepackage[dvipsnames]{xcolor}
\usepackage[
  colorlinks,
  citecolor=blue,
  linkcolor=blue,
  urlcolor=blue
]{hyperref}

\usepackage{amsmath, amssymb, amsthm}
\usepackage{mathtools}
\usepackage{bm}
\usepackage{array}
\usepackage{graphicx}
\usepackage{booktabs}
\usepackage{longtable}
\usepackage{multirow}
\usepackage{float}
\usepackage[
  font=footnotesize,
  justification=centering
]{caption}
\usepackage{siunitx}
\newcolumntype{L}[1]{>{\raggedright\arraybackslash}p{#1}}
\newcolumntype{C}[1]{>{\centering\arraybackslash}p{#1}}
\usepackage{enumitem}
\usepackage{microtype}
\allowdisplaybreaks[2]

\DeclareMathOperator*{\argmin}{arg\,min}
\DeclareMathOperator*{\MSE}{MSE}

\DeclareMathOperator*{\Cov}{Cov}

\newcommand{\R}{\mathbb{R}}

\providecommand{\pto}{\overset{p}{\to}}
\providecommand{\dto}{\rightsquigarrow}

\providecommand{\calH}{\mathcal H}
\providecommand{\calM}{\mathcal M}

\providecommand{\calA}{\mathcal A}
\providecommand{\calB}{\mathcal B}
\providecommand{\calW}{\mathcal W}

\providecommand{\one}{\mathbf 1}

\theoremstyle{plain}
\newtheorem{lemma}{Lemma}[section]
\newtheorem{theorem}{Theorem}[section]
\newtheorem{proposition}{Proposition}[section]
\newtheorem{corollary}{Corollary}[section]
\newtheorem{assumption}{Assumption}[section]

\theoremstyle{definition}
\newtheorem{definition}{Definition}[section]



\usepackage[round]{natbib}
\usepackage{url}

\title{Quasi-Bayesian Hierarchical Models}
\author{Desmond Fairall\thanks{Department of Economics, University of Warwick. Email: \href{mailto:desmond.fairall@warwick.ac.uk}{desmond.fairall@warwick.ac.uk}}\\University of Warwick
    \and Thomas Glinnan\thanks{\textit{Corresponding Author}. Department of Economics, London School of Economics and Political Science. Email: \href{mailto:t.m.glinnan@lse.ac.uk}{t.m.glinnan@lse.ac.uk}}\\LSE}
\date{}

\begin{document}
\maketitle

\begin{abstract}
We develop the Quasi-Bayesian Hierarchical Model (QBHM) for grouped GMM settings. The framework combines Bayesian hierarchical modelling with Laplace-type estimation: it preserves each group-specific objective function, while introducing a pooling term for economically comparable parameters. When the number of studies is fixed, the QBHM estimator---the quasi-posterior mean---has the same asymptotic distribution as GMM when estimating strongly identified study parameters. For weakly identified studies, we analyze the asymptotic properties of the method via a weak-GMM limit experiment: an asymptotic approximation in which the sample-moment criterion remains a random function over the weak parameter space, and the upper-level pooling relation induces a family of priors over weak values. In this experiment, the weak-limit QBHM rule is a Bayes rule under squared loss for the hierarchy-induced weak-limit prior, which provides a decision-theoretic justification for our procedure. We also extend our results to mixed within-study blocks, allowing a single study to contain both strongly and weakly identified parameters. Pooling can also reduce the pointwise asymptotic mean squared error (MSE) relative to unpooled estimation when the bias--variance tradeoff is favorable. Gaussian likelihood, nonlinear weak-GMM, and weak-IV calculations show when this happens, while simulations and a microenterprise application illustrate the method.
\end{abstract}

\medskip
\noindent\textbf{Keywords:} Quasi-Bayesian inference; Generalized Bayes; Hierarchical shrinkage; GMM; Weak identification.\\
\textbf{JEL codes:} C11, C14, C18, C23, C26, C51.

\section{Introduction}
\label{sec:intro}

Applied economists increasingly draw evidence from collections of related settings rather than from a single data set: for example, a treatment may be evaluated in several countries or experimental sites, or a structural parameter may be estimated separately across markets, firms, villages, or policy environments. In such cases, the empirical question is not only whether an effect is present in one sample; it is also how effects or structural parameters vary across contexts, and how evidence from one context should inform estimates in another. This is the external-validity problem: when are studies similar enough that information should be shared, and how should that sharing be done?\footnote{See, for example, \citet{Angrist2010} on the credibility revolution in internal validity, and \citet{Meager2019,Vivalt2020,Slough2022,Egami2023} for recent work on external-validity syntheses.  Recent Bayesian work on prior-study information with uncertain external validity includes \citet{FinanPouzo2026,You2026}.}

Bayesian Hierarchical Models (BHMs) provide a natural answer when each study can be described by a likelihood. To use them, a researcher specifies a sampling model within each group and an upper-level distribution for the group-specific parameters, and the posterior then combines the direct evidence from each group with information estimated from the collection of groups. This produces shrinkage: noisy group-level estimates are pulled toward values that are more plausible in light of the broader evidence. Hierarchical models have therefore become a standard tool for combining related treatment-effect estimates, especially when the within-study estimates are simple and approximately Gaussian; see, for example, \citet{GelmanHill2007} and \citet{Meager2019}.

Many econometric applications, however, are not naturally likelihood-based. Researchers often estimate the parameter of interest using instrumental variables, GMM, minimum distance, simulated moments, nonlinear least squares, or a structural objective chosen for its identifying content rather than for its interpretation as a full sampling model. In these settings, imposing a parametric likelihood only to obtain hierarchical pooling can be unattractive. It may add distributional assumptions that are not part of the original analysis, and it may move the within-group estimand away from the parameter defined by the original econometric criterion.

This paper develops \emph{Quasi-Bayesian Hierarchical Models} (QBHMs), a framework for hierarchical pooling when the within-group analysis is based on a GMM objective function, rather than a likelihood. Instead of imposing a likelihood, the construction starts from the objective function that the researcher would already use in each group. For group \(j\), let \(\phi_j\) denote the observation-level moment, residual, score, or minimum-distance map chosen by the researcher, and let \(q_{2,j,n_j}(x_j,\alpha_j)=-\|G_{j,n_j}(x_j,\alpha_j)\|_{W_{j,n_j}}^2/2\), with \(G_{j,n_j}(x_j,\alpha_j)=n_j^{-1/2}\sum_i\phi_j(X_{ji},\alpha_j)\), be the objective function the researcher would use if the group were analyzed on its own.  Let \(q_{1,j}(\alpha_j,\theta;z_j)\) describe the pooling relation between the economically comparable component\footnote{Comparability means that the researcher has selected component functions or transformations \(r_j(\alpha_j)\) that have the same substantive interpretation across groups.} of \(\alpha_j\) and a hierarchical parameter \(\theta\); that is, \(q_{1,j}\) assigns higher values to configurations in which the selected group component is more compatible with the hierarchical value \(\theta\).  A QBHM forms the quasi-posterior, the normalized exponential weighting of this composite objective,
\[
\Pi^{\mathrm{full}}_{n,\lambda}(d\alpha\,d\theta\mid x,z)
\propto
\exp\!\left[
\sum_{j=1}^J q_{2,j,n_j}(x_j,\alpha_j)
+
\lambda\sum_{j=1}^J q_{1,j}(\alpha_j,\theta;z_j)
\right]
\pi_\theta(d\theta)\prod_{j=1}^J \pi_j(d\alpha_j),
\]
where $\propto$ denotes that the left-hand and right-hand sides differ by a constant which does not depend on $\alpha$ or $\theta$. The scalar \(\lambda\ge0\) is the pooling strength.\footnote{We normalize the lower-level quasi-Bayes temperature, meaning the scalar multiplier on the lower-level objective, to one.  Fixed temperatures or learning rates are common in generalized-Bayes updates and can matter for quasi-posterior dispersion, finite-sample robustness, and calibration \citep{HolmesWalker2017,LyddonHolmesWalker2019}.  Under the strong-identification Laplace approximation, however, a fixed temperature does not change the first-order location of the Laplace-type estimator (LTE), and in the present normalization the substantively reported scale parameter is \(\lambda\), the relative weight on pooling.}  Throughout the paper, the \textit{hierarchy} refers to this upper-level pooling relation: the combination of \(q_1\), \(\pi_\theta\), and \(\lambda\) that specifies which study-level components should be similar and how strongly they are pulled together. It is a researcher-specified pooling device, not a requirement that the studies be generated from a random-effects population. At \(\lambda=0\), the marginal quasi-posteriors for each component $\alpha_j$ are the same as would be obtained by separate Laplace-Type Estimation \citep{ChernozhukovHong2003}. Furthermore, when the $q_{2,j,n_j}$ functions are log likelihoods, \(q_{1,j}\) is a conditional log density, and \(\lambda=1\), this replicates a classical Bayesian Hierarchical Model. The lower level therefore remains an econometric objective function, while the upper level supplies hierarchical shrinkage. Our analysis in this paper considers the properties of the quasi-posterior mean, which we will denote the \textit{QBHM estimator}; we treat its components as point estimators of the study-level estimands $\alpha_1, ..., \alpha_J$. Our theory considers the case where $J$ is a fixed number (rather than growing with $n$), as the number of studies is typically very small compared to the sample size. One consequence of this is that we do not consider estimation of $\theta$ (such as the so-called `average Average Treatment Effect'), as this is not generally identified unless $J \rightarrow \infty$.

The paper makes four main contributions. First, we introduce and highlight the flexibility of the QBHM framework. The group objective functions, instrument sets, controls, fixed effects, auxiliary statistics, weights, and dimensions may differ, as long as the pooled components are economically comparable. As such, one study can be estimated by IV with its own instruments, another by a fixed-effects residual objective function, another by a score objective function from a quasi-likelihood, and another by minimum distance on auxiliary moments. The hierarchy only links the components the researcher chooses to pool; these may be site-level treatment effects, IV returns to capital, production elasticities, switching thresholds, or other estimands with the same economic interpretation across studies.

Second, we give fixed-\(J\) theory under grouped-GMM conditions that are directly checkable in smooth applications. We formalize strong and weak identification by using differentiability in quadratic mean (DQM) local paths, defined below as \(O(n_j^{-1/2})\)-local perturbations of the group data law, following the weak-identification perspective in \citet{kaji2021weakidentification} and the weak-GMM limit experiment in \citet{AM2022}. Strong studies have the same asymptotic distribution as their unpooled GMM or LTE analogues, as in a standard IV design with a first stage bounded away from zero. Weak studies, such as IV designs with first stages local to zero or nonlinear GMM problems with weak curvature, remain in a weak-GMM limit experiment, meaning a limit in which the whole sample-moment objective remains a random function over the weak parameter space. In this limit, the hierarchy changes the \(\lambda\)-indexed prior over weak values, which we call the weak-limit prior path, and therefore remains in the asymptotic distribution. The appendix gives the corresponding profiled treatment of studies containing both strongly and weakly identified parameters.

Third, we characterize when our procedure can have lower asymptotic MSE, compared to estimating studies individually by Laplace-Type methods or GMM. Intuitively, pooling creates a bias--variance tradeoff, which can lead to lower MSE when the variance reduction dominates the increase in bias, in the same broad shrinkage-risk spirit as \citet{JamesStein1961}. A local nonlinear weak-GMM result and a scalar weak-IV example make this condition explicit in two salient cases.

Finally, we give practical guidance for reporting the pooling strength $\lambda$ and for statistical inference. We recommend reporting sensitivity paths over \(\lambda\), rather than demonstrating results over only one value of $\lambda$, though a form of cross validation (CV) provides a heuristic for a good value. Drawing analogies to \citet{AM2016conditional}, we show how one might construct weak-identification-robust confidence sets for hypothesis testing, noting that these do not agree with using the quasi-posterior itself.

\paragraph{Existing Literature.} The paper is related to several strands of work. First, generalized Bayes and moment-based Bayesian procedures replace a full likelihood with a loss, moment condition, or empirical-likelihood object in the update \citep{BissiriHolmesWalker2016,ChernozhukovHong2003,Schennach2005BayesianETEL,ChibShinSimoni2018MomentConditions}.  Recent work develops related quasi-Bayesian procedures for moment misspecification and conditional moment restriction settings \citep{ChernozhukovHansenKongWang2025PlausibleGMM,Kankanala2025GeneralizedBayesCMR}.  The present paper differs because we consider \(J\) to be fixed, and moreover study the shrinkage induced by the hierarchy, rather than misspecification robustness, frequentist coverage, or contraction for a single quasi-posterior. The closest conceptual comparison to our work is the quasi-Bayesian grouped-panel framework of \citet{Huang2023GroupedPanels}, which also combines flexible loss functions and priors for grouped data, but the object of interest is different. In our setting, the groups are observed econometric problems rather than latent panel clusters, the lower-level objective functions may differ across groups, and the theory focuses on fixed-\(J\) hierarchical shrinkage and weak-GMM MSE rather than posterior contraction for latent group assignments.

Hierarchical Bayes, empirical Bayes, and shrinkage show why pooling can reduce quadratic risk in normal means, multi-study, and labor-economics settings \citep{JamesStein1961,Meager2019,Walters2024EmpiricalBayes}. QBHM differs because the lower level need not be a common likelihood, a common normal approximation to a reduced-form summary, or a many-groups population-distribution problem.  Penalized or regularized GMM is also related to our work, since a QBHM mode maximizes the original grouped objective function plus an induced penalty.  The object analyzed in this paper, however, is the quasi-posterior mean \(\tilde\alpha^{\mathrm{full}}_{n,\lambda}\). The mode and mean are approximately the same in the strong identification case, by the generalized Bernstein von-Mises theorem in \citet{ChernozhukovHong2003}. However, this equivalence need not hold under weak identification, as shown in \citet{AM2022}, since the posterior is no longer asymptotically Gaussian; in general, it can be skewed, with the mean no longer equal to the mode. Relatedly, \citet{AndrewsMikusheva2023Inadmissible} show that GMM estimators can be inadmissible under weak identification, while quasi-Bayes posterior means and bootstrap-aggregated (bagged) GMM have stronger continuity properties. QBHM adds a cross-group hierarchy to this weak-GMM decision environment.

Finally, weak-instrument, semiparametric weak-identification, and weak-GMM work show that conventional first-order approximations can fail when the objective function is weakly curved or has only finite local drift \citep{StaigerStock1997,StockYogo2005,kaji2021weakidentification}.  Kaji's formulation explains weak identification through local paths that do not regularize the estimand at the root-\(n\) scale, and the Andrews--Mikusheva weak-GMM experiment supplies the decision-theoretic environment that QBHM uses below for point-estimation MSE.  For inference, the conditional-testing approach builds on the weak-instrument and functional-nuisance conditioning literatures \citep{Moreira2003,AndrewsMoreiraStock2006,AM2016conditional}. As in that literature, quasi-posterior intervals summarize the shrinkage estimator, while weak-identification-robust confidence sets should still be constructed by conditioning.

The paper proceeds as follows. Section~\ref{sec:framework} gives the normal--normal example that introduces the MSE tradeoff and the strong-versus-weak information scale. Section~\ref{sec:main-fixedJ-theory} defines the general grouped QBHM and gives the fixed-\(J\) theory, including the strong-study approximation, the weak-study weak-GMM reduction, and the Bayes-rule interpretation under squared loss of the rule induced by the hierarchy. Section~\ref{sec:mse-improvement} states the pointwise MSE comparison, a local sufficient condition, and a scalar weak-IV example. Section~\ref{sec:lambda-inference} discusses sensitivity paths and optional CV-selected reference choices for \(\lambda\), together with weak-identification-robust confidence sets. Sections~\ref{sec:mc} and~\ref{sec:structural} present the Monte Carlo evidence and empirical illustration. Section~\ref{sec:conclusion} concludes.

\section{A Normal-Normal Example}
\label{sec:framework}
\label{sec:normal-summary-example}

To fix intuition, we begin with an exact Bayesian hierarchical model. The rest of the paper does not impose these distributional assumptions, but the normal--normal case is useful because it gives closed-form shrinkage coefficients and makes the role of identification strength transparent.  It is also close to models used in applied hierarchical analyses, and so the results here may be of independent interest.\footnote{For example, \citet{Meager2019} considers a model with normal distributions in the main stages, apart from the prior on the heterogeneity parameter.  If that heterogeneity parameter were fixed, the model would fall within the case considered here.}  Suppose that there are \(J\) observed groups, and that group \(j\) is summarized by a scalar estimate \(\Delta_j\) of a group-specific estimand \(\alpha_j\), such as a site-level treatment effect or IV slope. Specifying a Bayesian hierarchical model requires us to specify our beliefs over the data in three parts. In this case, the lower level is
\[
        \Delta_j\mid \alpha_j
        \sim
        N\!\left(\alpha_j,\frac{\sigma_{\Delta,j}^2}{\mathcal I_{j,n}}\right),
        \qquad
        \sigma_{\Delta,j}^2>0,\quad \mathcal I_{j,n}>0 .
\]
The hierarchical level is
\[
        \alpha_j\mid \theta
        \sim
        N(\theta,\tau_0^2),
        \qquad \tau_0^2>0 ,
\]
and the prior level is
\[
        \theta\sim N(g_0,s_0^2),
        \qquad s_0^2>0 .
\]
Here \(\theta\) is the common mean toward which the group-specific estimands are pooled, \(\tau_0^2\) controls prior cross-group dispersion around that mean, and \(\mathcal I_{j,n}/\sigma_{\Delta,j}^2\) is the information in the group-\(j\) summary.  We use \(\mathcal I_{j,n}\) to model identification strength in this setup: strong identification corresponds to \(\mathcal I_{j,n}=n_j\to\infty\), so the group summary becomes increasingly precise, while weak identification corresponds to \(\mathcal I_{j,n}\) remaining bounded, so the effective information in \(\Delta_j\) does not grow with \(n_j\).

The joint quasi-posterior is proportional to
\[
\exp\left[
-\frac12\sum_{j=1}^J\frac{\mathcal I_{j,n}}{\sigma_{\Delta,j}^2}(\Delta_j-\alpha_j)^2
-\frac{\lambda}{2}\sum_{j=1}^J\frac{(\alpha_j-\theta)^2}{\tau_0^2}
-\frac12\frac{(\theta-g_0)^2}{s_0^2}
\right].
\]
When \(\lambda=1\), this is the ordinary posterior under the normal hierarchical model above; when \(\lambda=0\), the groups are estimated separately; and, more generally, \(\lambda/\tau_0^2\) is the effective pooling precision.  Increasing \(\lambda\), or decreasing \(\tau_0^2\), strengthens the pull of each \(\alpha_j\) toward the common mean. The same display can also be read as a quadratic QBHM objective.  The lower normal density contributes
\[
        q_{2,j,n_j}^{\mathrm{N}}(\Delta_j,\alpha_j)
        =
        -\frac12\frac{\mathcal I_{j,n}}{\sigma_{\Delta,j}^2}
        (\Delta_j-\alpha_j)^2,
\]
the hierarchy contributes
\[
        q_{1,j}^{\mathrm{N}}(\alpha_j,\theta)
        =
        -\frac12\frac{(\alpha_j-\theta)^2}{\tau_0^2},
\]
and the prior contributes the final quadratic term in \(\theta\).  The exponent in the quasi-posterior is therefore the composite objective
\[
        \sum_{j=1}^J q_{2,j,n_j}^{\mathrm{N}}(\Delta_j,\alpha_j)
        +
        \lambda\sum_{j=1}^J q_{1,j}^{\mathrm{N}}(\alpha_j,\theta)
        -
        \frac12\frac{(\theta-g_0)^2}{s_0^2},
\]
Setting \(\lambda=1\) gives the ordinary normal hierarchical posterior for this benchmark specification. The convenient functional forms yield explicit expressions for the marginal quasi-posterior mean of $\alpha_j$, denoted $\tilde\alpha_{j,\lambda}$, which is the \textit{QBHM estimator} that we study in this paper.

\begin{proposition}
\label{prop:normal-normal-summary}
In the Gaussian case described above, fix \(J<\infty\), \(g_0\in\mathbb R\), \(s_0^2>0\), \(\tau_0^2>0\), and, for each \(k=1,\ldots,J\), \(\sigma_{\Delta,k}^2>0\) and \(\mathcal I_{k,n}>0\).  For every \(\lambda\ge0\),
\begin{equation}
\label{eq:normal-normal-summary-alpha}
\tilde\alpha_{j,\lambda}
=
c_{j,n}(\lambda)\Delta_j+
\bigl[1-c_{j,n}(\lambda)\bigr]\tilde\theta_\lambda,
\qquad
c_{j,n}(\lambda)
=
\left(1+\frac{\lambda \sigma_{\Delta,j}^2}{\mathcal I_{j,n}\tau_0^2}\right)^{-1},
\end{equation}
where \(\tilde\theta_\lambda\) is the marginal quasi-posterior mean of \(\theta\) and satisfies
\begin{equation}
\label{eq:normal-normal-summary-theta}
\tilde\theta_\lambda
=
\frac{
s_0^{-2}g_0+(\lambda/\tau_0^2)\sum_{k=1}^J c_{k,n}(\lambda)\Delta_k
}{
s_0^{-2}+(\lambda/\tau_0^2)\sum_{k=1}^J c_{k,n}(\lambda)
}.
\end{equation}
For \(\lambda>0\), the diffuse-prior limit \(s_0^{-2}\downarrow0\) is
\begin{equation}
\label{eq:normal-normal-summary-diffuse}
\tilde\theta_\lambda^{\mathrm{diff}}
=
\sum_{k=1}^J w_{k,n}(\lambda)\Delta_k,
\qquad
w_{k,n}(\lambda)
=
\frac{c_{k,n}(\lambda)}{\sum_{\ell=1}^J c_{\ell,n}(\lambda)}.
\end{equation}
\end{proposition}

In this case, we see that the QBHM estimates are a weighted average of individual estimates and a term which combines information from all of the groups. Note that when \(\lambda=0\) we have no pooling, and consequently \(\tilde\alpha_{j,0}=\Delta_j\).  For fixed \(\lambda>0\), a group with weaker identification, meaning a larger sampling variance \(\sigma_{\Delta,j}^2/\mathcal I_{j,n}\), has a smaller \(c_{j,n}(\lambda)\) and so receives more shrinkage toward \(\tilde\theta_\lambda\).  Conversely, when the group summary is precise, \(c_{j,n}(\lambda)\) is close to one and the pooled estimate remains close to \(\Delta_j\).

Proposition~\ref{prop:normal-normal-summary} contains two sources of shrinkage: each \(\tilde\alpha_{j,\lambda}\) is shrunk toward the estimated common mean \(\tilde\theta_\lambda\), and \(\tilde\theta_\lambda\) itself is a compromise between the prior mean \(g_0\) and the group summaries.  When \(s_0^2\) is small, the prior pulls the common mean toward \(g_0\), but when \(s_0^2\) is large, this direct prior pull is weak. As the prior gets wider and wider, as in \eqref{eq:normal-normal-summary-diffuse}, the common mean instead converges to an empirical average of the group summaries, with weights determined by the shrinkage coefficients. Weakly identified groups are therefore pulled more toward the common mean, but they also receive less weight in forming that common mean. We can see from this that Bayesian Hierarchical models naturally provide an identification strength-adaptive way to pool data together from multiple studies.

One common concern when using these methods is the influence that the prior distribution has over the results. As discussed, we treat the mean of the quasi-posterior as a point estimator, and analyse its frequentist properties, rather than consider QBHMs as Bayesian objects. The standard solution to the choice of priors is twofold. First, as shown in our main results, prior choice is asymptotically negligible for the estimator, when working under strong identification. The prior does matter in small samples and under weak identification - as highlighted above - but this can be viewed as strength of the method, as it allows the incorporation of extra information. This is especially important in the case of weakly identified parameters, where the data themselves do not pin down the estimand consistently. Second, a researcher who wishes to de-emphasize the effect of priors can set them to be very wide. This is commonly done in applied work (e.g. by \citet{Meager2019}) and, as seen from diffuse-prior limit above, makes the shrinkage effect of the priors become asymptotically irrelevant. All of these features carry through to the QBHM case. As such, our method provides the flexibility to use prior information if needed, but with no real obligation to do so.

The next result records the consequences of identification strength. It shows that fixed pooling has no effect on the asymptotic distribution when group \(j\) is strongly identified, but it does have an effect when the identification is weak.

\begin{corollary}
\label{cor:normal-normal-scale}
In the setting of Proposition~\ref{prop:normal-normal-summary}, fix group \(j\) and a compact set \(\Lambda\subset[0,\bar\lambda]\), with \(\bar\lambda<\infty\).  Suppose \(J\), \(g_0\), \(s_0^2\), \(\tau_0^2\), and \(\sigma_{\Delta,k}^2\), \(k=1,\ldots,J\), are fixed in \(n\), and \(\max_{1\le k\le J}|\Delta_k|=O_p(1)\).  If \(\mathcal I_{j,n}=n_j\to\infty\), then \(\sup_{\lambda\in\Lambda}\sqrt{n_j}\,|\tilde\alpha_{j,\lambda}-\Delta_j|=o_p(1)\).  If instead \(\mathcal I_{j,n}\) is bounded above along a subsequence, then for every fixed \(\lambda>0\) the shrinkage coefficient \(1-c_{j,n}(\lambda)\) is bounded away from zero along that subsequence, so fixed-\(\lambda\) pooling remains in the leading finite-information problem.
\end{corollary}

Note from \eqref{eq:normal-normal-summary-diffuse} that, in the diffuse-prior limit and for \(\lambda>0\), the shrinkage target is not necessarily the true common mean.  Instead, it is the weighted empirical mean \(\tilde\theta_\lambda^{\mathrm{diff}}\) of the group summaries \(\Delta_k\).  Thus the pooled estimate is pulled toward an empirical center that includes group \(j\)'s own summary, rather than toward an oracle or leave-one-out target.  Pooling is therefore most likely to improve MSE when the unpooled summary is noisy and this empirical center is well aligned with the group-specific estimand.  On the other hand, it is likely to harm MSE when the summary is already precise or the center is poorly aligned with the group.

The next two results characterize when pooling can improve a form of mean squared error.  The first fixes \(\theta\) in a simplified known-center case to isolate the scalar signal-to-noise tradeoff behind the shrinkage coefficient and to evaluate integrated MSE.  That is, it asks whether pooling can lower mean squared error before conditioning on the realized group-specific effect.\footnote{This can be viewed as ex ante MSE, or equivalently as average MSE under the hierarchical distribution.  It differs from the pointwise MSE below, which conditions on the realized value of \(\alpha_{j0}\) and averages only over the sampling error.}  For this proposition only, suppose that, conditional on known \(\theta\), \(\alpha_j=\theta+u_j\) with \(u_j\sim N(0,\tau^2)\), and that \(\Delta_j=\alpha_j+\varepsilon_j\) with \(\varepsilon_j\sim N(0,\sigma_{\Delta,j}^2/\mathcal I_{j,n})\) independent of \(u_j\); write \(v_j=\sigma_{\Delta,j}^2/\mathcal I_{j,n}\).

\begin{proposition}
\label{prop:normal-normal-risk}
Under the known-center Gaussian example above, consider the linear rule \(\delta_j(\eta_j)=\eta_j\Delta_j+(1-\eta_j)\theta\), with \(\eta_j\in[0,1]\).  The integrated MSE for estimating \(\alpha_j\), averaging over both \(u_j\) and \(\varepsilon_j\), is \(R_j(\eta_j)=(1-\eta_j)^2\tau^2+\eta_j^2v_j\).  This risk is uniquely minimized at \(\eta_j^\star=\tau^2/(\tau^2+v_j)\).  For \(\lambda>0\), the coefficient \(c_{j,n}(\lambda)\) equals this risk-minimizing weight when \(\lambda/\tau_0^2=1/\tau^2\).
\end{proposition}

The integrated-risk-optimal level of pooling has an intuitive interpretation.  The parameter \(\tau^2\) measures how dispersed the true group effects are around the center, while \(v_j\) measures how noisy the group summary is.  When \(v_j\) is small relative to \(\tau^2\), the summary \(\Delta_j\) is precise, the center is weakly predictive of \(\alpha_j\), or both, so \(\eta_j^\star\) is close to one and the integrated-risk-optimal rule puts little weight on the center.  When \(v_j\) is large relative to \(\tau^2\), the summary is noisy, the group effects are expected to lie close to the center, or both, so \(\eta_j^\star\) is smaller and the optimal rule pools more aggressively.

The last statement of Proposition~\ref{prop:normal-normal-risk} calibrates $c_{j,n}(\lambda)$ against this integrated-risk benchmark.  Since \(c_{j,n}(\lambda)=1/(1+\lambda v_j/\tau_0^2)\), the QBHM puts less weight on the unpooled summary when \(v_j\) is large, when \(\lambda\) is large, or when \(\tau_0^2\) is small.  Weak information, stronger imposed pooling, and tighter hierarchical dispersion all push in the same direction.  The coefficient is integrated-risk optimal in this case only when the effective pooling precision \(\lambda/\tau_0^2\) matches the random-effects precision \(1/\tau^2\).  If the imposed pooling precision is larger than \(1/\tau^2\), the rule shrinks more than the optimum; if it is smaller, the rule shrinks less.

In contrast with the last result, the next corollary conditions on the realized group-specific estimand and asks when shrinkage toward a fixed center improves on the unpooled summary.  This distinction matters because pooling can be attractive on average under a well-centered random-effects distribution, while still increasing MSE for a group whose realized estimand is far from the center.

\begin{corollary}
\label{cor:normal-pointwise-mse}
Suppose \(\Delta_j=\alpha_{j0}+\varepsilon_j\), with \(\mathbb E\varepsilon_j=0\) and \(\text{Var}(\varepsilon_j)=v_j:=\sigma_{\Delta,j}^2/\mathcal I_{j,n}\in(0,\infty)\).  For a fixed center \(\theta\) and a linear shrinkage rule \(\delta_j(\eta)=\eta\Delta_j+(1-\eta)\theta\), with \(\eta\in[0,1]\), the pointwise MSE for estimating the fixed estimand \(\alpha_{j0}\) is \(\MSE(\delta_j(\eta),\alpha_{j0})=\eta^2v_j+(1-\eta)^2(\theta-\alpha_{j0})^2\).  Relative to the unpooled rule \(\Delta_j\), the shrinkage rule has lower pointwise MSE if and only if \((1-\eta)^2(\theta-\alpha_{j0})^2<(1-\eta^2)v_j\).  For \(\eta<1\), this is equivalently \((\theta-\alpha_{j0})^2<[(1+\eta)/(1-\eta)]v_j\).
\end{corollary}

The pointwise condition separates the variance gain, \(v_j\), from the bias cost, \((\theta-\alpha_{j0})^2\), the squared distance between the shrinkage target and the realized estimand.  Pooling is therefore most likely to help when \(v_j\) is large and \(\theta\) is close to \(\alpha_{j0}\): the unpooled estimate is noisy, and the bias from moving toward the center is small.  Pooling is most likely to hurt when \(v_j\) is small or when \(\theta\) is far from \(\alpha_{j0}\): the unpooled estimate is already precise, or the shrinkage target is badly aligned.  Full shrinkage to the center corresponds to \(\eta=0\), in which case pooling improves pointwise MSE exactly when \((\theta-\alpha_{j0})^2<v_j\).  As \(\eta\) approaches one, the condition becomes easier to satisfy because the amount of shrinkage is small, but the possible MSE gain also becomes small.

Proposition~\ref{prop:normal-normal-risk} and Corollary~\ref{cor:normal-pointwise-mse} describe the same shrinkage tradeoff from two perspectives.  Proposition~\ref{prop:normal-normal-risk} treats the center as known and evaluates MSE after averaging over a normal random-effects distribution.  Corollary~\ref{cor:normal-pointwise-mse} fixes the realized estimand and evaluates MSE pointwise.  Together, they show why pooling can help in weak-information settings: it reduces variance when the unpooled summary is noisy, but it can introduce bias when the center is not well aligned with the group-specific estimand.

The full QBHM in Proposition~\ref{prop:normal-normal-summary} is complicated not only because the objective functions need not correspond to Gaussian log likelihoods, but also because the hierarchy may use a general prior and need not have a known center.  The next section states the grouped GMM QBHM directly, with separate first-order treatments of strong and weak study blocks.

\section{Main Results}
\label{sec:main-fixedJ-theory}

We now state the regularity conditions needed for our main theory. Section~\ref{sec:dqm-identification} defines strong and weak identification, Section~\ref{sec:strong-components} gives the result for strongly identified studies, and Section~\ref{sec:weak-components} gives the result for weakly identified studies; throughout this section, ``group'' and ``study'' are used interchangeably. For the main text, identification strength is the same for all components in a study: the reported vector for a study is either strongly identified or weakly identified. Appendix~\ref{app:mixed-within-study-blocks} gives the extension that allows mixed identification status within a study. Throughout, stochastic processes indexed by compact sets are taken to have Borel measurable versions in the corresponding sup-norm spaces.\footnote{This is a measurability convention for process-level convergence statements: if an empirical-process construction gives only asymptotic measurability, weak convergence and probability statements are understood in the usual outer-probability sense.} For a symmetric positive semidefinite matrix \(A\), write \(\|x\|_A^2:=x'Ax\).

\subsection{Setup}
\label{sec:formal-grouped-problem}

Consider \(J\ge1\) independent groups, where group \(j\) has data \(x_j=(x_{j1},\ldots,x_{j n_j})\), fixed group-level covariates \(z_j\) when such covariates enter the pooling relation, and parameter \(\alpha_j\in\calA_j\subset\R^{d_j}\).  Stack \(\alpha=(\alpha_1',\ldots,\alpha_J')'\in\calA:=\prod_{j=1}^J\calA_j\), allowing the dimensions \(d_j\) to differ across groups.  The pooling relation may therefore be applied only to economically comparable components or transformations of them, denoted \(r_j(\alpha_j)\), leaving study-specific nuisance components unpooled; for example, \(r_j\) may select a treatment effect, an IV slope, a production elasticity, or a switching threshold from a larger study-specific parameter vector.  The hierarchical parameter \(\theta\in\Theta\subset\R^{d_\theta}\) indexes the function that creates pooling.  Let \(n_j\) denote the sample size of group \(j\).

At any law \(P_j\) at which the group-\(j\) estimand is identified, write \(\alpha_j(P_j)\) for the estimand.\footnote{The theory also allows pseudo-true estimands.  For example, under fixed-weight GMM, \(\alpha_j(P_j)\) may be any maximizer of the population objective \(a\mapsto-\|\mathbb E_{P_j}\phi_j(X_{ji},a)\|_{W_j}^2/2\).} The moment condition defining that estimand is
\begin{equation}
\label{eq:unconditional-moment-restriction}
        \mathbb{E}_{P_j}\phi_j(X_{ji},\alpha_j(P_j))=0,
\end{equation}
where \(\phi_j:\mathcal X_j\times\calA_j\to\R^{k_j}\) is chosen by the researcher. When this restriction has multiple solutions, the notation refers to the solution, selection, or pseudo-true value specified by the study-level criterion. It is important that we explicitly consider a map $P_j \mapsto \alpha_j(P_j)$, rather than only look at the moment condition under one `true' data distribution, as this is needed to study weak identification. We follow the formalization in \citet{kaji2021weakidentification} below, which nests several common definitions. The data need not be generated from a random-effects model; while a plausible random-effects description may help guide the choice of the pooling term, the study-level estimand is defined through the moment condition rather than through a random-effects law.  Let \(\hat g_j(a):=n_j^{-1}\sum_i\phi_j(X_{ji},a)\) and \(G_{j,n_j}(x_j,a):=\sqrt{n_j}\,\hat g_j(a)\).  With positive definite weight \(W_{j,n_j}\), the within-group objective function is
\begin{equation}
\label{eq:gmm-objective}
        q_{2,j,n_j}(x_j,a)
        :=
        -\frac12
        \|G_{j,n_j}(x_j,a)\|_{W_{j,n_j}}^2
        =
        -\frac{n_j}{2}\hat g_j(a)'W_{j,n_j}\hat g_j(a).
\end{equation}
This display is the fixed-weight GMM objective in maximization form.  With continuously updated GMM, the weight depends on the candidate value \(a\); for example, \(W_{j,n_j}(a)=\widehat\Sigma_{j,n_j}(a,a)^{-1}\), where \(\widehat\Sigma_{j,n_j}(a,\tilde a)\) estimates the covariance between \(G_{j,n_j}(X_j,a)\) and \(G_{j,n_j}(X_j,\tilde a)\).  The weak-study subsection uses a generic estimated weight process for the lower-level objective, though in the continuously updated case, that process is the diagonal inverse-covariance weight.  Write \(\bar q_{2,j,n_j}=n_j^{-1}q_{2,j,n_j}\), and let \(M_j\) denote the population criterion that is the uniform limit of \(\bar q_{2,j,n_j}\) for strongly identified studies.  In fixed-weight GMM, if \(\mu_j(a):=\mathbb E_{P_{j0}}\phi_j(X_{ji},a)\) and \(W_{j,n_j}\pto W_j\), then
\[
        \bar q_{2,j,n_j}(x_j,a)
        =
        -\frac12\hat g_j(a)'W_{j,n_j}\hat g_j(a),
        \qquad
        M_j(a)=-\frac12\|\mu_j(a)\|_{W_j}^2 .
\]
As discussed in earlier sections, the pooling term is a second objective function \(q_{1,j}(\alpha_j,\theta;z_j)\).  In an ordinary likelihood-based hierarchy, it would be the log density of \(\alpha_j\) given \((\theta,z_j)\). In a QBHM, however, it may instead be any continuous pooling objective, such as a smooth penalty that rewards similarity in \(r_j(\alpha_j)\).  For example, a Gaussian common-mean term for comparable components is
\[
        q_{1,j}(\alpha_j,\theta)
        =
        -\frac12(r_j(\alpha_j)-\theta)'\Omega_{0j}^{-1}(r_j(\alpha_j)-\theta),
\]
up to constants.  Set
\[
        q_{2,n}(\alpha;x):=\sum_{j=1}^J q_{2,j,n_j}(x_j,\alpha_j),
        \qquad
        q_1(\alpha,\theta;z):=\sum_{j=1}^J q_{1,j}(\alpha_j,\theta;z_j),
\]
with \(z=(z_1,\ldots,z_J)\).  The observations \(X_{ji}\) enter the group objective; the fixed covariates \(z_j\) enter the pooling relation.  For a pooling-strength parameter \(\lambda\ge0\), define
\begin{equation}
\label{eq:composite-objective}
L_{n,\lambda}(\alpha,\theta;x,z)
:=
q_{2,n}(\alpha;x)
+
\lambda q_1(\alpha,\theta;z),
\end{equation}
and let \(\pi(d\alpha\,d\theta)=\pi_\theta(d\theta)\prod_{j=1}^J\pi_j(d\alpha_j)\) be a proper product prior.  The QBHM quasi-posterior is
\begin{equation}
\label{eq:qbhm-posterior}
\begin{aligned}
\Pi^{\mathrm{full}}_{n,\lambda}(d\alpha\,d\theta\mid x,z)
&=
\frac{
\exp\!\left[L_{n,\lambda}(\alpha,\theta;x,z)\right]
\pi_\theta(d\theta)\prod_{j=1}^J\pi_j(d\alpha_j)
}
{Z_{n,\lambda}},\\
Z_{n,\lambda}
&:=
\int_{\Theta}\int_{\calA}
\exp\!\left[L_{n,\lambda}(a,\vartheta;x,z)\right]
\pi_\theta(d\vartheta)\prod_{j=1}^J\pi_j(da_j).
\end{aligned}
\end{equation}
Recall that if \(q_{2,j,n_j}\) is a group log likelihood, \(q_{1,j}\) is a conditional log density, and \(\lambda=1\), this display is ordinary hierarchical Bayes.  Otherwise, the lower level remains the econometric objective function in \eqref{eq:gmm-objective}.  At \(\lambda=0\), the marginal law of \(\alpha\) is the product of the unpooled quasi-posteriors,
\[
\begin{aligned}
\Pi^{\mathrm{full}}_{n,0}(d\alpha\,d\theta\mid x,z)
&=
\pi_\theta(d\theta)\prod_{j=1}^J\Pi^{\mathrm{un}}_{j,n}(d\alpha_j\mid x_j),\\
\Pi^{\mathrm{un}}_{j,n}(d\alpha_j\mid x_j)
&\propto
\exp[q_{2,j,n_j}(x_j,\alpha_j)]\pi_j(d\alpha_j).
\end{aligned}
\]
The parameter \(\theta\) then has no effect on the marginal estimates of the study-level estimands.  Under squared loss, the \textit{QBHM estimator} is the marginal quasi-posterior mean
\begin{equation}
\label{eq:qpost-mean}
\tilde\alpha^{\mathrm{full}}_{n,\lambda}
:=
\int \alpha\,\Pi^{\mathrm{full}}_{n,\lambda}(d\alpha\,d\theta\mid x,z),
\qquad
\tilde\alpha^{\mathrm{full}}_{j,n,\lambda}
:=
\int \alpha_j\,\Pi^{\mathrm{full}}_{n,\lambda}(d\alpha\,d\theta\mid x,z).
\end{equation}
Throughout the paper, a tilde denotes a marginal quasi-posterior mean and a circumflex denotes an objective-function maximizer.  The choice of \(\lambda\) determines the position on the bias--variance tradeoff.  Section~\ref{sec:lambda-inference} therefore recommends reporting estimates over a sensitivity path in \(\lambda\), rather than only at one value.

The first formal condition fixes the sampling environment, compact supports, continuity needed for the quasi-posterior, and base-prior regularity.

\begin{assumption}\label{ass:primitive-setup}
The following conditions hold.
\begin{enumerate}[label=(\roman*),leftmargin=2.1em]
\item The number of groups \(J\) is fixed and \(n_j\to\infty\) for every \(j\). For each \(n_j\), the \(n_j\) observations in group \(j\) are independent and identically distributed under the row law for that sample size; this row law may depend on \(n_j\) and on the local direction. The group samples \((X_{j1},\ldots,X_{jn_j})\), \(j=1,\ldots,J\), are independent across groups.
\item Each \(\calA_j\) and \(\Theta\) is compact, and \(\Lambda\subset[0,\bar\lambda]\) is compact.
\item With probability approaching one under the triangular-array law under consideration, \(a\mapsto q_{2,j,n_j}(x_j,a)\) is finite and continuous on \(\calA_j\) for each group. The pooling objective \(q_1(\alpha,\theta;z)\) is finite and continuous on \(\calA\times\Theta\).
\item The priors \(\pi_j\) and \(\pi_\theta\) are proper Borel priors with positive continuous densities on the relevant compact sets, relative to fixed dominating measures.
\end{enumerate}
\end{assumption}

In applications, compactness can be interpreted as a sufficiently wide truncation chosen before the asymptotic approximation is applied.  Assumption~\ref{ass:primitive-setup} imposes independent sampling units within groups and independent samples across groups, but it does not impose a common likelihood, a common moment dimension, a common instrument set, or common supports across groups.  A collection of studies with different instruments, controls, fixed effects, or auxiliary statistics can therefore be pooled when the components linked by \(q_1\) have a common interpretation. The notation is abstract, but the requirement is the standard GMM one: each study supplies a moment objective for its own estimand, and the hierarchy only adds a cross-study relation for selected comparable components. The compactness, continuity, and density conditions are used to make the quasi-posterior proper and to give uniform bounds in the reductions below.

\subsection{Defining strong and weak identification}
\label{sec:dqm-identification}

This section defines what it means for a reported study block to be strongly or weakly identified. Recall that the data law in group \(j\) is denoted \(P_{j0}\). In the strongly identified and point-identified case, we write \(\alpha_{j0}:=\alpha_j(P_{j0})\), so the estimand is well defined at this fixed law. Weak identification requires a slightly different formalization. We treat \(P_{j0}\) as a reference law at which the estimand may be only set identified, and study sequences of nearby laws \(P_{j,n_j,f_j}\) that converge to \(P_{j0}\) at the \(O(n_j^{-1/2})\) scale. Along these sequences, the estimand is point identified, so the local path selects a particular nearby value from the set of values consistent with \(P_{j0}\). Equivalently, when the sample size in group \(j\) is \(n_j\), we view the data as drawn from \(P_{j,n_j,f_j}\) rather than exactly from \(P_{j0}\). This formalization, originally due to \citet{kaji2021weakidentification}, covers several familiar forms of weak identification. For example in weak IV, the correlation between the instrument and the endogenous regressor converges to $0$ as the sample size increases, so the sequence \((P_{j,n_j,f_j})_{n_j}\) consists of laws under which the first stage becomes weak at the same rate as sampling noise\footnote{Rank-deficient GMM and local-to-singular moment problems fit the same description. When regular and weak coordinates appear inside the same study, Appendix~\ref{app:mixed-within-study-blocks} profiles the regular coordinates and leaves a reduced moment process with Gaussian variation and finite local drift.}. The direction \(f_j\) indexes which local departure from \(P_{j0}\) is being considered: in the weak-IV example, it determines the local first-stage strength, while in a rank-deficient GMM problem it determines the local drift that selects among otherwise observationally similar parameter values.

Since groups are independent, local sequences can be specified group by group: for \(f=(f_1,\ldots,f_J)\), \(P_{n,f}=\bigotimes_{j=1}^J P_{j,n_j,f_j}^{n_j}\). Thus local identification is defined within each group and then combined across groups by independence. The relevant notion of convergence is differentiability in quadratic mean, or DQM. After choosing a common dominating measure \(\mu_j\),\footnote{For example, Lebesgue measure or any measure dominating the local path.} let \(p_{j,n_j,f_j}\) and \(p_{j0}\) denote the densities of \(P_{j,n_j,f_j}\) and \(P_{j0}\).

\begin{definition}\label{def:dqm-path}
For group \(j\), a triangular-array local path \((P_{j,n_j,f_j})_{n_j}\) through \(P_{j0}\) is differentiable in quadratic mean with score \(f_j\in L_2(P_{j0})\) if \(\mathbb E_{P_{j0}}f_j=0\) and
\[
        \int
        \left[
        \sqrt{n_j}
        \left(\sqrt{p_{j,n_j,f_j}}-\sqrt{p_{j0}}\right)
        -\frac12 f_j\sqrt{p_{j0}}
        \right]^2
        d\mu_j
        \to0 .
\]
\end{definition}

Write \(L_2^0(P_{j0})=\{f\in L_2(P_{j0}):\mathbb E_{P_{j0}}f=0\}\). This formulation is useful because it separates the original structural model from the local weak experiment. Let \(\mathcal P_{j,\alpha}\) denote the part of the group-\(j\) model on which the reported estimand \(\alpha_j(P_j)\) is point identified. In a weak-identification sequence, the reference law \(P_{j0}\) may be an identification-failure law, so \(\alpha_j(P_{j0})\) need not be a single well-defined structural value. We therefore look at DQM local paths that approach \(P_{j0}\) while remaining in \(\mathcal P_{j,\alpha}\). A score \(f_j\) is called \emph{pertinent} for the reported estimand if at least one such DQM path induces score \(f_j\) and all such paths with the same score have the same limiting estimand value. Write \(\dot{\mathcal P}_{j0,\alpha}\) for the resulting pertinent tangent cone, and write
\[
        \alpha^W_{j0}(f_j)
        :=
        \lim_{n_j\to\infty}\alpha_j(P_{j,n_j,f_j})
\]
for the common limit when it exists\footnote{The cone language matters: multiplying the local score by a positive constant changes the speed along the same local ray but not the limiting weak value. Thus the weak estimand is allowed to depend on the direction of approach to the failure point, but not linearly on the magnitude of the perturbation.}. We can now define identification strength as follows:

\begin{definition}\label{def:dqm-identification}
A component \(b_j(P_j)\) is DQM-regular at \(P_{j0}\) if \(b_j(P_{j0})\) is well defined and there is a continuous linear map \(\dot b_j:L_2^0(P_{j0})\to\R^{d_b}\) such that, for every DQM path \((P_{j,n_j,f_j})_{n_j}\) through \(P_{j0}\) along which the component is well defined,
\[
        \sqrt{n_j}\bigl(b_j(P_{j,n_j,f_j})-b_j(P_{j0})\bigr)
        \to
        \dot b_j[f_j].
\]
We call such a component \textit{strongly identified}. A reported block \(\alpha_j(P_j)\) is \textit{weakly regular}, or \textit{weakly identified}, at \(P_{j0}\) if it is not DQM-regular at \(P_{j0}\) and there is a map \(\alpha^W_{j0}:\dot{\mathcal P}_{j0,\alpha}\to\R^{d_j}\) such that
\[
        \alpha_j(P_{j,n_j,f_j})\to \alpha^W_{j0}(f_j)
\]
for every pertinent DQM local path, with \(\alpha^W_{j0}\) continuous on the pertinent cone and homogeneous of degree zero: \(\alpha^W_{j0}(cf_j)=\alpha^W_{j0}(f_j)\) for all \(c>0\) whenever both scores are pertinent.
\end{definition}

Whether a reported estimand is strongly or weakly identified is a property of the map \(P\mapsto\alpha_j(P)\), not of a particular estimator. Strong identification means that the estimand has a first-order linear derivative with respect to local changes in the law. Weak identification means instead that the estimand has a well-defined limit along each pertinent local direction, but this limit is a nonlinear direction-only object. That direction-only dependence is the source of the non-Gaussian weak-GMM limit below.

For weak GMM, there is one additional reduction before the process limit is written down. We keep \(\calA_j\) for the researcher's original candidate set and define the reference-law population moment and zero set by
\[
        \mu_{j0}(a)=\mathbb E_{P_{j0}}\phi_j(X_{ji},a),
        \qquad
        \mathcal Z_{j0}=\{a\in\calA_j:\mu_{j0}(a)=0\}.
\]
The weak-GMM limit is indexed by a \emph{retained weak space}: a compact set \(\calW_j\subseteq\mathcal Z_{j0}\) that contains the weak limits \(\alpha^W_{j0}(f_j)\) for the local directions under consideration. The limit below is indexed by \(\calW_j\), not by all of \(\calA_j\). Values in \(\calA_j\setminus\mathcal Z_{j0}\) do not have a finite weak-process limit, since \(g_{j,n}(a)=\sqrt{n_j}\mu_{j0}(a)+O_p(1)\) under the reference law. They are therefore regularized, profiled out, or asymptotically discarded, rather than included in the weak experiment. If \(\mathcal Z_{j0}\) is lower-dimensional, the retained-space base measure introduced below is understood to be a fixed coordinate, Hausdorff, or other appropriate measure on \(\calW_j\).

\begin{assumption}\label{ass:study-level-split}
The index sets \(\mathcal J_S\) and \(\mathcal J_W\) form a fixed partition of \(\{1,\ldots,J\}\). For each \(j\in\mathcal J_S\), every reported component of the estimand map \(P_j\mapsto \alpha_j(P_j)\) is strongly identified at \(P_{j0}\). For each \(j\in\mathcal J_W\), the reported block is weakly regular at \(P_{j0}\), and the compact set \(\calW_j\subseteq\mathcal Z_{j0}\) used below is a retained weak space associated with that block.
\end{assumption}

When each reported estimand is scalar, this condition only requires classifying each study as strong or weak. The condition is imposed only for the main fixed-\(J\) theory; Appendix~\ref{app:mixed-within-study-blocks} gives the profiled extension for a study containing both strongly and weakly identified components.

\subsection{Strong studies}
\label{sec:strong-components}

We now give primitive conditions for the strong-study results. The following condition gives a locally quadratic population objective and enough sample regularity for the corresponding lower-level quasi-posterior to concentrate at the usual \(O(n_j^{-1/2})\) rate.

\begin{assumption}
\label{ass:primitive-strong}
For each \(j\in\mathcal J_S\), the following conditions hold. All stochastic convergence statements are under the triangular-array law used for group \(j\) in the asymptotic statement under consideration.
\begin{enumerate}[label=(\roman*),leftmargin=2.1em]
\item The point \(\alpha_{j0}\) is the unique interior maximizer of \(M_j\) on \(\calA_j\). The function \(M_j\) is continuous on \(\calA_j\) and twice continuously differentiable on an open neighborhood \(\mathcal N_j\) of \(\alpha_{j0}\).

\item Write \(H_j(a):=-\nabla_{\alpha\alpha}^2M_j(a)\) on \(\mathcal N_j\), and write \(H_j:=H_j(\alpha_{j0})\). The matrix \(H_j\) is positive definite.

\item The normalized sample objective converges uniformly to its population counterpart, and its first two derivatives converge locally near \(\alpha_{j0}\). Specifically,
\[
        \sup_{a\in\calA_j}
        |\bar q_{2,j,n_j}(X_j,a)-M_j(a)|\pto0 .
\]
There is \(\delta_j>0\) such that
\[
        \mathcal U_j(\delta_j)
        =
        \{a\in\calA_j:\|a-\alpha_{j0}\|\le\delta_j\}
\]
is contained in \(\mathcal N_j\), and, with probability approaching one, \(\bar q_{2,j,n_j}(X_j,\cdot)\) is twice continuously differentiable on a neighborhood of \(\mathcal U_j(\delta_j)\). For \(r=1,2\),
\[
        \sup_{a\in\mathcal U_j(\delta_j)}
        \left\|
        \nabla_\alpha^r\bar q_{2,j,n_j}(X_j,a)
        -
        \nabla_\alpha^r M_j(a)
        \right\|
        \pto0.
\]
The sample score at the population value satisfies
\[
        \left\|
        n_j^{-1/2}\nabla_\alpha q_{2,j,n_j}(X_j,\alpha_{j0})
        \right\|
        =O_p(1).
\]

\item Writing \(p_j\) for the prior density of \(\alpha_j\), \(p_j\) is continuously differentiable on \(\mathcal U_j(\delta_j)\). For every fixed \(z\), the pooling objective \(q_1\), viewed as a function of \(\alpha_j\) with the other group coordinates and \(\theta\) held fixed, has first and second partial derivatives on a neighborhood of \(\mathcal U_j(\delta_j)\); these derivatives are jointly continuous in \((\alpha,\theta)\) on \(\mathcal U_j(\delta_j)\times\prod_{\ell\ne j}\calA_\ell\times\Theta\).
\end{enumerate}
\end{assumption}

These high-level objective conditions for smooth finite-dimensional GMM impose the local criterion behavior needed for the quasi-posterior reduction once a study is classified as strong, without requiring the objective to be a likelihood. Assumption~\ref{ass:primitive-strong} is checked by verifying that the population criterion has a single interior maximizer, has nonsingular curvature there, and is uniformly approximated by the sample criterion and its first two derivatives near that point under the triangular-array law being studied. For a smooth finite-dimensional GMM criterion, this follows from uniform laws of large numbers for the moments and their derivatives, convergence of the weighting matrix when the criterion uses one, and a central limit theorem for the score at \(\alpha_{j0}\). The Hessian condition is often implied by full rank of the moment Jacobian under correct specification; under misspecification it can be checked directly from the population criterion. Under Assumption~\ref{ass:primitive-strong}, with probability approaching one there exists \(\widehat\alpha_j^{\mathrm{un}}\) such that
\[
        \widehat\alpha_j^{\mathrm{un}}
        \in
        \arg\max_{a\in\mathcal U_j(\delta_j)}
        q_{2,j,n_j}(X_j,a),
        \qquad
        \sqrt{n_j}\bigl(\widehat\alpha_j^{\mathrm{un}}-\alpha_{j0}\bigr)
        =O_p(1).
\]
Here \(\widehat\alpha_j^{\mathrm{un}}\) is the standard GMM estimator, and Theorem~\ref{thm:strong-equivalence} is stated as a quasi-posterior law result in the local coordinate centered at this maximizer. In the result below, fix \(j\in\mathcal J_S\), and let \(\widehat\alpha_j^{\mathrm{un}}\) be a measurable local maximizer over \(\mathcal U_j(\delta_j)\) satisfying \(\sqrt{n_j}(\widehat\alpha_j^{\mathrm{un}}-\alpha_{j0})=O_p(1)\). Let \(P^S_{j,n,\lambda}\) be the marginal law under the full QBHM quasi-posterior of \(h_j=\sqrt{n_j}(\alpha_j-\widehat\alpha_j^{\mathrm{un}})\), and let \(P^{S,0}_{j,n}\) be the law of the same \(h_j\) under the unpooled quasi-posterior \(\Pi^{\mathrm{un}}_{j,n}\).

\begin{theorem}\label{thm:strong-equivalence}
Suppose Assumptions~\ref{ass:primitive-setup}, \ref{ass:study-level-split}, and \ref{ass:primitive-strong} hold. Then
\[
        \sup_{\lambda\in\Lambda}
        d_{\mathrm{BL}}\!\left(P^S_{j,n,\lambda},N(0,H_j^{-1})\right)=o_p(1),
        \qquad
        d_{\mathrm{BL}}\!\left(P^{S,0}_{j,n},N(0,H_j^{-1})\right)=o_p(1),
\]
where \(d_{\mathrm{BL}}\) is the bounded-Lipschitz metric. Furthermore
\[
        \sup_{\lambda\in\Lambda}
        \left\|
        \int h_j\,P^S_{j,n,\lambda}(dh_j)
        \right\|
        +
        \left\|
        \int h_j\,P^{S,0}_{j,n}(dh_j)
        \right\|
        =o_p(1).
\]
Consequently,
\[
        \sup_{\lambda\in\Lambda}
        \left\|
        \sqrt{n_j}\bigl(\tilde\alpha^{\mathrm{full}}_{j,n,\lambda}-\alpha_{j0}\bigr)
        -
        \sqrt{n_j}\bigl(\widehat\alpha^{\mathrm{un}}_j-\alpha_{j0}\bigr)
        \right\|
        =o_p(1).
\]
\end{theorem}

The final result says that the marginal QBHM quasi-posterior mean is first-order equivalent to the unpooled local maximizer. More specifically, it implies that the QBHM estimator $\tilde\alpha^{\mathrm{full}}_{j,n,\lambda}$ and the usual GMM estimator $\widehat\alpha^{\mathrm{un}}_j$ have the same asymptotic distribution, no matter the value of $\lambda$ chosen. This parallels the main LTE result of \citet{ChernozhukovHong2003} - the difference here is that we focus on the quasi-posterior mean, rather than Bayes estimates defined by a more general loss function, and we show uniformity over the $\lambda$ term. For fixed pooling strengths, the theorem shows that, in a strongly identified study, the QBHM quasi-posterior mean differs from the unpooled local GMM maximizer by \(o_p(n_j^{-1/2})\), uniformly in \(\lambda\). Thus any usual first-order limit law for the unpooled maximizer carries over to the QBHM mean. In particular, uniform integrability implies that the two estimators have the same asymptotic MSE. The estimators can still differ in finite samples, and the prior can matter at that scale, but fixed pooling has no first-order effect under strong identification. Pooling can enter at first order only in the weak-study limit considered next.

\subsection{Weak studies}
\label{sec:weak-components}

For the retained weak spaces \(\calW_j\) introduced in Section~\ref{sec:dqm-identification}, set
\[
        \calW:=\prod_{j\in\mathcal J_W}\calW_j,
        \qquad
        \alpha_W=(\alpha_j)_{j\in\mathcal J_W}\in\calW .
\]
For weak study \(j\), define the local moment process
\[
        g_{j,n}(\alpha_j)
        :=
        G_{j,n_j}(X_j,\alpha_j)
        =
        n_j^{-1/2}\sum_{i=1}^{n_j}\phi_j(X_{ji},\alpha_j),
\]
and stack these processes as \(g_n(\alpha_W):=\bigl(g_{j,n}(\alpha_j)'\bigr)_{j\in\mathcal J_W}'\). Let \(k_W:=\sum_{j\in\mathcal J_W}k_j\). Under the primitive conditions below, \(g_n\) converges on \(\calW\) to \(g=m+\mathbb G\), and the weak-study objective converges to the grouped weak-GMM criterion \(Q\). In this notation, \(\alpha_W\) is a candidate retained weak value, \(\alpha_W^*\) is the fixed weak value selected by the local path, \(g\) is the random weak-GMM moment process, and \(m\) is its deterministic drift in the fixed weak-GMM state. Local coordinates \(h\) are used only after centering at \(\alpha_W^*\), so \(h=\alpha_W-\alpha_W^*\); when examples write a generic candidate value \(a\), it is translated into \(h\) by the same identity. The testing subsection later defines \(h_{a_0}\) as a residual process, not as this local coordinate. At this scale the hierarchy enters through the prior path \(\pi_\lambda\), so changing \(\lambda\) changes the limiting weak quasi-posterior mean even though $\phi_j$ is unchanged. The weak-study part of the QBHM lower level is
\begin{equation}
\label{eq:weak-lower-level-link}
        \sum_{j\in\mathcal J_W} q_{2,j,n_j}(X_j,\alpha_j)
        =
        -\frac12 Q_n(\alpha_W)+c_n,
        \qquad
        Q_n(\alpha_W)
        =
        g_n(\alpha_W)'\widehat W_n(\alpha_W)g_n(\alpha_W).
\end{equation}
The additive term \(c_n\) may depend on the data but not on \(\alpha_W\), and therefore cancels from normalized quasi-posteriors.  The display is only a reduction of the original weak-study criteria to the retained weak coordinates; it does not replace the lower-level objective by a different one.

The same weight matrix appears in the original lower-level criterion and in \(Q_n\).  With fixed weights, \(\widehat W_n\) is the fixed, or converging, block-diagonal GMM weight.  With continuously updated GMM, it is the estimated weight evaluated at the candidate value \(\alpha_W\).  The off-diagonal covariance kernel used for feasible weak-identification-robust tests is introduced in Section~\ref{sec:testing}; it is not part of the candidate-value objective in this subsection.  Mixed within-study blocks require a profiled weak experiment and are treated in Appendix~\ref{app:mixed-within-study-blocks}.

The assumption below fixes the local law, the retained weak spaces, the empirical-process limit, and the weight convergence.  Let \(P^W_{n,f}:=\bigotimes_{j\in\mathcal J_W}P_{j,n_j,f_j}^{n_j}\) be the product local law for the weak studies.  For the DQM score \(f_j\), define \(m_j(\alpha_j):=\mathbb E_{P_{j0}}[f_j(X_{ji})\phi_j(X_{ji},\alpha_j)]\) and \(m(\alpha_W):=(m_j(\alpha_j)')_{j\in\mathcal J_W}'\).  Let \(W:\calW\to\mathbb R^{k_W\times k_W}\) denote the nonrandom limit of the estimated GMM weight process.

\begin{assumption}\label{ass:primitive-weak}
The following conditions hold for the weak studies and the local direction \(f=(f_j)_{j\in\mathcal J_W}\).
\begin{enumerate}[label=(\roman*),leftmargin=2.1em]
\item For each \(j\in\mathcal J_W\), \((P_{j,n_j,f_j})_{n_j}\) is a DQM local path through \(P_{j0}\) with score \(f_j\in L_2^0(P_{j0})\). The score is pertinent, \(f_j\in\dot{\mathcal P}_{j0,\alpha}\), and the path remains in the identified submodel \(\mathcal P_{j,\alpha}\) for all large \(n_j\). Write \(\alpha_j^*:=\alpha^W_{j0}(f_j)\) for the associated weak value.
\item For each \(j\in\mathcal J_W\), \(\calW_j\subseteq\mathcal Z_{j0}\) is the retained weak space associated with the local sequence under consideration. The path-specific weak value satisfies \(\alpha_j^*\in\calW_j\) and \(m_j(\alpha_j^*)=0\).
\item For each \(j\in\mathcal J_W\), the vector-valued class \(\{\phi_j(\cdot,\alpha_j):\alpha_j\in\calW_j\}\) is \(P_{j0}\)-Donsker as a tight Borel element of \(\ell^\infty(\calW_j,\mathbb R^{k_j})\), has a square-integrable envelope, and is continuous as a map \(\alpha_j\mapsto\phi_j(\cdot,\alpha_j)\) into \(L_2(P_{j0})\). The corresponding reference-law Gaussian limit has a version with continuous sample paths.
\item Under \(P^W_{n,f}\), the weight process satisfies \(\sup_{\alpha_W\in\calW}\|\widehat W_n(\alpha_W)-W(\alpha_W)\|\pto0\). The limit weight \(W\) is continuous on \(\calW\), symmetric, and uniformly positive definite.
\end{enumerate}
\end{assumption}

Under Assumption~\ref{ass:primitive-weak}, \(\alpha_W^*:=(\alpha_j^*)_{j\in\mathcal J_W}\) belongs to \(\calW\) and satisfies \(m(\alpha_W^*)=0\), the state restriction used later for pointwise MSE and weak-identification-robust testing.  Clause \textup{(i)} selects the local sequence and therefore the drift \(m\), clause \textup{(ii)} records the retained-space restriction, and clauses \textup{(iii)}--\textup{(iv)} give the uniform central limit theorem and weight convergence needed for the weak quasi-posterior limit.  For continuously updated GMM, clause \textup{(iv)} is typically verified by uniform convergence and nonsingularity of the estimated covariance matrix. Under this condition, we get the following.

\begin{proposition}\label{prop:weak-process-main}
Suppose Assumptions~\ref{ass:primitive-setup}, \ref{ass:study-level-split}, and \ref{ass:primitive-weak} hold. Then, under \(P^W_{n,f}\), \((g_n,\widehat W_n)\dto(g,W)\) in \(\ell^\infty(\calW,\mathbb R^{k_W})\times\ell^\infty(\calW,\mathbb R^{k_W\times k_W})\), where \(g(\alpha_W)=m(\alpha_W)+\mathbb G(\alpha_W)\) and \(\mathbb G\) is mean-zero Gaussian with continuous sample paths. Consequently, \(Q_n\dto Q\) in \(\ell^\infty(\calW)\), where
\[
        Q(\alpha_W)=g(\alpha_W)'W(\alpha_W)g(\alpha_W).
\]
\end{proposition}

Proposition~\ref{prop:weak-process-main} supplies the random criterion for the reduced weak quasi-posterior.  Because \(Q\) is a process on the retained space \(\calW\), rather than a local quadratic expansion around a single point, the limit is not a single point, like it would be in the strongly-identified case. Since strong studies are consistently estimated by the QBHM, the weak-study limit includes the strong studies, evaluated at their true values $\alpha_{j0}$.  Write \(\calA_S:=\prod_{j\in\mathcal J_S}\calA_j\) and \(\alpha_{S0}:=(\alpha_{j0})_{j\in\mathcal J_S}\). For \(a_S=(a_j)_{j\in\mathcal J_S}\in\calA_S\), set \(q_{1,S}(a_S,\theta;z_S):=\sum_{j\in\mathcal J_S}q_{1,j}(a_j,\theta;z_j)\) and \(q_{1,W}(\alpha_W,\theta;z_W):=\sum_{j\in\mathcal J_W}q_{1,j}(\alpha_j,\theta;z_j)\), with an empty sum interpreted as zero.  The vectors \(z_S\) and \(z_W\) collect the corresponding fixed covariates.

For each weak study, let \(\mu_j^W\) be a finite retained-space dominating measure with full support on \(\calW_j\), obtained by restricting or reparametrizing the original base measure when \(\calW_j\) is full-dimensional and by using the chosen coordinate, Hausdorff, or other retained-space dominating measure when \(\calW_j\) is lower-dimensional. Let \(p_j^W\) denote the retained-space baseline prior density induced by \(\pi_j\) on \(\calW_j\) in the same restriction, reparametrization, or coordinate chart, and set \(\mu_W(d\alpha_W):=\prod_{j\in\mathcal J_W}\mu_j^W(d\alpha_j)\) and \(p_W^0(\alpha_W):=\prod_{j\in\mathcal J_W}p_j^W(\alpha_j)\). Thus \(p_W^0\mu_W\) is the unpooled baseline prior measure on the retained weak coordinates. The substantive hierarchy enters through
\[
        r_\lambda(\alpha_W)
        =
        \int_\Theta
        \exp\!\left(
        \lambda\left[q_{1,W}(\alpha_W,\theta;z_W)+q_{1,S}(\alpha_{S0},\theta;z_S)\right]
        \right)
        \pi_\theta(d\theta).
\]
The corresponding hierarchy-induced prior path on weak values is
\[
        \pi_\lambda(d\alpha_W)
        =
        p^\alpha_\lambda(\alpha_W)\mu_W(d\alpha_W),
        \qquad
        p^\alpha_\lambda(\alpha_W)
        =
        \frac{p_W^0(\alpha_W)r_\lambda(\alpha_W)}
        {\int_{\calW}p_W^0(u)r_\lambda(u)\mu_W(du)} .
\]
The notation separates the prior measure from its density: \(\pi_\lambda\) is the probability measure on weak values, while \(p^\alpha_\lambda\) is its density in \(\alpha_W\)-coordinates. They are not different priors. Both are induced by the lower-level baseline prior, the upper-level prior \(\pi_\theta\), the pooling objective \(q_1\), and the pooling strength \(\lambda\). In local coordinates \(h=\alpha_W-\alpha_W^*\), we write \(p_\lambda(h)=p^\alpha_\lambda(\alpha_W^*+h)\) with respect to the translated retained-space measure. When \(q_1\) rewards similarity across studies, increasing \(\lambda\) shifts prior mass toward weak values that better satisfy the pooling relation with the other weak studies and the strong-study limiting values.  In empirical terms, precisely estimated sites can affect the limiting weights placed on weakly estimated site-level treatment effects, IV slopes, production elasticities, or threshold values, without changing the weak-study moment criterion.

The finite-sample version replaces the limiting strong-study values by the unpooled strong-study estimates. Let \(\widehat\alpha^{\mathrm{un}}_{S,n}:=(\widehat\alpha_j^{\mathrm{un}})_{j\in\mathcal J_S}\), and set
\[
\begin{aligned}
        r_{n,\lambda}(\alpha_W)
        &=
        \int_\Theta
        \exp\!\left(
        \lambda\left[q_{1,W}(\alpha_W,\theta;z_W)
        +q_{1,S}(\widehat\alpha^{\mathrm{un}}_{S,n},\theta;z_S)\right]
        \right)
        \pi_\theta(d\theta),\\
        \pi_{n,\lambda}^{\mathrm{plug}}(d\alpha_W)
        &\propto
        p_W^0(\alpha_W)r_{n,\lambda}(\alpha_W)\mu_W(d\alpha_W).
\end{aligned}
\]

Under this definition, we have the following.

\begin{proposition}\label{prop:primitive-prior-main}
Suppose Assumptions~\ref{ass:primitive-setup}, \ref{ass:study-level-split}, and \ref{ass:primitive-strong} hold. Then the induced prior path \((\pi_\lambda)_{\lambda\in\Lambda}\) is proper, and for every \(f\in C_b(\calW)\), the map \(\lambda\mapsto\int f(\alpha_W)\pi_\lambda(d\alpha_W)\) is continuous on \(\Lambda\). Moreover,
\[
        \sup_{\lambda\in\Lambda}
        \|\pi_{n,\lambda}^{\mathrm{plug}}-\pi_\lambda\|_{TV}=o_p(1).
\]
Here and below, \(\|\cdot\|_{TV}\) denotes total variation distance.
\end{proposition}

Proposition~\ref{prop:primitive-prior-main} says that the hierarchy is a
well-behaved way to tilt the weak experiment.  The hierarchy need not be
literally correct as a model for how the studies were generated.  What matters
is that, after the weak coordinates are retained, each value of \(\lambda\)
gives a proper prior on \(\calW\), and small changes in \(\lambda\) lead only
to small changes in prior averages.  The strong studies enter this prior
through their first-order limits, so replacing those limits by the feasible
strong-study estimates does not change the induced prior path asymptotically,
uniformly over \(\lambda\in\Lambda\).
The finite-sample reduced quasi-posterior uses \(Q_n\), while the limiting quasi-posterior uses \(Q\):
\[
\begin{aligned}
        \Pi_{n,\lambda}(d\alpha_W\mid X,z)
        &=\frac{\exp[-Q_n(\alpha_W)/2]\pi_\lambda(d\alpha_W)}
        {\int_{\calW}\exp[-Q_n(u)/2]\pi_\lambda(du)}, \\
        \Pi_\lambda(d\alpha_W\mid g)
        &=\frac{\exp[-Q(\alpha_W)/2]\pi_\lambda(d\alpha_W)}
        {\int_{\calW}\exp[-Q(u)/2]\pi_\lambda(du)} .
\end{aligned}
\]
Let \(\Sigma(\alpha,\tilde\alpha):=\Cov(\mathbb G(\alpha),\mathbb G(\tilde\alpha))\). When \(W(\alpha_W)=\Sigma(\alpha_W,\alpha_W)^{-1}\), the limiting quasi-posterior is the continuously updated weak-GMM quasi-posterior. Appendix~\ref{app:am-diffuse-limit} gives the corresponding Andrews--Mikusheva diffuse-prior result: under the uniform diffuse-likelihood condition stated there, after absorbing the data-independent determinant factor into the structural prior, \(\Pi_\lambda(\cdot\mid g)\) is the diffuse-nuisance-prior limit of proper Andrews--Mikusheva Bayes posteriors. Using \(\pi_\lambda\), rather than the plug-in path, is harmless at this order because Proposition~\ref{prop:primitive-prior-main} gives uniform total-variation convergence. Let \(T^{W}_{n,\lambda}:=\int_{\calW}\alpha_W\,\Pi_{n,\lambda}(d\alpha_W\mid X,z)\) and \(t_\lambda(g):=\int_{\calW}\alpha_W\,\Pi_\lambda(d\alpha_W\mid g)\). Within the reduced weak experiment, the unpooled quasi-posterior mean is \(T^{W}_{n,0}\) when \(\pi_0\) is proper; the usual unpooled weak-GMM point estimator is any approximate minimizer of \(Q_n\).

Our main result for weak studies is that the QBHM quasi-posterior is asymptotically equivalent to the limit experiment, when we look at the marginal distribution for weakly identified studies.

\begin{proposition}\label{prop:weak-marginal-main}
Suppose Assumptions~\ref{ass:primitive-setup}, \ref{ass:study-level-split}, \ref{ass:primitive-strong}, and \ref{ass:primitive-weak} hold, and let \((\pi_\lambda)_{\lambda\in\Lambda}\) be the prior path in Proposition~\ref{prop:primitive-prior-main}. Under the corresponding product triangular-array law, the weak-study marginal of the full QBHM quasi-posterior is asymptotically equivalent, uniformly over \(\lambda\in\Lambda\), to the reduced weak-GMM quasi-posterior \(\Pi_{n,\lambda}\):
\[
        \sup_{\lambda\in\Lambda}
        \|\Pi^{\mathrm{full},W}_{n,\lambda}-\Pi_{n,\lambda}\|_{TV}=o_p(1).
\]
Consequently, \(\sup_{\lambda\in\Lambda}\|\int \alpha_W\,\Pi^{\mathrm{full},W}_{n,\lambda}(d\alpha_W\mid X,z)-T^{W}_{n,\lambda}\|=o_p(1)\).
\end{proposition}

This proposition links the original QBHM to the reduced weak-GMM experiment. The convergence is strong enough to transfer the posterior-mean rule from the reduced weak experiment back to the feasible QBHM weak marginal. Let \(d_W:=\sum_{j\in\mathcal J_W}d_j\), the dimension of \(\alpha_W\). This transfer is shown by the following.

\begin{theorem}\label{thm:main}
Suppose Assumptions~\ref{ass:primitive-setup}, \ref{ass:study-level-split}, \ref{ass:primitive-strong}, and \ref{ass:primitive-weak} hold, and let \((\pi_\lambda)_{\lambda\in\Lambda}\) be the induced prior path in Proposition~\ref{prop:primitive-prior-main}. Then, under the corresponding product triangular-array law,
\[
        (T^{W}_{n,\lambda})_{\lambda\in\Lambda}
        \dto
        (t_\lambda(g))_{\lambda\in\Lambda}
\]
in \(\ell^\infty(\Lambda,\mathbb R^{d_W})\). Together with Proposition~\ref{prop:weak-marginal-main},
\[
        \left(\int_{\calW}\alpha_W\,\Pi^{\mathrm{full},W}_{n,\lambda}(d\alpha_W\mid X,z)\right)_{\lambda\in\Lambda}
        \dto
        (t_\lambda(g))_{\lambda\in\Lambda}
\]
in \(\ell^\infty(\Lambda,\mathbb R^{d_W})\).
\end{theorem}

The convergence in Theorem~\ref{thm:main} is uniform over
\(\lambda\in\Lambda\).  Hence the same weak limit applies not only at any fixed
choice of \(\lambda\), but also along data-dependent choices
\(\hat\lambda_n\in\Lambda\).  This is useful for the reporting strategy in
Section~\ref{sec:lambda-inference}: one can display posterior means over a
range of hierarchy strengths, or select a value of \(\lambda\) using the data,
without changing the first-order weak asymptotic approximation.

The limit rule \(t_\lambda(g)\) has a direct Bayes-risk interpretation in the weak experiment. Once the weak-GMM process \(g\) has been observed, \(\Pi_\lambda(\cdot\mid g)\) is the quasi-posterior distribution over weak values, and its mean is the posterior squared-loss rule.  For a reported linear estimand, let \(B\) be a fixed matrix; for example, \(B\) may select one component, a subset of components, a contrast across studies, or the full vector. Let a decision rule be a measurable map \(t:g\mapsto t(g)\in\mathbb R^{d_W}\), and use loss \(L_B(a,\alpha_W)=\|B(a-\alpha_W)\|^2\). If \(\mathsf M\) is any probability law on weak-process sample paths for which the quasi-posterior is well defined, define the integrated posterior risk
\[
IPR_{\lambda,B,\mathsf M}(t)
=
\int
\left[
\int_{\calW}\|B(t(g)-\alpha_W)\|^2\Pi_\lambda(d\alpha_W\mid g)
\right]
\mathsf M(dg),
\]
interpreted as an extended nonnegative integral, so the value may be \(+\infty\).

\begin{proposition}\label{prop:posterior-risk}
Fix \(\lambda\in\Lambda\), a matrix \(B\), and a probability law \(\mathsf M\) on weak-process sample paths such that \(\Pi_\lambda(\cdot\mid g)\) is well defined for \(\mathsf M\)-almost every \(g\). \(t_\lambda\) is a\footnote{It is not necessarily unique, and so we cannot refer to it as \textit{the} Bayes rule.} Bayes rule: it minimizes \(IPR_{\lambda,B,\mathsf M}\) among measurable finite-valued rules. For every such rule \(t\),
\[
IPR_{\lambda,B,\mathsf M}(t)-IPR_{\lambda,B,\mathsf M}(t_\lambda)
=
\int\|B(t(g)-t_\lambda(g))\|^2\mathsf M(dg)
\geq 0,\]
with both sides interpreted in the extended sense. If \(IPR_{\lambda,B,\mathsf M}(t)<\infty\), the identity is an ordinary finite equality.
\end{proposition}

Proposition~\ref{prop:posterior-risk} gives the posterior squared-loss meaning of the pooled weak-GMM quasi-posterior. Combined with Proposition~\ref{prop:weak-marginal-main} and Theorem~\ref{thm:main}, it gives a decision-theoretic justification for the use of the QBHM estimator: for each fixed \(\lambda\), the weak marginal of the feasible QBHM estimator is asymptotically equivalent to a Bayes rule under the hierarchy-induced weak-limit prior \(\pi_\lambda\) and squared loss. In the continuously updated case, Proposition~\ref{prop:am-diffuse-limit} further shows that, under the diffuse-likelihood conditions stated there, this weak-limit posterior can be obtained as an Andrews--Mikusheva diffuse-nuisance-prior limit. Corollary~\ref{cor:mixed-reported-targets} in Appendix~\ref{app:mixed-within-study-blocks} generalizes the result to the case where the same study can contain both strongly and weakly identified components. The conclusion is unchanged: fixed pooling does not alter the first-order law of genuinely regular coordinates, while weak coordinates retain the nondegenerate shrinkage comparison.

\section{When does pooling lower MSE?}
\label{sec:mse-improvement}

The previous section gives the posterior decision-theoretic interpretation of the QBHM estimator in the weak-GMM limit experiment.  This section asks the pointwise frequentist question: for a fixed weak-GMM state \((\alpha_W^*,m)\), when does pooling lower MSE relative to an unpooled estimator? An applied version of this is essentially asking: when does borrowing from related sites reduce repeated-sampling error for a weakly estimated treatment effect. As indicated above, the answer comes from a bias-variance trade-off: pooling helps when the variance reduction from shrinkage exceeds the bias introduced by the pooling relation. The aim of this section is to formalize this trade-off for general QBHMs, and provide some more interpretable conditions in salient special cases. Throughout, \(\alpha_W\) denotes the weak-study vector from Section~\ref{sec:weak-components}, and \(\calW\) denotes its compact retained weak space.  The state restriction \(m(\alpha_W^*)=0\) means that the candidate weak value \(\alpha_W^*\) is compatible with the first-order drift of the weak moments. Appendix~\ref{app:mixed-within-study-blocks} gives the corresponding formulas for profiled weak coordinates when a single study contains both regular and weak components.

\subsection{Pointwise MSE}\label{sec:pointwise-risk}

Proposition~\ref{prop:posterior-risk} is a posterior risk statement:
after observing the weak-GMM process \(g\), the quasi-posterior mean is
optimal for posterior squared loss under the hierarchy-induced weak-limit
prior. Here, by contrast, we consider repeated-sampling MSE at a fixed
weak-GMM state \(e=(\alpha_W^*,m)\), with \(m(\alpha_W^*)=0\). Thus
\(\mathbb E_m\) denotes expectation over draws of the weak-GMM limit process
\(g\), holding fixed both the true weak value \(\alpha_W^*\) and the drift
\(m\). The distinction is exactly analogous to the distinction between
Proposition~\ref{prop:normal-normal-risk} and
Corollary~\ref{cor:normal-pointwise-mse} in the normal-normal example:
the former averages risk using the posterior distribution, whereas the
latter evaluates frequentist risk at a fixed underlying state. For any measurable rule \(t:g\mapsto t(g)\in\mathbb R^{d_W}\), and for any fixed matrix \(B\), define
\[
        M_B(t;\alpha_W^*,m)=\mathbb E_m\|B(t(g)-\alpha_W^*)\|^2 .
\]
This is the pointwise asymptotic MSE of the reported estimand \(B\alpha_W\).  The choice of \(B\) records the estimand being reported: \(B=I_{d_W}\) gives the vector MSE, a row vector \(B=b'\) gives the scalar MSE for \(b'\alpha_W^*\), and a coordinate-selection matrix gives the MSE for selected components.  The word pointwise is important because the risk does not average over the hierarchy, over possible values of \(\alpha_W^*\), or over a population distribution of studies.

The finite-sample analogue is \(\mathcal R_{n,\lambda}=\mathbb E_{P_{n,f}}\|B(T_{n,\lambda}-\alpha_W^*)\|^2\), where \(T_{n,\lambda}\) is the feasible estimator path, for example the weak-study component of the QBHM quasi-posterior mean.  Proposition~\ref{prop:transfer} connects this finite-sample risk to the weak-limit risk.  If \(B(T_{n,\lambda}-\alpha_W^*)\dto B(t_\lambda(g)-\alpha_W^*)\), uniformly over \(\lambda\in\Lambda\), and the squared errors are uniformly integrable, then \(\sup_{\lambda\in\Lambda}|\mathcal R_{n,\lambda}-M_B(t_\lambda;\alpha_W^*,m)|\to0\).  Uniform integrability is a tail condition; it is automatic when the reported estimator path is bounded on a compact weak space and also follows from a common \(2+\delta\) moment bound. As such, a strict MSE improvement in the weak-limit experiment transfers to the original finite-sample sequence for all sufficiently large \(n\).

The comparison below is therefore a weak-study comparison.  The comparator is the benchmark rule whose MSE is used as the reference point. In the QBHM comparison, the comparator \(t_0(g)\) is the unpooled weak-limit quasi-posterior mean obtained when \(\lambda=0\) and the baseline prior is proper, while \(t_\lambda(g)\) is the corresponding pooled weak-limit quasi-posterior mean. The same notation also allows \(t_0\) to denote another benchmark, such as an unpooled GMM rule or a leave-one-study-out rule. Define \(e_0(g)=B(t_0(g)-\alpha_W^*)\) as the benchmark estimation error for the reported estimand and \(d_\lambda(g)=t_\lambda(g)-t_0(g)\) as the change in the weak estimate caused by pooling.

\begin{proposition}\label{prop:identity}
Fix a weak-GMM limit experiment \((\alpha_W^*,m)\), a matrix \(B\), a measurable benchmark rule \(t_0\) (the comparator), and a measurable pooled rule \(t_\lambda\).  If \(M_B(t_0;\alpha_W^*,m)<\infty\) and \(M_B(t_\lambda;\alpha_W^*,m)<\infty\), then \(Bd_\lambda\) is square integrable, \(e_0'Bd_\lambda\) is integrable, and
\[
M_B(t_\lambda;\alpha_W^*,m)-M_B(t_0;\alpha_W^*,m)
=
2\mathbb E_m\!\left[e_0(g)'B d_\lambda(g)\right]
+
\mathbb E_m\!\left\|B d_\lambda(g)\right\|^2 .
\]
\end{proposition}

Proposition~\ref{prop:identity} says that pooling lowers MSE exactly when the covariance-type term \(2\mathbb E_m[e_0(g)'B d_\lambda(g)]\) is negative enough to offset the nonnegative cost term \(\mathbb E_m\|B d_\lambda(g)\|^2\).  In words, pooling helps when the hierarchy tends to move the unpooled estimate back toward the truth, and it hurts when it tends to move the estimate away from the truth or when the pooling movement is too large relative to the error it corrects.

Two sufficient conditions are useful for interpreting the later results and are stated formally as Proposition~\ref{prop:sufficient} in the appendix.  First, it is trivially true that if the unpooled rule has infinite pointwise asymptotic MSE and the pooled rule has finite pointwise asymptotic MSE, then pooling strictly improves MSE.  This is the extreme version of the variance-reduction logic, where the hierarchy regularizes an otherwise unstable weakly identified rule.  Second, consider local coordinates \(h=\alpha_W-\alpha_W^*\).  Suppose a sequence of hierarchy-induced priors concentrates around a local center \(h_c\), and let \(\delta_v(g)\) be the corresponding quasi-posterior mean.  If \(\delta_v(g)\to h_c\) in \(L^2(P_m)\), then the MSE of the strongly pooled rule converges to \(\|Bh_c\|^2\).  Therefore strong pooling improves on \(t_0\) whenever
\[
        \|Bh_c\|^2
        <
        M_B(t_0;\alpha_W^*,m).
\]
The intuition is that very strong pooling makes the weak-study data almost irrelevant for the reported weak coordinate: the quasi-posterior mean collapses toward the target selected by the hierarchy.  The remaining risk is then not a variance term from weak identification, but the squared distance between that target and the truth, measured after applying \(B\).  Thus strong pooling is helpful only when replacing the noisy unpooled rule by the pooling target creates less squared error than the unpooled rule has on average.  If the target is close to the true weak value, the variance reduction can dominate; if the target is badly centered, strong pooling simply replaces variance with bias.  Lemma~\ref{lem:concentration} verifies the needed concentration step on compact local spaces, and the condition does not require the weak objective itself to become strongly curved.

The conditions above are general but abstract. The next two subsections make them more concrete: first by studying when a small amount of pooling helps relative to no pooling, and then by working through a scalar weak-IV example.

\subsection{A local condition}\label{sec:local}

Starting from no pooling, the relevant local question is the sign of the first derivative of pointwise MSE with respect to \(\lambda\) at zero.  The derivative separates the first movement of the weak quasi-posterior mean from the existing unpooled error, and in the quadratic case it becomes an explicit shrinkage-gain versus centering-cost comparison.

Work in local coordinates \(h=\alpha_W-\alpha_W^*\), and write \(\calH=\calW-\alpha_W^*\).  In these coordinates the true weak value is the origin.  Suppose \(\calH\subset\mathbb R^{d_W}\) is compact with positive finite coordinate measure; lower-dimensional weak spaces are interpreted using a fixed coordinate parametrization.  For a realized weak-GMM limit process \(g\), let \(Q_g(h)=Q(\alpha_W^*+h)\). Using the density notation introduced above, \(p_\lambda(h)\) is the local-coordinate density of the hierarchy-induced prior path \(\pi_\lambda\). The corresponding weak-limit quasi-posterior on local coordinates is the probability measure \(\nu_\lambda^g\) with density proportional to \(\exp[-Q_g(h)/2]p_\lambda(h)\) on \(\calH\).  Its mean
\[
        \delta_\lambda(g)=\int_{\calH} h\nu_\lambda^g(dh)
\]
is the local error of the reported rule, so \(M_{\lambda,B}(\alpha_W^*,m)=\mathbb{E}_m\|B\delta_\lambda(g)\|^2\). At \(\lambda=0\), write \(\bar h(g)=\int h\nu_0^g(dh)\),
\[
        \Omega(g)=\int(h-\bar h(g))(h-\bar h(g))'\nu_0^g(dh),
\]
and, for \(P=P'\succeq0\),
\[
        \tau_P(g)=\int(h-\bar h(g))(h-\bar h(g))'P(h-\bar h(g))\nu_0^g(dh).
\]
Here \(\bar h(g)\) is the unpooled local mean error, \(\Omega(g)\) is the unpooled quasi-posterior variance, and \(\tau_P(g)\) is the third-central-moment correction that appears when the unpooled weak quasi-posterior is not centrally symmetric around \(\bar h(g)\).

The calculation only differentiates the hierarchy-induced prior path.  It does not require differentiability or strong curvature of the weak objective \(Q_g\).

\begin{assumption}\label{ass:smooth-prior}
The following conditions hold for some \(\varepsilon>0\).
\begin{enumerate}[label=(\roman*),leftmargin=2.1em]
\item The densities \(p_\lambda\) satisfy \(0<\inf_{h\in\calH,0\le\lambda\le\varepsilon}p_\lambda(h)\le\sup_{h\in\calH,0\le\lambda\le\varepsilon}p_\lambda(h)<\infty\).
\item For each \(h\), \(\lambda\mapsto\log p_\lambda(h)\) is continuously differentiable on \([0,\varepsilon]\). The derivative \(\ell_\lambda(h)=\partial_\lambda\log p_\lambda(h)\) is jointly measurable in \((\lambda,h)\), uniformly bounded on \([0,\varepsilon]\times\calH\), and satisfies \(\ell_\lambda(h)\to\ell_0(h)\) as \(\lambda\downarrow0\) for every \(h\in\calH\).
\end{enumerate}
\end{assumption}

The score \(\ell_0\) is the first-order direction in which the hierarchy changes the local prior.  Since \(p_\lambda\) is normalized and \(\ell_\lambda\) is uniformly bounded, dominated convergence gives \(\int_{\calH} \ell_0(h)p_0(h)\,dh=0\). Thus the local prior score is automatically centered. If \(p_\lambda(h)\propto p_0(h)\exp[\lambda \tilde\ell_0(h)]\), with \(p_0\) bounded away from zero and infinity and \(\tilde\ell_0\) bounded, the normalized prior has score
\[
        \ell_0(h)=\tilde\ell_0(h)-\int_{\calH} \tilde\ell_0(u)p_0(u)\,du.
\]
The centering constant has no effect on the covariance formulas below.

Under the compact local-coordinate conditions and Assumption~\ref{ass:smooth-prior}, Proposition~\ref{prop:small} in the appendix shows that, as \(\lambda\downarrow0\),
\[
M_{\lambda,B}(\alpha_W^*,m)
=
M_{0,B}(\alpha_W^*,m)
+
2\lambda
\mathbb{E}_m\left[
(B\bar h(g))'B\operatorname{Cov}_{\nu_0^g}(h,\ell_0(h))
\right]
+
o(\lambda).
\]
The covariance term is the instantaneous movement of the quasi-posterior mean caused by the hierarchy.  Small positive pooling lowers pointwise asymptotic MSE whenever this average directional derivative is strictly negative.

We can make this more explicit when the hierarchy induces a quadratic tilt of the local prior. Suppose the pooling contribution in local coordinates is \(q_1^W(h)=-(h-h_c)'P(h-h_c)/2\), where \(h_c\) is the local pooling center and \(P=P'\succeq0\). Since \(p_\lambda(h)\propto p_0(h)\exp[\lambda q_1^W(h)]\), the tilted density satisfies \(p_\lambda(h)\propto p_0(h)\exp[-\lambda(h-h_c)'P(h-h_c)/2]\). If \(p_0\) is the density of \(N(m_0,V_0)\), the exponent is \(-h'(V_0^{-1}+\lambda P)h/2+h'(V_0^{-1}m_0+\lambda P h_c)\), up to a constant, so \(p_\lambda\) is again normal, with precision \(V_0^{-1}+\lambda P\). Thus a normal prior and squared-loss pooling generate the quadratic path.

By the pooling manifold we mean the values that exactly satisfy the pooling relation: they are the values toward which the hierarchy shrinks and that receive no quadratic penalty. If \(P\) is nonsingular, the pooling manifold is the single point \(h_c\); if \(P\) is singular, it is an affine set. In the present local coordinates, the pooling manifold is \(\{h\in\calH:P(h-h_c)=0\}\). In the fixed-center normal-normal hierarchy with \(P=I\), the manifold is the single point \(h_c\). In a common-mean normal hierarchy with \(J\) weak values and a diffuse common mean, the induced penalty is on deviations from the group average, \(P=I_J-J^{-1}\one\one'\), so the manifold is \(\{h:h_1=\cdots=h_J\}\); the hierarchy shrinks cross-study differences but leaves the common level unpenalized. For the quadratic path,
\[
        p_\lambda(h)
        =
        \frac{
        \exp[-\lambda (h-h_c)'P(h-h_c)/2]p_0(h)}
        {\int_{\calH}\exp[-\lambda (u-h_c)'P(u-h_c)/2]p_0(u)\,du}.
\]
The local prior score at no pooling is
\[
        \ell_0(h)
        =
        \left.\partial_\lambda\log p_\lambda(h)\right|_{\lambda=0}
        =
        c+h_c'Ph-\frac12 h'Ph,
\]
where the scalar \(c\) is the derivative of the normalizing constant.  Differentiating only the prior weight gives, for each realized weak-GMM process \(g\),
\[
        \dot\delta_0(g)
        :=
        \left.\partial_\lambda\delta_\lambda(g)\right|_{\lambda=0+}
        =
        \operatorname{Cov}_{\nu_0^g}(h,\ell_0(h))
        =
        \Omega(g)P[h_c-\bar h(g)]-\frac12\tau_P(g).
\]
Define
\[
\begin{aligned}
        A_B(P)&=\mathbb{E}_m\left[(B\bar h(g))'B\left[\Omega(g)P\bar h(g)+\frac12\tau_P(g)\right]\right],\\
        L_B(P)&=\mathbb{E}_m\left[P\Omega(g)B'B\bar h(g)\right].
\end{aligned}
\]
Then
\[
M_{\lambda,B}(\alpha_W^*,m)
=
M_{0,B}(\alpha_W^*,m)
+
2\lambda[h_c'L_B(P)-A_B(P)]
+
o(\lambda).
\]
The term \(A_B(P)\) is the first-order shrinkage gain that would remain if the hierarchy were centered at the true local value \(h=0\).  It is large when the unpooled rule has error in directions where the quasi-posterior still has dispersion to shrink.  The term \(h_c'L_B(P)\) is the first-order centering cost from pulling toward \(h_c\) rather than toward the truth.  Hence small positive pooling lowers asymptotic MSE whenever \(h_c'L_B(P)<A_B(P)\). If the unpooled weak quasi-posterior is centrally symmetric around \(\bar h(g)\), then \(\tau_P(g)=0\), and the gain margin reduces to \(\mathbb E_m[(B\bar h(g))'B\Omega(g)P\bar h(g)]\).

The same derivative calculation yields a uniform local improvement result, stated in \(\alpha_W\)-coordinates.  For a weak-GMM limit experiment \(e=(\alpha_W^*,m)\), let
\[
        M_{\lambda,B}(e)=\mathbb{E}_m\|B(t_\lambda(g)-\alpha_W^*)\|^2 .
\]
Under the unpooled quasi-posterior \(\Pi_0(d\alpha_W\mid g)\), define the unpooled error and dispersion objects
\[
\begin{aligned}
        r_e(g)&=t_0(g)-\alpha_W^*,\\
        \Omega_e(g)&=\int_{\calW}(\alpha_W-t_0(g))(\alpha_W-t_0(g))'\Pi_0(d\alpha_W\mid g),\\
        \tau_{e,P}(g)&=\int_{\calW}(\alpha_W-t_0(g))(\alpha_W-t_0(g))'P(\alpha_W-t_0(g))\Pi_0(d\alpha_W\mid g).
\end{aligned}
\]
Thus \(r_e(g)\) is the unpooled weak-limit error, \(\Omega_e(g)\) is the unpooled quasi-posterior variance, and \(\tau_{e,P}(g)\) is the same skewness correction in the original coordinates.  Set
\[
\begin{aligned}
        A_e(P)&=\mathbb{E}_m\left[(B r_e(g))'B\left[\Omega_e(g)P r_e(g)+\frac12\tau_{e,P}(g)\right]\right],\\
        L_e(P)&=\mathbb{E}_m(\Omega_e(g)B'B r_e(g)).
\end{aligned}
\]
The scalar \(A_e(P)\) is the gain margin from shrinking in the penalized directions.  The vector \(L_e(P)\) is the coefficient multiplying the centering error \(P(\alpha_W^*-\bar\alpha_W)\) in the expansion below.  Finally, let \(D_{\calW}=\sup_{u,v\in\calW}\|u-v\|\) and \(K=\max(1,D_{\calW}^3\|B\|_{\mathrm{op}}^2)\); this \(K\) is only a compactness bound used to state a uniform neighborhood around the pooling manifold.

\begin{theorem}\label{thm:uniform-local}
Let \(\calW\subset\mathbb R^{d_W}\) be compact, let \(B\) be fixed, let \(P=P'\succeq0\), and let \(\bar\alpha_W\in\mathbb R^{d_W}\).  Suppose the induced prior path on \(\calW\) has density proportional to
\(\exp[-\lambda(\alpha_W-\bar\alpha_W)'P(\alpha_W-\bar\alpha_W)/2]p_0(\alpha_W)\), where \(p_0\) is continuous, bounded, and bounded away from zero on \(\calW\).  Let \(\mathcal E\) be any class of weak-GMM limit experiments \(e=(\alpha_W^*,m)\) such that, for every \(e\in\mathcal E\), \(Q(\cdot)\) is finite and continuous almost surely and \(t_\lambda(g)\) is well defined almost surely for every \(\lambda\in[0,\bar\lambda]\).  There is a constant \(C<\infty\), depending only on \(\calW\), \(B\), \(P\), \(\bar\alpha_W\), and \(\bar\lambda\), such that, uniformly over \(e\in\mathcal E\) and uniformly over \(0\le\lambda\le\bar\lambda\),
\[
M_{\lambda,B}(e)-M_{0,B}(e)
=
-2\lambda\left[A_e(P)+[P(\alpha_W^*-\bar\alpha_W)]'L_e(P)\right]+R_e(\lambda),
\qquad
|R_e(\lambda)|\le C\lambda^2 .
\]
If a subclass \(\mathcal C\subset\mathcal E\) satisfies, for some \(a>0\), \(A_e(P)\ge a\) and \(\|P(\alpha_W^*-\bar\alpha_W)\|\le a/(2K)\) for every \(e=(\alpha_W^*,m)\in\mathcal C\), then there is \(\lambda^*>0\) such that, uniformly over \(e\in\mathcal C\) and \(0<\lambda\le\lambda^*\), \(M_{\lambda,B}(e)-M_{0,B}(e)\le -a\lambda/2<0\).
\end{theorem}

A scalar fixed-center specialization makes the two terms in Theorem~\ref{thm:uniform-local} explicit.  Let the weak parameter be scalar, take \(B=1\), and use a quadratic prior path centered at \(c\), so that \(P=1\) and \(\bar\alpha_W=c\).  Write \(r(g)=t_0(g)-\alpha_W^*\), \(\omega(g)=\operatorname{Var}_{\Pi_0(\cdot\mid g)}(\alpha_W)\), and \(\kappa_3(g)=\int(\alpha_W-t_0(g))^3\Pi_0(d\alpha_W\mid g)\).  The expansion becomes
\[
M_\lambda-M_0
=
-2\lambda
\left[
\mathbb E_m\left[r(g)^2\omega(g)+\frac12r(g)\kappa_3(g)\right]
+(\alpha_W^*-c)\mathbb E_m[\omega(g)r(g)]
\right]
+O(\lambda^2).
\]
If the unpooled weak quasi-posterior is locally symmetric, \(\kappa_3(g)=0\).  If the hierarchy is correctly centered, \(c=\alpha_W^*\), small positive pooling lowers MSE whenever \(\mathbb E_m[r(g)^2\omega(g)]>0\): the rule has nonzero weak-limit error and the unpooled quasi-posterior has dispersion to shrink.  If \(c\ne\alpha_W^*\), the extra term \((\alpha_W^*-c)\mathbb E_m[\omega(g)r(g)]\) is the centering-bias cost.

The two inequalities in Theorem~\ref{thm:uniform-local} have a direct interpretation. The first says that there must be residual dispersion for pooling to reduce: the unpooled weak quasi-posterior must still be spread out in directions that matter for the reported target \(B\alpha_W\), and the quadratic pooling term \(P\) must shrink those directions. This is the possible variance gain from pooling. The second says that the true weak value must not be too far from the values favored by the pooling relation. More precisely, it only has to be close in the directions that \(P\) actually penalizes. If \(P\) is singular, some directions are left alone. The set of values left unpenalized is the pooling manifold. In common-mean pooling, for example, the hierarchy shrinks cross-study deviations but does not shrink the common level, so the common level may be far from zero without violating the condition. Corollary~\ref{cor:projection-gain} gives a projection-based sufficient condition for the variance-gain part of the theorem, and the appendix records scalar fixed-center, coordinatewise ridge, and common-mean examples. In applications, large movements of \(\tilde\alpha_{n,\lambda}^{\mathrm{full}}\) along a reported \(\lambda\)-path should be read as evidence that the conclusion depends materially on the chosen pooling relation.

\subsection{Scalar just-identified weak IV}\label{sec:weak-iv}

The scalar weak-IV example isolates a case in which the full comparison with \(\lambda>0\) can be written in closed form. Consider one study and suppress the study index. Let \(Y_i\) be the outcome, \(X_i\) a scalar endogenous regressor, and \(Z_i\) a scalar instrument, after any desired residualization on controls. The familiar just-identified IV model is
\[
        Y_i=\alpha_W^* X_i+u_i,
        \qquad
        X_i=\pi_n Z_i+v_i,
        \qquad
        \pi_n=\gamma/\sqrt n,
\]
with \(\mathbb E[Z_i u_i]=0\) and \(\mathbb E[Z_i v_i]=0\). Here \(\gamma\) is the example-specific local first-stage coefficient. The sequence \(\pi_n=\gamma/\sqrt n\) is the weak-first-stage sequence: the first stage is local to zero, so the IV slope is not pinned down at the usual \(\sqrt n\) scale.

For a candidate slope \(a\), the just-identified IV/GMM moment is \(G_n(a)=n^{-1/2}\sum_{i=1}^n Z_i(Y_i-aX_i)\). Writing \(h=a-\alpha_W^*\), we have \(G_n(\alpha_W^*+h)=G_{u,n}-G_{x,n}h\), where \(G_{u,n}=n^{-1/2}\sum_{i=1}^n Z_i u_i\) and \(G_{x,n}=n^{-1/2}\sum_{i=1}^n Z_i X_i\). Under the weak-IV sequence, \((G_{u,n},G_{x,n})\dto(G_u,G_x)\). Thus the weak-limit moment is \(M(h)=G_u-G_xh\), with scalar GMM weight one. The unpooled IV weak-limit rule solves \(M(h)=0\), giving \(H=G_u/G_x\) when \(G_x\ne0\). An induced normal pooling prior \(h\sim N(h_c,\lambda^{-1})\), with scalar pooling precision \(\lambda>0\), gives the pooled quasi-posterior mean \(\Delta_\lambda=(G_xG_u+\lambda h_c)/(G_x^2+\lambda)\). Here \(\Delta_\lambda\) is the weak-limit error of the pooled slope estimator, so the corresponding slope is \(\alpha_W^*+\Delta_\lambda\). For the conditional comparison, let \(b(G_x)=\mathbb{E}_m(G_u\mid G_x)\) and \(\sigma^2(G_x)=\operatorname{Var}_m(G_u\mid G_x)\).  Conditional on \(G_x\ne0\),
\[
\mathbb{E}_m(H^2\mid G_x)=\frac{b(G_x)^2+\sigma^2(G_x)}{G_x^2},
\qquad
\mathbb{E}_m(\Delta_\lambda^2\mid G_x)
=
\frac{G_x^2\sigma^2(G_x)+[G_xb(G_x)+\lambda h_c]^2}{(G_x^2+\lambda)^2}.
\]
Subtracting the pooled conditional MSE from the unpooled conditional MSE gives
\[
\mathbb{E}_m(H^2\mid G_x)-\mathbb{E}_m(\Delta_\lambda^2\mid G_x)
=
\frac{\lambda D_\lambda(G_x)}{G_x^2(G_x^2+\lambda)^2},
\]
where
\[
D_\lambda(G_x)
=
(2G_x^2+\lambda)(b(G_x)^2+\sigma^2(G_x))
-2G_x^3b(G_x)h_c
-\lambda G_x^2h_c^2 .
\]

\begin{proposition}\label{prop:iv}
In the scalar just-identified weak-IV limit above, suppose the conditional first and second moments defining \(b(G_x)\) and \(\sigma^2(G_x)\) exist at the realized value of \(G_x\).  Condition on a realized \(G_x\ne0\).  For any fixed \(\lambda>0\), the pooled rule has lower conditional MSE than the unpooled IV rule if and only if \(D_\lambda(G_x)>0\). If \(P_m(G_x\ne0)=1\), \(\mathbb{E}_m[(b(G_x)^2+\sigma^2(G_x))/G_x^2]=\infty\), \(\mathbb{E}_m(G_x^2G_u^2)<\infty\), and \(h_c\) is finite, then unpooled weak IV has infinite asymptotic MSE, while the pooled rule has finite asymptotic MSE uniformly over every compact \(\Lambda_{IV}\subset(0,\bar\lambda]\).
\end{proposition}

Uniform improvement over a set \(\Lambda_1\) of positive pooling precisions requires \(\inf_{\lambda\in\Lambda_1}D_\lambda(G_x)>0\).  The conditional mean \(b(G_x)\) records remaining local drift in the structural reduced-form component after conditioning on the weak first stage.  In the orthogonalized case, \(b(G_x)=0\), as in the canonical Gaussian weak-IV limit when the structural reduced-form component is uncorrelated with the first-stage limit, or after replacing it by the residual from its linear projection on \(G_x\).  The sign condition then reduces to
\[
h_c^2<\sigma^2(G_x)(G_x^{-2}+2\lambda^{-1}),
\]
with a positive margin uniformly over \(\lambda\in\Lambda_1\) for uniform improvement on \(\Lambda_1\).  The left side is the squared centering error of the hierarchy.  The right side is the weak-identification variance gain from replacing \(G_x^{-1}\) by the ridge-stabilized factor \(G_x/(G_x^2+\lambda)\).  When \(G_x\) is close to zero, the allowable centering error is large because unpooled IV is very noisy; when \(G_x\) is large, the allowable centering error is smaller because the unpooled ratio is already stable.

The finite-MSE part of Proposition~\ref{prop:iv} allows \(G_x^{-1}\) to have infinite second moment while requiring the ridge-stabilized numerator \(G_xG_u\) to remain square integrable.  Positive pooling then regularizes the weak-IV ratio uniformly on compact subsets bounded away from \(\lambda=0\).

\section{Reporting \texorpdfstring{$\lambda$}{lambda} and conducting inference}
\label{sec:lambda-inference}

\subsection{Pooling Strength}\label{sec:lambda}\label{sec:interpretation-inference}

The MSE results have a direct implication: a QBHM should not be presented as a single pooled estimate indexed by an unexplained value of \(\lambda\). The numerical value of \(\lambda\) is interpretable only after the normalization of \(q_1\) has been fixed. In the normal--normal example, \(\lambda\) determines the effective weight \(c_{j,n}(\lambda)\) on the unpooled group summary in Proposition~\ref{prop:normal-normal-summary}. Equivalently, when \(q_1\) is a Gaussian least-squares term with scale \(\tau_0^2\), the hierarchical precision is \(\lambda/\tau_0^2\). Changing \(\lambda\) and changing the variance scale inside \(q_1\) therefore change the same practical object: the amount of pooling. The role of \(\lambda\) in the general QBHM is both to handle \(q_1\) functions that do not have a natural ``variance'' term and to emphasize that researchers should report results over a grid of \(\lambda\) values, rather than at a single chosen value.

In applications, this should be a grid \(\Lambda_G\) that gives a range of strong to weak pooling under the chosen normalization of \(q_1\). A dense grid near zero is useful when small amounts of pooling are likely to matter, as highlighted by Theorem \ref{thm:uniform-local}. Tables should report marginal quasi-posterior means and quasi-posterior intervals over the grid, together with the exact normalization of the pooling term and the prior distribution used. Ideally, results should be robust to changes in \(\lambda\) and the prior, since these are chosen by the researcher. Theorem \ref{thm:strong-equivalence} suggests that this should be true for strongly identified studies, but for weakly identified studies, both the prior and \(\lambda\) remain asymptotically relevant. It may nevertheless be the case that this does not matter much in practice, and reporting these checks can make the results more convincing to others. If the study's parameter changes substantially with \(\lambda\) and/or the prior, a researcher should note this and try to provide an empirical basis for their main choices, such as genuine prior knowledge from earlier studies, empirical Bayes for the prior, or cross-validation for \(\lambda\).

In such settings, cross-validation can provide a heuristic way to choose a good value for \(\lambda\). It may also be useful if a researcher wants to present one main set of results for policy reasons, provided that they still report results over several different values of \(\lambda\). For leave-one-study-out cross-validation, fit the QBHM on the \(J-1\) studies other than \(j\), and use that fit to predict an object in the left-out study. Let \(x_{-j}\) and \(z_{-j}\) collect the data and covariates from all studies except \(j\), and let \(\Pi^{\theta}_{n,\lambda,-j}(d\theta\mid x_{-j},z_{-j})\) be the marginal quasi-posterior for \(\theta\) computed without study \(j\). Define
\[
        w_{j,\lambda,-j}(a_j)
        =\int \exp[\lambda q_{1,j}(a_j,\theta;z_j)]
        \Pi^{\theta}_{n,\lambda,-j}(d\theta\mid x_{-j},z_{-j}).
\]
The induced predictive law for the left-out group parameter is
\[
        \Pi^{\mathrm{pred}}_{\lambda,-j}(d a_j\mid x_{-j},z)
        =
        \frac{w_{j,\lambda,-j}(a_j)\pi_j(d a_j)}
        {\int w_{j,\lambda,-j}(\tilde a_j)\pi_j(d\tilde a_j)} .
\]
In order to perform cross-validation, we need a way to evaluate how well the predicted law performs. Fortunately, many structural applications already provide such a criterion. For example, it is common to evaluate a model using moments that were not used to estimate its parameters, or more generally to assess whether it matches features of the data that it was not explicitly designed to match. This idea can be formalized through a validation criterion \(\mathcal V_j(a_j;x_j,z_j)\), such as a collection of validation moments, where smaller values indicate better predictive performance on held-out data or validation moments. Then
\[
        \operatorname{CV}(\lambda)
        =\sum_{j=1}^J
        \big \| \, \mathbb{E}_{\Pi^{\mathrm{pred}}_{\lambda,-j}}
        [\mathcal V_j(a_j;x_j,z_j)] \, \big\|,
        \qquad
        \hat\lambda_{\mathrm{CV}}\in\argmin_{\lambda\in\Lambda_G}\operatorname{CV}(\lambda).
\]

The selected value \(\hat\lambda_{\mathrm{CV}}\) is a tuning summary for the reported grid. Since \(J\) is fixed, we should not treat \(\hat\lambda_{\mathrm{CV}}\) as consistently learning some notion of the ``oracle \(\lambda\).'' However, we would generally expect the value selected by cross-validation to be more informative when \(J\) is large and \(\mathcal V_j(a_j;x_j,z_j)\) is an accurate validation criterion for the model. That said, a researcher should report results over a range of \(\lambda\) for robustness purposes.

\subsection{Statistical Inference}\label{sec:testing}

Statistical inference under strongly identified studies is handled in the usual way. Due to Theorem \ref{thm:strong-equivalence}, the QBHM estimator is asymptotically equivalent to the GMM estimator, and so is asymptotically normal. Theorems 3 and 4 of \citet{ChernozhukovHong2003} then apply verbatim; under a generalized information inequality\footnote{See section 4 of \citet{ChernozhukovHong2003} for more information about what this is and how it can be guaranteed to hold. This can be done by appropriately constructing the weighting matrix in a GMM problem.}, draws from the quasi-posterior itself can be used to estimate the asymptotic variance. However, testing under weak identification is more difficult, and requires us to consider the weak limit discussed in Section \ref{sec:main-fixedJ-theory}. If a researcher is concerned that some weakly identified parameter may be part of their hypothesis, they should test using the procedure below, rather than relying on draws from the quasi-posterior.

In that limit, testing a candidate weak value \(a_0\in\calW\) corresponds to \(H_{0,a_0}:m(a_0)=0\): when \(a_0\) is the unique zero of \(m\), this is the usual point null \(\alpha_W^*=a_0\), while more generally it is the moment-validity null that is inverted to form a weak-identification-robust confidence set. In this subsection, \(\Sigma\) denotes the covariance function of the Gaussian weak moment process, not the GMM weight \(W\) used in the objective. Let \(\Sigma_{00}=\Sigma(a_0,a_0)\) be nonsingular, set \(V_{a_0}(\alpha_W)=\Sigma(\alpha_W,a_0)\Sigma_{00}^{-1}\), and define the residual path \(h_{a_0}(\alpha_W)=g(\alpha_W)-V_{a_0}(\alpha_W)g(a_0)\), the component of the process left after the Gaussian projection on \(g(a_0)\).  Under \(H_{0,a_0}\) in the Gaussian weak-GMM limit experiment, \(g(a_0)\sim N(0,\Sigma_{00})\), and the Gaussian projection construction makes \(g(a_0)\) independent of the residual path \(h_{a_0}(\cdot)\).  Hence, conditional on the observed residual path \(h_{a_0}\), the null law of the full process is generated by \(g^*(\alpha_W)=h_{a_0}(\alpha_W)+V_{a_0}(\alpha_W)\xi^*\), where \(\xi^*\sim N(0,\Sigma_{00})\).  This conditional-inference argument follows \citet{AM2016conditional}: the hierarchy can enter the test statistic or the alternative weighting, while the conditional null simulation supplies the critical value. For a statistic \(T_\lambda\), let the randomized conditional critical rule be
\begin{equation}\label{eq:randomized-conditional-test}
        \varphi_{\lambda,a_0}(g)
        =\one[T_\lambda(g)>c_{\zeta,\lambda}(h_{a_0})]
        +\rho_{\zeta,\lambda}(h_{a_0})\one[T_\lambda(g)=c_{\zeta,\lambda}(h_{a_0})],
\end{equation}
where \(c_{\zeta,\lambda}(h)\) and \(\rho_{\zeta,\lambda}(h)\in[0,1]\) are chosen so that the conditional null rejection probability is \(\zeta\).  If the conditional null distribution is continuous at the critical value, the randomization term is immaterial.

\begin{proposition}\label{prop:testing}
Work in the Gaussian weak-GMM limit experiment under $H_{0,a_0}$, suppose \(\Sigma_{00}=\Sigma(a_0,a_0)\) is nonsingular, and let \(T_\lambda\) be measurable for every \(\lambda\in\Lambda\).  Suppose that, on a residual-path set \(\calH_{0,a_0}\) that has probability one under every null nuisance mean,
\begin{equation}\label{eq:conditional-size}
        P_0^h[T_\lambda(g)>c_{\zeta,\lambda}(h)]
        +\rho_{\zeta,\lambda}(h)P_0^h[T_\lambda(g)=c_{\zeta,\lambda}(h)]=\zeta
\end{equation}
for every \(h\in\calH_{0,a_0}\) and every \(\lambda\in\Lambda\), where \(P_0^h\) is the regular conditional null law generated by \(h(\alpha_W)+V_{a_0}(\alpha_W)\xi\), \(\xi\sim N(0,\Sigma_{00})\).  Then \(\varphi_{\lambda,a_0}\) has conditional size \(\zeta\) for every \(\lambda\in\Lambda\), and \(\mathbb{E}_{a_0,m}\varphi_{\lambda,a_0}(g)=\zeta\) for every null nuisance mean with \(m(a_0)=0\).
\end{proposition}

The test above is one which would be valid in the weak limit, not one which is directly implementable. In practice, we would use a feasible version of this test. This uses a covariance-kernel estimator \(\widehat\Sigma_n\) for this inference step, replaces \(g\) and \(\Sigma\) by \(g_n\) and \(\widehat\Sigma_n\), and forms
\[
h_{n,a_0}(\alpha_W)
=
g_n(\alpha_W)
-
\widehat\Sigma_n(\alpha_W,a_0)
\widehat\Sigma_n(a_0,a_0)^{-1}
g_n(a_0).
\]
The conditional simulation uses
\[
g_n^*(\alpha_W)
=
h_{n,a_0}(\alpha_W)
+
\widehat\Sigma_n(\alpha_W,a_0)
\widehat\Sigma_n(a_0,a_0)^{-1}\xi_n^*,
\qquad
\xi_n^*\sim N(0,\widehat\Sigma_n(a_0,a_0)).
\]
Proposition~\ref{prop:feasible-testing} in the appendix records sufficient covariance-consistency, continuity, and no-boundary conditions for this test to have valid asymptotic size. Confidence sets are then obtained by inverting the conditional tests over candidate values \(a_0\).  If a researcher is interested in testing \(B\alpha_W\) for some matrix $B$, rather than $\alpha_W$ itself, a valid confidence set is obtained by projection: if \(C_{n,W}\) is the inverted confidence set for \(\alpha_W\), then \(C_{n,B}:=\{B\alpha_W:\alpha_W\in C_{n,W}\}\) has at least the coverage of \(C_{n,W}\).

The hierarchy can also be used to form weighted-average-power (WAP) objective functions, which average power over alternatives using weights induced by the pooling relation \citep{AM2022}.  A hierarchy-weighted conditional quasi-likelihood-ratio statistic, formed from the weak-GMM quasi-likelihood under the null and over alternatives, puts more weight on alternatives that are close to the pooling relation when those alternatives are substantively central.  Proposition~\ref{prop:testing-wap} in the appendix gives the conditional WAP calculation, and the main reporting distinction is that \(\lambda\) controls point-estimation shrinkage and can weight alternatives for power, while weak-identification-robust confidence sets deliver coverage.

For hypothesis that contain both strongly and weakly identified components, Appendix~\ref{app:mixed-within-study-blocks}, especially Corollary~\ref{cor:mixed-reported-targets}, gives the separation between the root-\(n\) regular rate and the weak component's \(O(1)\), nondegenerate limit.  Standard inference for strongly identified components applies to the strong block, while weak-identification-robust confidence sets apply to the weak block.  A conservative set for a linear estimand on the original, unscaled parameter level, \(R_S\beta^S+R_W\alpha_W\), can be formed by combining a standard strong confidence set \(C_{n,S}(1-\eta_S)\) with a weak-identification-robust confidence set \(C_{n,W}(1-\eta_W)\), where \(\eta_S+\eta_W=\eta\).  Projecting \(\{R_S\beta+R_W\alpha_W:\beta\in C_{n,S}(1-\eta_S),\,\alpha_W\in C_{n,W}(1-\eta_W)\}\) then gives a set with asymptotic coverage at least \(1-\eta\) under the joint strong/weak convergence conditions.

\section{Monte Carlo evidence}\label{sec:mc}

This section studies the finite-sample behavior of the QBHM estimator in an IV production-function design with ten groups.  The design varies two features: first-stage strength and within-group sample size.  Group parameters are drawn from a common distribution, so the pooling relation is informative but not exact.  We compare the QBHM with unpooled GMM, using the same group-specific objective function in both cases.  The results show the expected pattern: when instruments are weak or groups are small, pooling reduces dispersion and lowers average mean squared error (MSE) across sites.  When within-group information is strong, the unpooled and pooled estimators are close.

\subsection{Data-generating process}

The model estimates a Cobb--Douglas production function with capital as the variable input.  The instrument is an exogenous treatment-induced shock to capital, in the spirit of \citet{de2008returns} and \citet{fafchamps2014microenterprise}.  To ground the simulation, the data-generating process uses micro data from eight randomized controlled trials: \citet{CaiSzeidl2024}, \citet{EggerEtAl2022}, \citet{FinkJackMasiye2020}, \citet{crepon2015estimating}, \citet{augsburg2015impacts}, \citet{attanasio2015impacts}, \citet{karlan2019debt}, and \citet{bari2024asset}.  Each study contains a treatment that shifts firm capital, and from these data we use three variables: a study identifier, pre-treatment capital, and treatment status.  For each synthetic group, we sample a source study with replacement, sample observations from that study with replacement, and add lognormal jitter to capital.  This preserves realistic cross-sectional capital distributions while allowing us to choose the number of groups and within-group sample sizes.  Post-treatment capital is generated as:
\begin{align}
\log K_{ji1}
&=
\mu_{\log K_0}
+
\delta_T T_{ji}
+
\delta_P\bigl(\log K_{ji0}-\mu_{\log K_0}\bigr)
+
\nu_{ji},\\
\nu_{ji}&\sim\mathcal N(0,\sigma_\nu^2).
\end{align}
Changing \(\delta_T\) and \(\sigma_\nu^2\) changes the first-stage strength of treatment as an instrument for post-treatment capital, with smaller \(\delta_T\) and larger \(\sigma_\nu^2\) producing weaker instruments.

Each synthetic group \(j\) is assigned production parameters drawn independently from the hierarchical distributions
\begin{align}
\ell_j:=\log a_{1j} &\sim \mathcal{N}\!\left(\mu_{a_1}, \sigma_{a_1}^2\right), \\
\operatorname{logit}(\rho_j) &\sim \mathcal{N}\!\left(\mu_\rho, \sigma_\rho^2\right).
\end{align}
Pre- and post-treatment revenue are generated from \(Y_{jit}=a_{1j}K_{jit}^{\rho_j}\exp(\varepsilon_{jit})\),
where \(\varepsilon_{jit}\) is drawn independently across pre- and post-treatment periods from a Student-\(t\) distribution.  In the moment objective function below, \(Y_{ji}:=Y_{ji1}\) and \(K_{ji}:=K_{ji1}\) denote post-treatment revenue and capital.

\subsection{Estimation}

We apply the following QBHM to the synthetic data.  The within-group stage uses an affine log-linear IV/GMM parameterization of the kind analyzed in the weak-IV example.  Write \(\ell_j:=\log a_{1j}\) and \(\alpha^{\mathrm{CD}}_j=(\ell_j,\rho_j)^\top\), where CD denotes Cobb--Douglas.  The objective function is written so that larger values are better:
\[
q_{2,j,n_j}(\alpha^{\mathrm{CD}}_j)
=
-\frac{n_j}{2}\,
\hat g_j(\alpha^{\mathrm{CD}}_j)^\top
\hat g_j(\alpha^{\mathrm{CD}}_j).
\]
Here \(\hat u_{ji}(\alpha^{\mathrm{CD}}_j)=\log Y_{ji}-\ell_j-\rho_j\log K_{ji}\), \(\hat g_j(\alpha^{\mathrm{CD}}_j)=n_j^{-1}\sum_{i=1}^{n_j}\zeta_{ji}\hat u_{ji}(\alpha^{\mathrm{CD}}_j)\), and \(\zeta_{ji}=(1,T_{ji})^\top\).  The moment is affine in \(\alpha^{\mathrm{CD}}_j\), and reported summaries for \(a_{1j}\) use the transformation \(a_{1j}=\exp(\ell_j)\).

The Gaussian pooling term follows the hierarchical structure of the data-generating process, with pooling on $\ell_j$ for the affine coefficient and on $\operatorname{logit}(\rho_j)$ to keep the production elasticity inside $(0,1)$:
\begin{align}
\ell_j &\sim \mathcal N(\mu_{a_1},\sigma_{a_1}^2),\\
\operatorname{logit}(\rho_j) &\sim \mathcal N(\mu_\rho,\sigma_\rho^2).
\end{align}
The QBHM quasi-posterior is sampled using Hamiltonian Monte Carlo as implemented in Stan, with capital centered and scaled within groups to improve numerical conditioning.

The comparison estimator is the unpooled GMM estimator obtained by maximizing the same group-specific objective function \(q_{2,j,n_j}(\alpha^{\mathrm{CD}}_j)\) group by group.  The transformation \(a_{1j}=\exp(\ell_j)\) is used when reporting results on the original production-function scale.

\subsection{Simulation design and results}

We run simulations that vary instrument strength and group sample size.  The simulation parameters and weakly informative hierarchical priors are reported in Appendix~\ref{app:mc}, Table~\ref{tab:simple_design_parameters}.  The priors are centered on reasonable but false values, and their standard deviations are broad relative to the data-generating process.  To check sensitivity to the prior specification, we rerun the simulations using very diffuse priors centered at zero, with variances more than 150 times the true data-generating-process scale.  The very diffuse priors change which value of \(\lambda\) performs best in this design, but do not change the main comparison with the unpooled estimator; the sensitivity results are reported in Appendix~\ref{app:mc}.

Table~\ref{tab:lowtech_combined_overall} reports the simulation results, with the top panel varying instrument strength.  The ``Very strong'' scenario corresponds to \((\delta_T=0.5,\sigma_\nu=0.4)\) and an average realized first-stage \(F\)-statistic across groups and replications of \(3813.89\).  The ``Weak'' scenario, \((\delta_T=0.06,\sigma_\nu=1.1)\), and the ``Very weak'' scenario, \((\delta_T=0.03,\sigma_\nu=1.4)\), have average \(F\)-statistics equal to \(8.34\) and \(2.06\), respectively.

The table reports average mean squared error for point estimates and simulation coverage of nominal 95\% intervals, using central quasi-posterior intervals for the QBHM and Wald intervals for the unpooled estimator. The overall pattern is consistent with the theory: when instruments are strong, the unpooled and QBHM estimators are close because within-group information dominates the fixed pooling term.  When instruments are weak, the QBHM estimator reduces dispersion enough to lower root mean squared error (RMSE) even after allowing for centering induced by the hierarchy. The sample-size experiments show the same mechanism, with the bottom panel using the ``Very strong'' instrument setting and varying the common group sample size \((n_j)\). When each group contains little information, pooling reduces dispersion and improves average MSE, while the gains from pooling are small when each group is precise. In addition, the results on coverage show the lessons of Section~\ref{sec:testing}: confidence intervals made by drawing from the quasi-posterior have good coverage when instrument strength is strong and group size is relatively large, but coverage gets generally worse when these do not hold. As such, for weak-identification-robust coverage in empirical work, Section~\ref{sec:testing} recommends weak-identification-robust confidence sets. 

The \(\lambda\)-grid gives the corresponding sensitivity check: in the ``Very strong'' scenario with large sample size, the unpooled estimator has lower RMSE at \(\lambda=1\), because strong within-group information leaves little variance for the hierarchy to remove.  As \(\lambda\) decreases, the QBHM estimator moves toward the unpooled estimator.  Full results for the grid are reported in Appendix~\ref{app:mc}, Table~\ref{tab:lowtech_lambda_sweep_overall}.  The \(\lambda\)-grid result is the finite-sample simulation analogue of the theoretical bias--variance tradeoff: more pooling can lower dispersion when within-group information is weak, but it can increase MSE when the hierarchy adds centering bias to a strongly informative group objective function.  In this design, the central QBHM quasi-posterior intervals tend to be conservative as finite-sample simulation summaries, while the unpooled estimator Wald summaries undercover.

\begin{table}[htbp]
\centering
\footnotesize
\setlength{\tabcolsep}{5pt}
\renewcommand{\arraystretch}{1.12}
\caption{Simple scenario estimator performance for varying instrument strength and group size}
\label{tab:lowtech_combined_overall}
\begin{tabular}{@{}C{1.8cm} C{2.6cm} S[table-format=1.3] S[table-format=1.3] S[table-format=1.3] S[table-format=1.3]@{}}
\toprule
\textbf{Parameter} & \textbf{Setting} & {\shortstack{\textbf{RMSE}\\\textbf{QBHM}}} & {\shortstack{\textbf{RMSE}\\\textbf{Unpooled}}} & {\shortstack{\textbf{Coverage}\\\textbf{QBHM}}} & {\shortstack{\textbf{Coverage}\\\textbf{Unpooled}}} \\
\midrule
\multicolumn{6}{@{}l}{\textit{Panel A: Varying instrument strength}} \\
$a_1$ & Very strong & 0.564 & 0.269 & 0.962 & 0.940 \\
      & Weak        & 1.490 & 4.790 & 0.988 & 0.861 \\
      & Very weak   & 1.560 & 4.610 & 0.989 & 0.868 \\
\addlinespace
$\rho$ & Very strong & 0.056 & 0.038 & 0.962 & 0.958 \\
      & Weak        & 0.082 & 0.223 & 0.988 & 0.947 \\
      & Very weak   & 0.081 & 0.232 & 0.990 & 0.945 \\
\midrule
\multicolumn{6}{@{}l}{\textit{Panel B: Varying group size}} \\
1000 & $a_1$    & 0.564 & 0.269 & 0.962 & 0.940 \\
     & $\rho$ & 0.056 & 0.038 & 0.962 & 0.958 \\
\addlinespace
50   & $a_1$    & 1.770 & 3.700 & 0.995 & 0.843 \\
     & $\rho$ & 0.084 & 0.199 & 0.995 & 0.964 \\
\addlinespace
20   & $a_1$    & 1.870 & 5.870 & 0.999 & 0.776 \\
     & $\rho$ & 0.087 & 0.266 & 0.999 & 0.919 \\
\bottomrule
\end{tabular}
\end{table}

\section{Empirical illustration}\label{sec:structural}

We illustrate the QBHM approach using microenterprise production functions from randomized controlled trials in the capital-drop literature.  The empirical exercise uses micro data from four studies: \citet{EggerEtAl2022}, denoted GE; \citet{bari2024asset}, denoted AB; \citet{de2008returns}, denoted RC; and \citet{fafchamps2014microenterprise}, denoted FY.  All four studies analyze microenterprises, include a treatment that induces a capital shock, and record firm capital as a primary or secondary variable.  The purpose is to show how the estimator behaves in a nonlinear grouped objective function with unstable study-specific features.

The model is based on the production-function component of the microenterprise poverty-trap model in \citet{banerjee2019can}.\footnote{\citet{banerjee2019can} is a working paper at the time of writing and the dataset is not yet publicly available.} Entrepreneurs choose between two technologies.  One has diminishing returns to capital, while the other has constant returns but requires a fixed cost to be paid.  This choice is represented by a production function with a Cobb--Douglas component at low levels of capital and a linear component at high levels of capital.  A nonconvex production function is necessary for a poverty trap in this class of models, and the location of the nonconvexity determines both whether a trap exists and how important it is.

Estimating the production function is challenging because the point of nonconvexity is unknown, depends on the parameters of the function, and may lie outside the support of observed capital for some groups.  When the switching threshold lies on the edge of, or outside, the observed support for a group, the model parameters are weakly identified or not identified.  \citet{banerjee2019can} addresses this difficulty in a well-powered experimental setting by assuming the existence of a trap and selecting from a discrete grid of parameter values using a GMM objective function. Because the QBHM approach pools information across studies while preserving study-level heterogeneity, the application makes the theory's shrinkage mechanism visible in a nonlinear moment problem.

The nonconvex production function we estimate is the following:

\begin{align}
y_{1,jit} &= a_{1j} K_{jit}^{\rho_j}, \\
y_{2,jit} &= a_{2j}\!\left(K_{jit} - K_{\mathrm{thresh},j}\right), \\
Y_{jit} &= \max\!\left(y_{1,jit}, y_{2,jit}\right)\exp(\varepsilon_{jit}).
\end{align}
Here \(Y_{jit}\) is observed revenue, \(K_{jit}\) is observed capital, and \(K_{\mathrm{thresh},j}\) is an unobserved fixed cost required to access the higher-return technology.  The QBHM and unpooled estimators estimate the parameter vector \(\alpha^{\mathrm{NC}}_j=(a_{1j},a_{2j},\rho_j,K_{\mathrm{thresh},j})^\top\), where NC denotes nonconvex.

\subsection{Estimation}
The estimation process is the same as in Section~\ref{sec:mc}, except that the residual \(\hat u_{ji}(\alpha^{\mathrm{NC}}_j)\) is computed from the nonconvex production function and the QBHM additionally assumes
\begin{align}
\log a_{2j} &\sim \mathcal{N}\!\left(\mu_{\log a_2}, \sigma_{a_2}^2\right), \\
K_{\mathrm{thresh},j} &\sim \mathcal{N}\!\left(\mu_K, \sigma_K^2\right).
\end{align}
The compact-support theory in Section~\ref{sec:formal-grouped-problem} treats parameter spaces as bounded, so these Gaussian specifications can be read as their restrictions to the finite parameter sets used for estimation; the display reports the underlying centers and scales.

To evaluate finite-sample performance in a setting close to the application, we run Monte Carlo simulations with a data-generating process whose structure and parameters mimic the applied setting.  The details and results of the simulations are reported in Appendix~\ref{app:mc}.  In this application-calibrated simulation design, the QBHM has lower RMSE and better finite-sample simulation interval coverage than the unpooled estimator for all four parameters.

Both estimators recover the proportion of firms below the switching threshold, with RMSE for the trapped proportion ranging from 0.014 to 0.056 and a slight edge to the QBHM.\footnote{The proportion of trapped firms depends on the location of the nonconvexity, which in turn depends on multiple parameter values.  It is possible for two models to have very different parameter estimates while still choosing the same nonconvexity location.}  The result reflects the fact that most simulated groups have a switching threshold inside the observed support, so the threshold is disciplined by observed variation rather than by extrapolation.

We use $\lambda=1$ as the reported pooled reference because the pooling terms are normalized as Gaussian hierarchical log densities.  We also report $\lambda=0$ as the unpooled reference and rerun the analysis with $\lambda=0.5$ and $\lambda=2$ to assess sensitivity.  The sensitivity checks, reported in Appendix~\ref{app:emp}, show that pooling affects estimates relative to $\lambda=0$, while estimates are stable across the positive pooling values $\lambda\in\{0.5,1,2\}$.  The $\lambda$ values index alternative induced prior paths in the sense of Section~\ref{sec:lambda}.

The empirical application does not select \(\lambda\) by cross-validation.  Instead, it uses the prespecified grid \(\{0,0.5,1,2\}\) to make the normalization and sensitivity transparent.  Stable conclusions across the grid are driven mainly by within-study moment information, while conclusions that appear only for larger $\lambda$ rely more heavily on the pooling relation.  When cross-validation is implemented, the CV-selected tuning value should be chosen from the same grid using a prespecified left-out-study score, such as predictive GMM fit or accuracy for the switching threshold.

\subsection{Empirical results}

Both the QBHM and unpooled estimators are used to estimate the production-function parameters in the four study data sets.  Parameter estimates are reported in Table~\ref{tab:real_fullmodel_4strongiv_parameter_estimates} in Appendix~\ref{app:emp}. Using the production-function point estimates, we calculate the implied switching threshold for capital, where \(a_{1j} K^{\rho_j}=a_{2j}\!\left(K - K_{\mathrm{thresh},j}\right)\), and the proportion of firms in the sample operating below this threshold.  Table~\ref{tab:real_fullmodel_4strongiv_switching_thresholds} reports the switching capital threshold and the proportion of firms below it for each of the four studies.

The estimator behavior is consistent with the weak-identification motivation, and we interpret the estimates as evidence of a poverty trap in the observed support only when the fitted nonconvexity switch lies within the study sample.  If all observed firms are fitted with the Cobb--Douglas component, we treat the exercise as finding no within-support switch, rather than as precisely estimating a trap outside the support.

The QBHM estimates a within-support nonconvexity switch, and hence a poverty-trap pattern under the maintained production-function model, in one study, \citet{fafchamps2014microenterprise} (FY).  Approximately 80\% of the sample lies below the switching threshold, and the estimated fixed cost creating the trap is \$446 in purchasing-power-parity terms.  In the other three studies, the QBHM places the switching threshold above or near the edge of the study support. Care is needed when interpreting parameter values that fall outside the study support, as illustrated by the very high GE and AB switching-capital estimates.  We interpret these results as evidence of no estimated switch within the support, not as precise estimates of distant economically meaningful thresholds.  A switching threshold far outside the sample is consistent with the model behavior expected when no within-support switch is present.\footnote{In both studies the QBHM has little direct evidence about the linear-branch parameters of the production function, $(a_2,K_{\mathrm{thresh}})$, so those components are disciplined mainly by the pooling relation and the other studies.}

The comparison with unpooled estimation shows where pooling matters.  The unpooled estimator agrees with the QBHM in the AB and RC studies in detecting no within-support switching threshold, although it assigns the observed support to a different branch of the production function.  The unpooled estimator also estimates a within-support trap in FY, but at a lower capital threshold than the QBHM.  It also estimates a within-support trap in \citet{EggerEtAl2022} (GE), while the QBHM places the switch outside or near the edge of the observed support.  Taken as an implementation exercise, the estimates indicate a trap in FY, no trap in RC or AB, and a less stable conclusion in GE.

Although the model implies that treatment effects should be larger for firms raised above the threshold by the RCT intervention, the present empirical section does not estimate that heterogeneity; its role is to show how QBHM organizes an unstable nonlinear grouped objective function.

\begin{table}[htbp]
\centering
\footnotesize
\setlength{\tabcolsep}{4pt}
\renewcommand{\arraystretch}{1.12}
\caption{Implied switching capital by study and estimator}
\label{tab:real_fullmodel_4strongiv_switching_thresholds}
\begin{tabular}{@{}C{1.4cm} C{1.8cm} S[table-format=3.3] S[table-format=7.3] S[table-format=1.3]@{}}
\toprule
\textbf{Group} & \textbf{Estimator} & {\textbf{$K_{\mathrm{thresh}}$}} & {\textbf{Switching capital}} & {\textbf{Prop. below}} \\
\midrule
GE & QBHM & 435.655 & 582199.232 & 1.000 \\
GE & Unpooled & 71.078 & 146.796 & 0.880 \\
AB & QBHM & 439.799 & 5213516.072 & 1.000 \\
AB & Unpooled & 0.330 & 0.336 & 0.000 \\
FY & QBHM & 446.301 & 2405.249 & 0.809 \\
FY & Unpooled & 16.151 & 269.617 & 0.227 \\
RC & QBHM & 435.031 & 22791.259 & 0.998 \\
RC & Unpooled & 0.038 & 0.743 & 0.000 \\
\bottomrule
\end{tabular}
\end{table}

\section{Conclusion}
\label{sec:conclusion}

This paper develops the QBHM as a way to borrow information across related econometric studies without forcing all studies into the same model.  Each study keeps its own objective function, instruments, controls, weights, fixed effects, and parameter dimension, with the hierarchy only added for components that the researcher believes are meaningfully comparable.  With a fixed number of studies, the goal is not to learn one common cross-study parameter; instead, it is to learn the study-specific estimands while using the hierarchy as a transparent pooling relation.

The results imply that, when a study is already precise, fixed pooling has little first-order effect.  When a study is noisy or weakly informative, pooling can materially affect estimates: it can reduce error if the hierarchy is well centered, but it can add bias if the pooling relation is wrong.  The decision-theoretic result should be read in this operational sense: for weakly identified blocks, the feasible QBHM quasi-posterior mean converges to the weak-GMM Bayes rule under squared loss for the reported estimand. As such, there is a decision-theoretic foundation for recommending the use of the QBHM estimator, even when some studies may be weakly identified. Since $\lambda$ and prior choice may potentially affect results, empirical work should report estimates over a path of pooling strengths, not just one pooled estimate.  A cross-validation choice of \(\lambda\), when used, should be treated as a reference point on that path.  Findings that are stable across the path are mainly driven by within-study information; findings that appear only under strong pooling rely more heavily on the hierarchy. The simulations and application show the intended use: QBHM behaves much like unpooled estimation when the group evidence is strong, and is most useful when some group objectives are noisy, nonlinear, or weakly informative.  Taken together, QBHM gives researchers a disciplined way to adaptively combine information from various studies, without making any strong distributional assumptions. Further applying this approach to substantive questions in economics remains an exciting direction for future research.

\appendix
\section*{Appendix}
\addcontentsline{toc}{section}{Appendix}

The appendix gives the formal proofs, regularity checks, auxiliary calculations, and additional simulation and empirical tables.

\section{Proofs}
\label{app:normal-summary-example}

\begin{proof}[Proof of Proposition~\ref{prop:normal-normal-summary}]
Fix the observed summaries \(\Delta_1,\ldots,\Delta_J\) and fix \(\lambda\ge0\). Write \(a_{j,n}:=\mathcal I_{j,n}/\sigma_{\Delta,j}^2\), \(b_\lambda:=\lambda/\tau_0^2\), and \(s:=s_0^{-2}\). The negative log quasi-posterior kernel, up to an additive constant, is
\[
\mathcal J_N(\alpha,\theta)
=
\sum_{j=1}^J a_{j,n}(\alpha_j-\Delta_j)^2
+b_\lambda\sum_{j=1}^J(\alpha_j-\theta)^2
+s(\theta-g_0)^2 .
\]
For a perturbation \((u_1,\ldots,u_J,v)\), the associated quadratic form is \(\sum_{j=1}^J a_{j,n}u_j^2+b_\lambda\sum_{j=1}^J(u_j-v)^2+sv^2\). Since \(a_{j,n}>0\) for all \(j\) and \(s>0\), this form is zero only at the origin. Hence the precision matrix is positive definite for every \(\lambda\ge0\), including the no-pooling endpoint \(\lambda=0\). The quasi-posterior is therefore a proper multivariate Gaussian law, and its mean is the unique minimizer of \(\mathcal J_N\).

The first-order condition for \(\alpha_j\), holding \(\theta\) fixed, is \(a_{j,n}(\alpha_j-\Delta_j)+b_\lambda(\alpha_j-\theta)=0\). Thus
\[
\alpha_j(\theta)
=
\frac{a_{j,n}\Delta_j+b_\lambda\theta}{a_{j,n}+b_\lambda}
=
c_{j,n}(\lambda)\Delta_j+
\bigl[1-c_{j,n}(\lambda)\bigr]\theta,
\qquad
c_{j,n}(\lambda)=\frac{a_{j,n}}{a_{j,n}+b_\lambda}.
\]
Substituting \(a_{j,n}=\mathcal I_{j,n}/\sigma_{\Delta,j}^2\) and \(b_\lambda=\lambda/\tau_0^2\) gives \(c_{j,n}(\lambda)=\left(1+\lambda\sigma_{\Delta,j}^2/(\mathcal I_{j,n}\tau_0^2)\right)^{-1}\). The first-order condition for \(\theta\) is \(b_\lambda\sum_{j=1}^J(\theta-\alpha_j)+s(\theta-g_0)=0\). At the minimizer, \(\theta-\alpha_j(\theta)=c_{j,n}(\lambda)(\theta-\Delta_j)\), so
\[
\left[s+b_\lambda\sum_{k=1}^J c_{k,n}(\lambda)\right]\theta
=
sg_0+b_\lambda\sum_{k=1}^Jc_{k,n}(\lambda)\Delta_k .
\]
Because the bracketed coefficient is strictly positive when \(s>0\), solving gives \eqref{eq:normal-normal-summary-theta}, and substituting this solution into \(\alpha_j(\theta)\) gives \eqref{eq:normal-normal-summary-alpha}. If \(\lambda>0\), then \(b_\lambda>0\) and \(\sum_{k=1}^Jc_{k,n}(\lambda)>0\). Letting \(s=s_0^{-2}\downarrow0\) in \eqref{eq:normal-normal-summary-theta} gives \eqref{eq:normal-normal-summary-diffuse}. This last step is a fixed-\(\lambda>0\) diffuse-prior limit; it is not a proper-prior definition of \(\tilde\theta_\lambda\) at \(\lambda=0\).
\end{proof}

\medskip

\begin{proof}[Proof of Corollary~\ref{cor:normal-normal-scale}]
By Proposition~\ref{prop:normal-normal-summary}, \(\tilde\alpha_{j,\lambda}-\Delta_j=[1-c_{j,n}(\lambda)](\tilde\theta_\lambda-\Delta_j)\). First control the common-mean path uniformly over \(\Lambda\). With \(b_\lambda=\lambda/\tau_0^2\), \(s=s_0^{-2}\), and \(D_\lambda=s+b_\lambda\sum_{k=1}^Jc_{k,n}(\lambda)\), equation \eqref{eq:normal-normal-summary-theta} can be written as \(\tilde\theta_\lambda=\omega_{0,n}(\lambda)g_0+\sum_{k=1}^J\omega_{k,n}(\lambda)\Delta_k\), where \(\omega_{0,n}(\lambda)=s/D_\lambda\) and \(\omega_{k,n}(\lambda)=b_\lambda c_{k,n}(\lambda)/D_\lambda\). The denominator is strictly positive, the weights are nonnegative, and their sum is one. Hence $\sup_{\lambda\in\Lambda}|\tilde\theta_\lambda|\le |g_0|+\max_{1\le k\le J}|\Delta_k|=O_p(1)$, and therefore \(\sup_{\lambda\in\Lambda}|\tilde\theta_\lambda-\Delta_j|=O_p(1)\). If \(\mathcal I_{j,n}=n_j\), then
\(1-c_{j,n}(\lambda)=\frac{x_{j,n}(\lambda)}{1+x_{j,n}(\lambda)}\), where \(x_{j,n}(\lambda)=\lambda\sigma_{\Delta,j}^2/(n_j\tau_0^2)\).
Since \(\Lambda\subset[0,\bar\lambda]\), \(\sup_{\lambda\in\Lambda}[1-c_{j,n}(\lambda)]\le\bar\lambda\sigma_{\Delta,j}^2/(n_j\tau_0^2)\). Multiplying by the preceding \(O_p(1)\) bound gives \(\sup_{\lambda\in\Lambda}\sqrt{n_j}|\tilde\alpha_{j,\lambda}-\Delta_j|=O_p(n_j^{-1/2})=o_p(1)\). If instead \(\mathcal I_{j,n}\le M<\infty\) along a subsequence and \(\lambda>0\) is fixed, then
\[
1-c_{j,n}(\lambda)
=
\frac{\lambda\sigma_{\Delta,j}^2}{\mathcal I_{j,n}\tau_0^2+\lambda\sigma_{\Delta,j}^2}
\ge
\frac{\lambda\sigma_{\Delta,j}^2}{M\tau_0^2+\lambda\sigma_{\Delta,j}^2}>0
\]
along that subsequence. Thus the shrinkage coefficient does not vanish in the bounded-information problem.
\end{proof}

\medskip

\begin{proof}[Proof of Proposition~\ref{prop:normal-normal-risk}]
All expectations in this proof are under the benchmark law conditional on the fixed center \(\theta\). Thus the integrated risk averages over both the random effect \(u_j\) and the sampling error \(\varepsilon_j\). Write \(v_j:=\sigma_{\Delta,j}^2/\mathcal I_{j,n}\). Since \(\alpha_j=\theta+u_j\) and \(\Delta_j=\alpha_j+\varepsilon_j\),
\(\delta_j(\eta_j)-\alpha_j=\eta_j(\alpha_j+\varepsilon_j)+(1-\eta_j)\theta-\alpha_j=-(1-\eta_j)u_j+\eta_j\varepsilon_j\).
The variables \(u_j\) and \(\varepsilon_j\) are independent and mean zero, with variances \(\tau^2\) and \(v_j\). Hence the cross term has expectation zero and \(R_j(\eta_j)=(1-\eta_j)^2\tau^2+\eta_j^2v_j\). This is a strictly convex quadratic because \(R_j''(\eta_j)=2(\tau^2+v_j)>0\). The first-order condition \(-2(1-\eta_j)\tau^2+2\eta_jv_j=0\) gives the unique minimizer \(\eta_j^\star=\tau^2/(\tau^2+v_j)\). Finally, \(c_{j,n}(\lambda)=1/[1+(\lambda/\tau_0^2)v_j]\), so \(c_{j,n}(\lambda)=\eta_j^\star=1/(1+v_j/\tau^2)\) if and only if \(\lambda/\tau_0^2=1/\tau^2\).
\end{proof}

\medskip

\begin{proof}[Proof of Corollary~\ref{cor:normal-pointwise-mse}]
Here \(\alpha_{j0}\), \(\theta\), and \(\eta\) are fixed, and the expectation is only over the sampling error \(\varepsilon_j\). The estimation error is $\delta_j(\eta)-\alpha_{j0}
=
\eta\varepsilon_j+(1-\eta)(\theta-\alpha_{j0})$. Since \(\mathbb E\varepsilon_j=0\) and \(\mathbb E\varepsilon_j^2=v_j\), the pointwise MSE is \(\eta^2v_j+(1-\eta)^2(\theta-\alpha_{j0})^2\). The unpooled rule corresponds to \(\eta=1\) and has MSE \(v_j\). Therefore
\(\MSE(\delta_j(\eta),\alpha_{j0})-v_j=-(1-\eta^2)v_j+(1-\eta)^2(\theta-\alpha_{j0})^2\).
Strict improvement is exactly negativity of this difference. If \(\eta<1\), then \(1-\eta^2=(1-\eta)(1+\eta)\), and division by \((1-\eta)^2>0\) gives \((\theta-\alpha_{j0})^2<[(1+\eta)/(1-\eta)]v_j\). If \(\eta=1\), the shrinkage and unpooled rules coincide and strict improvement is impossible.
\end{proof}

\medskip

\subsection{Auxiliary results}

We now prove the reductions and limit results used in the main text: the strong-study subsection derives the local quadratic reduction used for strongly identified studies, and the weak-experiment subsection gives the reduced weak-GMM process, the primitive prior-path result, and the study-level weak-marginal reduction. The posterior-continuity subsection then proves the main weak-limit theorem and the posterior-risk identity, and the Andrews--Mikusheva subsection records the diffuse-nuisance-prior interpretation of the continuously updated weak quasi-posterior. The mixed within-study subsection gives the profiled reduction for studies containing both regular and retained weak coordinates. MSE transfer follows from path convergence and uniform integrability, and the final proof subsections give the pointwise MSE tools, local-pooling expansion, weak-IV calculation, weak-identification testing arguments, and auxiliary calculations used in the text.

All expectations without an explicit subscript in a proof are taken under the probability measure identified at the start of that proof. Constants denoted by \(C\) may change from line to line but do not depend on the weak-GMM limit experiment when a uniform statement is being proved.

\subsection{Strong studies}

This subsection derives the local objects used in the strong-study and weak-marginal reductions from Assumptions~\ref{ass:primitive-setup}, \ref{ass:study-level-split}, and \ref{ass:primitive-strong}. Fix \(j\in\mathcal J_S\) and a compact set \(\Lambda\subset[0,\bar\lambda]\). Let \(\omega\) collect all compactly supported variables other than \(\alpha_j\) that are held fixed when the strong study is integrated out; this includes \(\theta\) when conditioning on the hierarchy parameter. For the local chart, set \(\alpha_j=\widehat\alpha_j^{\mathrm{un}}+h/\sqrt{n_j}\). The next lemma records the primitive-to-local reduction needed for both Theorem~\ref{thm:strong-equivalence} and Proposition~\ref{prop:weak-marginal-main}.

\begin{lemma}\label{lem:strong-primitive}
Suppose Assumptions~\ref{ass:primitive-setup}, \ref{ass:study-level-split}, and \ref{ass:primitive-strong} hold. Fix \(j\in\mathcal J_S\), and fix any deterministic sequence \(R_{j,n}\uparrow\infty\) with \(R_{j,n}=o(\sqrt{n_j})\) and \(R_{j,n}^2/\log n_j\to\infty\). There are events whose probabilities tend to one, a measurable local maximizer \(\widehat\alpha_j^{\mathrm{un}}\), and a sequence \(K_{j,n}\uparrow\infty\) with \(2K_{j,n}\le R_{j,n}\) and \(K_{j,n}=o(\sqrt{n_j})\), such that the following statements hold uniformly over retained \(\omega\) and \(\lambda\in\Lambda\). First, \(\sqrt{n_j}\|\widehat\alpha_j^{\mathrm{un}}-\alpha_{j0}\|=O_p(1)\). Second, for \(\|h\|\le K_{j,n}\),
\[
q_{2,j,n_j}(X_j,\widehat\alpha_j^{\mathrm{un}}+h/\sqrt{n_j})
-q_{2,j,n_j}(X_j,\widehat\alpha_j^{\mathrm{un}})
=-\frac12 h'\widehat H_{j,n}h+r^q_{j,n}(h),
\]
where \(\|\widehat H_{j,n}-H_j\|=o_p(1)\), the eigenvalues of \(\widehat H_{j,n}\) are bounded away from zero and infinity with probability approaching one, and \(\sup_{\|h\|\le K_{j,n}}|r^q_{j,n}(h)|=o_p(1)\). Third, writing \(p_j\) for the prior density,
\[
        \sup_{\|h\|\le K_{j,n}}
        \left|
        \frac{p_j(\widehat\alpha_j^{\mathrm{un}}+h/\sqrt{n_j})}
        {p_j(\widehat\alpha_j^{\mathrm{un}})}-1
        \right|=o_p(1),
\]
and
\[
        \sup_{\lambda\in\Lambda,\omega,\|h\|\le K_{j,n}}
        \left|
        \lambda\left[q_1(\widehat\alpha_j^{\mathrm{un}}+h/\sqrt{n_j},\omega;z)-q_1(\widehat\alpha_j^{\mathrm{un}},\omega;z)\right]
        \right|=o_p(1).
\]
Finally, if \(L_{j,n,\lambda}(h,\omega)\) denotes the full conditional integrand in the \(h\)-coordinate after division by its value at \(h=0\), and if \(\calH_{j,n}\) is the transformed support of \(h\), then
\[
\sup_{\lambda\in\Lambda,\omega}
\frac{\int_{\calH_{j,n}\cap\{\|h\|>K_{j,n}\}}(1+\|h\|)L_{j,n,\lambda}(h,\omega)\,dh}
{\int_{\calH_{j,n}\cap\{\|h\|\le K_{j,n}\}}L_{j,n,\lambda}(h,\omega)\,dh}
=o_p(1).
\]
\end{lemma}

\begin{proof}[Proof of Lemma \ref{lem:strong-primitive}]
Fix the group \(j\) and suppress the subscript \(j\) where doing so does not create ambiguity. Let \(n=n_j\), \(\alpha_0=\alpha_{j0}\), and \(S_n=n^{-1/2}\nabla_\alpha q_{2,j,n}(X_j,\alpha_0)\). By Assumption~\ref{ass:primitive-strong}, \(S_n=O_p(1)\). Since \(H_j=-\nabla_{\alpha\alpha}^2M_j(\alpha_0)\) is positive definite and \(H_j(a)=-\nabla_{\alpha\alpha}^2M_j(a)\) is continuous near \(\alpha_0\), reduce \(\delta_j\), if needed, so that \(v'H_j(a)v\ge c\|v\|^2\) on \(\mathcal U_j(\delta_j)\) for some \(c>0\). On the event \(\sup_{a\in\mathcal U_j(\delta_j)}\| -\nabla_{\alpha\alpha}^2\bar q_{2,j,n}(X_j,a)-H_j(a)\|\le c/2\), every unit vector \(v\) satisfies \(v'[-\nabla_{\alpha\alpha}^2\bar q_{2,j,n}(X_j,a)]v\ge c/2\), uniformly in \(a\in\mathcal U_j(\delta_j)\). The derivative convergence in Assumption~\ref{ass:primitive-strong} makes this event have probability approaching one, so the sample criterion is uniformly strictly concave on \(\mathcal U_j(\delta_j)\) on high-probability events.

For each fixed \(M<\infty\), Taylor's theorem at \(\alpha_0\), the derivative convergence, and the score bound give, uniformly over \(\|h\|\le M\),
\(q_{2,j,n}(X_j,\alpha_0+h/\sqrt n)-q_{2,j,n}(X_j,\alpha_0)=S_n'h-\frac12h'H_jh+o_p(1)\).
Fix \(B<\infty\) such that \(\mathbb P(\|S_n\|\le B)>1-\varepsilon/2\) for all large \(n\), and choose \(M\) so large that \(BM-cM^2/2+cM^2/8<0\). With probability at least \(1-\varepsilon\) for all large \(n\), \(\|S_n\|\le B\) and the remainder in the preceding expansion has absolute value at most \(cM^2/8\) on \(\|h\|\le M\). On this event, for \(\|h\|=M\), the increment is at most \(BM-cM^2/2+cM^2/8<0\), while it is zero at \(h=0\). Because the sample objective is continuous on the compact local ball, it has a maximizer in that ball; the negativity on the boundary places at least one maximizer in the interior. The measurable maximum theorem applied to the random continuous objective on the compact ball gives a measurable selection, denoted \(\widehat\alpha_j^{\mathrm{un}}\). Since \(\varepsilon\) is arbitrary, \(\sqrt n\|\widehat\alpha_j^{\mathrm{un}}-\alpha_0\|=O_p(1)\). On the same high-probability events, the selected maximizer is interior and satisfies \(\nabla_\alpha q_{2,j,n}(X_j,\widehat\alpha_j^{\mathrm{un}})=0\).

The preceding construction gives a local maximizer, and we also need a separation bound away from \(\alpha_0\). For \(\varepsilon>0\), let \(A(\varepsilon)=\{a\in\calA_j:\|a-\alpha_0\|\ge\varepsilon\}\). If this set is nonempty, compactness, continuity of \(M_j\), and uniqueness of the maximizer imply \(\Delta(\varepsilon):=M_j(\alpha_0)-\sup_{a\in A(\varepsilon)}M_j(a)>0\). On the event \(\sup_{a\in\calA_j}|\bar q_{2,j,n}(X_j,a)-M_j(a)|\le\Delta(\varepsilon)/4\), it follows that \(\sup_{a\in A(\varepsilon)}\bar q_{2,j,n}(X_j,a)\le\bar q_{2,j,n}(X_j,\alpha_0)-\Delta(\varepsilon)/2\), and hence \(q_{2,j,n}(X_j,\alpha_0)-q_{2,j,n}(X_j,a)\ge n\Delta(\varepsilon)/2\) for every \(a\in A(\varepsilon)\). Since \(q_{2,j,n}(X_j,\widehat\alpha_j^{\mathrm{un}})\ge q_{2,j,n}(X_j,\alpha_0)\) on the local-maximizer event, the same lower bound holds relative to \(\widehat\alpha_j^{\mathrm{un}}\). If \(A(\varepsilon)\) is empty, the corresponding separation statement is vacuous.

Define \(\widehat H_{j,n}:=-\nabla_{\alpha\alpha}^2\bar q_{2,j,n}(X_j,\widehat\alpha_j^{\mathrm{un}})\). The rate of \(\widehat\alpha_j^{\mathrm{un}}\), the local derivative convergence, and continuity of \(H_j(a)\) at \(\alpha_0\) imply \(\|\widehat H_{j,n}-H_j\|=o_p(1)\). Hence the eigenvalues of \(\widehat H_{j,n}\) are bounded away from zero and infinity with probability approaching one. Because \(R_{j,n}=o(\sqrt n)\) and \(\widehat\alpha_j^{\mathrm{un}}-\alpha_0=O_p(n^{-1/2})\), the random set \(\{\widehat\alpha_j^{\mathrm{un}}+u/\sqrt n:\|u\|\le R_{j,n}\}\) is contained in \(\mathcal U_j(\delta_j)\) with probability approaching one. On that event set
\[
\eta_{n}(R_{j,n})
:=
\sup_{\|u\|\le R_{j,n}}
\left\|-
\nabla_{\alpha\alpha}^2\bar q_{2,j,n}(X_j,\widehat\alpha_j^{\mathrm{un}}+u/\sqrt n)
-\widehat H_{j,n}\right\| .
\]
Uniform derivative convergence and continuity of \(H_j(a)\) on the shrinking neighborhood imply \(\eta_n(R_{j,n})=o_p(1)\). Indeed, on the local event,
\[
\eta_n(R_{j,n})
\le
2\sup_{a\in\mathcal U_j(\delta_j)}\left\|-\nabla_{\alpha\alpha}^2\bar q_{2,j,n}(X_j,a)-H_j(a)\right\|
+\sup_{\|u\|\le R_{j,n}}\left\|H_j(\widehat\alpha_j^{\mathrm{un}}+u/\sqrt n)-H_j(\widehat\alpha_j^{\mathrm{un}})\right\|,
\]
and the second term is bounded by the continuity modulus of \(H_j\) at radius \(R_{j,n}/\sqrt n=o(1)\). Choose deterministic \(b_n\downarrow0\) such that \(\mathbb P(\eta_n(R_{j,n})>b_n)\to0\), and set \(a_n=b_n^{-1/4}\) and \(K_{j,n}=\min(R_{j,n}/2,a_n)\). Then \(K_{j,n}\uparrow\infty\), \(2K_{j,n}\le R_{j,n}\), \(K_{j,n}=o(\sqrt n)\), and \(K_{j,n}^2\eta_n(R_{j,n})=o_p(1)\). Taylor's theorem at \(\widehat\alpha_j^{\mathrm{un}}\), using the first-order condition, gives for \(\|h\|\le K_{j,n}\)
\[
q_{2,j,n}(X_j,\widehat\alpha_j^{\mathrm{un}}+h/\sqrt n)
-q_{2,j,n}(X_j,\widehat\alpha_j^{\mathrm{un}})
=-\frac12h'\widehat H_{j,n}h+r^q_{j,n}(h),
\]
with \(\sup_{\|h\|\le K_{j,n}}|r^q_{j,n}(h)|\le K_{j,n}^2\eta_n(R_{j,n})/2=o_p(1)\).

Assumption~\ref{ass:primitive-setup} gives a positive continuous prior density, and Assumption~\ref{ass:primitive-strong} gives continuous differentiability of \(p_j\) on \(\mathcal U_j(\delta_j)\). On the high-probability local event, \(\log p_j\) has bounded derivative on the relevant compact neighborhood. Hence
\[
\sup_{\|h\|\le K_{j,n}}
\left|
\log p_j(\widehat\alpha_j^{\mathrm{un}}+h/\sqrt n)-\log p_j(\widehat\alpha_j^{\mathrm{un}})
\right|
\le C K_{j,n}/\sqrt n=o(1),
\]
which implies the displayed prior-density ratio. The same mean-value argument, using bounded first derivatives of \(q_1\) in the \(\alpha_j\) coordinate uniformly over compact retained variables and \(\lambda\in\Lambda\), gives
\[
\sup_{\lambda,\omega,\|h\|\le K_{j,n}}
\left|\lambda\left[q_1(\widehat\alpha_j^{\mathrm{un}}+h/\sqrt n,\omega;z)-q_1(\widehat\alpha_j^{\mathrm{un}},\omega;z)\right]\right|
\le C K_{j,n}/\sqrt n=o(1).
\]

It remains to prove the first-moment tail ratio. Write \(D^{\mathrm{loc}}_{n,\lambda,\omega}\) for the denominator over \(\calH_{j,n}\cap\{\|h\|\le K_{j,n}\}\). On the high-probability event already fixed, the local expansion, bounded eigenvalues, and bounded prior and hierarchy ratios give \(L_{j,n,\lambda}(h,\omega)\ge C^{-1}\exp[-C\|h\|^2]\) on \(\|h\|\le1\), and therefore \(D^{\mathrm{loc}}_{n,\lambda,\omega}\ge C^{-1}\int_{\|h\|\le1}\exp[-C\|h\|^2]dh=:d_0>0\), uniformly over \(\lambda\) and \(\omega\). For the annulus \(K_{j,n}<\|h\|\le R_{j,n}/2\), the line segment from \(\widehat\alpha_j^{\mathrm{un}}\) to \(\widehat\alpha_j^{\mathrm{un}}+h/\sqrt n\) remains in \(\mathcal U_j(\delta_j)\). Taylor's theorem and the first-order condition give \(q_{2,j,n}(X_j,\widehat\alpha_j^{\mathrm{un}}+h/\sqrt n)-q_{2,j,n}(X_j,\widehat\alpha_j^{\mathrm{un}})=-h'[-\nabla_{\alpha\alpha}^2\bar q_{2,j,n}(X_j,\tilde a)]h/2\le-c\|h\|^2\), after reducing \(c\) if necessary, where \(\tilde a\) lies on the segment. Bounded prior and hierarchy factors then give \(L_{j,n,\lambda}(h,\omega)\le C\exp[-c\|h\|^2]\) on the annulus. Hence the annulus contribution to the numerator divided by \(D^{\mathrm{loc}}_{n,\lambda,\omega}\) is bounded by \((C/d_0)\int_{\|h\|>K_{j,n}}(1+\|h\|)\exp[-c\|h\|^2]dh=o(1)\).

For the outer region \(\|h\|>R_{j,n}/2\), write \(a=\widehat\alpha_j^{\mathrm{un}}+h/\sqrt n\). If \(a\in\mathcal U_j(\delta_j)\), the same local concavity argument gives an objective gap at least \(cR_{j,n}^2/4\). If \(a\notin\mathcal U_j(\delta_j)\), then \(\|a-\alpha_0\|\ge\delta_j/2\) for all large \(n\) on the local event. When that nonlocal region is nonempty, the fixed-distance separation with \(\varepsilon=\delta_j/2\), together with uniform convergence of \(\bar q_{2,j,n}\) to \(M_j\), gives a sample objective gap at least \(n\Delta(\delta_j/2)/2\) relative to \(\alpha_0\). The selected local maximizer has sample objective value at least the value at \(\alpha_0\), so the same lower bound holds relative to \(\widehat\alpha_j^{\mathrm{un}}\). If the nonlocal region is empty, set \(\Delta_*=1\) and interpret the nonlocal bound as vacuous; otherwise set \(\Delta_*=\Delta(\delta_j/2)\). The lower-level objective gap throughout the outer region is therefore at least \(D_n=\min(cR_{j,n}^2/4,n\Delta_*/2)\). Since \(\calA_j\) is compact and \(h=\sqrt n(a-\widehat\alpha_j^{\mathrm{un}})\), the transformed support is contained in a ball of radius \(C\sqrt n\), so \(\int_{\calH_{j,n}}(1+\|h\|)dh\le Cn^{(d_j+1)/2}\). Therefore the outer contribution is bounded by \(Cn^{(d_j+1)/2}\exp[-D_n]\), which is \(o(1)\) because \(R_{j,n}^2/\log n\to\infty\). Combining the denominator bound, the local annulus, and the outer region proves the displayed tail ratio, uniformly over \(\lambda\) and retained \(\omega\).
\end{proof}

\medskip

The proof uses the following deterministic normalization fact, which is stated separately because the same ratio calculation turns a local log-density approximation into a quasi-posterior-mean approximation.

\begin{lemma}\label{lem:local-moment-ratio}
Let \(T_n\) be an index set. For each \(t\in T_n\), let \(I_n(t)\) be symmetric with eigenvalues in \([c,C]\), let \(b_n(t)\in\R^d\) satisfy \(\sup_{t\in T_n}\|b_n(t)\|\le C\), let \(H_n\to\infty\), and let \(\rho_n(h,t)\) satisfy \(\sup_{t,\|h\|\le H_n}|\rho_n(h,t)|\to0\). If \(P_{n,t}\) is the probability law on \(\|h\|\le H_n\) with density proportional to \(\exp[-h'I_n(t)h/2+b_n(t)'h+\rho_n(h,t)]\), then
\[
\sup_{t\in T_n}\left\| \int h\,P_{n,t}(dh)-I_n(t)^{-1}b_n(t) \right\|\to0.
\]
Moreover,
\[
        \sup_{t\in T_n}
        d_{\mathrm{BL}}\!\left(
        P_{n,t},
        N(I_n(t)^{-1}b_n(t),I_n(t)^{-1})
        \right)
        \to0 .
\]
The conclusions are unchanged if all preceding deterministic bounds hold on events whose probabilities tend to one.
\end{lemma}

\begin{proof}[Proof of Lemma \ref{lem:local-moment-ratio}]
Work first with deterministic arrays satisfying the stated bounds. Let \(A_n=\{h:\|h\|\le H_n\}\), \(\mu_n(t)=I_n(t)^{-1}b_n(t)\), and let \(N_t\) denote \(N(\mu_n(t),I_n(t)^{-1})\). The eigenvalue bounds and the uniform bound on \(b_n(t)\) imply \(\sup_t\|\mu_n(t)\|<\infty\) and \(\sup_t\mathbb E_{N_t}\|h\|^2<\infty\). Markov's inequality gives \(\sup_tN_t(A_n^c)\to0\), and Cauchy--Schwarz gives \(\sup_t\mathbb E_{N_t}[\|h\|\one(A_n^c)]\to0\).

Let \(P^0_{n,t}\) be \(N_t\) conditional on \(A_n\). The preceding tail bounds imply \(\sup_t\|\int h\,P^0_{n,t}(dh)-\mu_n(t)\|\to0\). For every bounded Lipschitz test function \(\varphi\) with \(\|\varphi\|_\infty\le1\),
\(\left|\int\varphi\,dP^0_{n,t}-\int\varphi\,dN_t\right|\le \frac{2N_t(A_n^c)}{1-N_t(A_n^c)}\)
for all large \(n\), uniformly in \(t\). Hence \(\sup_td_{\mathrm{BL}}(P^0_{n,t},N_t)\to0\).

Write \(R_{n,t}(h)=\exp[\rho_n(h,t)]\) and \(\delta_n=\sup_{t,h\in A_n}|\rho_n(h,t)|\). Relative to \(P^0_{n,t}\), the law \(P_{n,t}\) has density \(R_{n,t}(h)/E_{P^0_{n,t}}R_{n,t}\). Since \(e^{-\delta_n}\le E_{P^0_{n,t}}R_{n,t}\le e^{\delta_n}\),
\(\sup_{t,h\in A_n}\left|\frac{R_{n,t}(h)}{E_{P^0_{n,t}}R_{n,t}}-1\right|\le e^{2\delta_n}-1=o(1)\).
Consequently \(\|P_{n,t}-P^0_{n,t}\|_{TV}\le\int\left|R_{n,t}/E_{P^0_{n,t}}R_{n,t}-1\right|dP^0_{n,t}\le e^{2\delta_n}-1\), uniformly in \(t\). Hence \(\sup_t\|P_{n,t}-P^0_{n,t}\|_{TV}\to0\) and, more specifically,
\[
\sup_t\left\|\int h\,P_{n,t}(dh)-\int h\,P^0_{n,t}(dh)\right\|
\le
(e^{2\delta_n}-1)\sup_t\int\|h\|\,P^0_{n,t}(dh)=o(1).
\]
The first-moment conclusion follows by combining this display with the truncated-normal mean bound. Total variation controls bounded-Lipschitz distance, so the bounded-Lipschitz conclusion follows by the triangle inequality. The stochastic version applies the deterministic argument on events whose probabilities tend to one.
\end{proof}

\medskip

\begin{proof}[Proof of Theorem \ref{thm:strong-equivalence}]
Work on the high-probability events supplied by Lemma~\ref{lem:strong-primitive}. Conditional on retained variables \(\omega\), and on \(\theta\) as well when \(\theta\) is not already included in \(\omega\), use the local coordinate \(h=\sqrt{n_j}(\alpha_j-\widehat\alpha_j^{\mathrm{un}})\). The Jacobian factor \(n_j^{-d_j/2}\) is common to numerator and denominator of every conditional law in this coordinate and therefore cancels.

On \(\|h\|\le K_{j,n}\), Lemma~\ref{lem:strong-primitive} gives the conditional log density after subtracting its value at zero as
\[
        -\frac12h'\widehat H_{j,n}h+\rho_{n,\lambda}(h,\omega),
        \qquad
        \sup_{\lambda,\omega,\|h\|\le K_{j,n}}|\rho_{n,\lambda}(h,\omega)|=o_p(1),
\]
where \(\rho_{n,\lambda}\) consists of the quadratic remainder, the logarithm of the prior-density ratio, and the fixed-\(\lambda\) hierarchy increment. Lemma~\ref{lem:local-moment-ratio}, applied with \(b_n=0\), \(I_n=\widehat H_{j,n}\), and \(H_n=K_{j,n}\), gives a local conditional first moment \(o_p(1)\) and bounded-Lipschitz distance \(o_p(1)\) from \(N(0,\widehat H_{j,n}^{-1})\), uniformly in \(\lambda\) and \(\omega\). Because \(\widehat H_{j,n}\to_p H_j\) and the eigenvalues are uniformly bounded away from zero and infinity, the Gaussian laws \(N(0,\widehat H_{j,n}^{-1})\) converge to \(N(0,H_j^{-1})\) in bounded-Lipschitz distance.

The first-moment tail bound in Lemma~\ref{lem:strong-primitive} implies that replacing the truncated conditional law by the full conditional law changes bounded-Lipschitz expectations and first moments by \(o_p(1)\), uniformly over \(\lambda\) and \(\omega\). To see the normalization step, write \(P^{\mathrm{loc}}_{n,\lambda,\omega}\) and \(P^{\mathrm{full}}_{n,\lambda,\omega}\) for the truncated and full conditional laws, and let \(r_{n,\lambda,\omega}\) be the tail-to-local ratio with integrand \((1+\|h\|)L_{j,n,\lambda}(h,\omega)\). Then \(\sup_{\lambda,\omega}r_{n,\lambda,\omega}=o_p(1)\), and, for every \(\varphi\) with \(\|\varphi\|_\infty\le1\), \(|\int\varphi\,dP^{\mathrm{full}}_{n,\lambda,\omega}-\int\varphi\,dP^{\mathrm{loc}}_{n,\lambda,\omega}|\le2r_{n,\lambda,\omega}/(1+r_{n,\lambda,\omega})\). Also \(\|\int h\,dP^{\mathrm{full}}_{n,\lambda,\omega}\|\le \|\int h\,dP^{\mathrm{loc}}_{n,\lambda,\omega}\|+r_{n,\lambda,\omega}\{1+\|\int h\,dP^{\mathrm{loc}}_{n,\lambda,\omega}\|\}\), so the local first-moment bound transfers to the full conditional law. Since the conditional error bounds are uniform in the retained variables, integrating them with respect to the conditional quasi-posterior law of those variables preserves the same bounds for the marginal law \(P^S_{j,n,\lambda}\). This proves the marginal Gaussian approximation and the marginal first-moment bound for the full QBHM quasi-posterior. The unpooled law is the special case \(\lambda=0\) with no retained hierarchy term, so the same argument gives the two unpooled conclusions.

Finally,
\(\sqrt{n_j}(\tilde\alpha^{\mathrm{full}}_{j,n,\lambda}-\widehat\alpha_j^{\mathrm{un}})=\int h\,P^S_{j,n,\lambda}(dh)\).
The uniform first-moment bound gives \(\sup_{\lambda\in\Lambda}\sqrt{n_j}\|\tilde\alpha^{\mathrm{full}}_{j,n,\lambda}-\widehat\alpha_j^{\mathrm{un}}\|=o_p(1)\), which is equivalent to the displayed centered statement.
\end{proof}

\medskip

The following deterministic Laplace tools are used in the study-level weak-marginal proof and in the mixed within-study profiling argument.

\begin{lemma}\label{lem:deterministic-laplace}
Let \(T_n\) be an index set. Suppose \(\sup_{t\in T_n}\|\ell_n(t)\|\le C\), the eigenvalues of \(I_n(t)\) lie in \([c,C]\), \(H_n\to\infty\), \(a_n(0,t)\in[c,C]\),
\[
\sup_{t,\|h\|\le H_n}|r_n(h,t)|\to0,
\qquad
\sup_{t,\|h\|\le H_n}\left|\frac{a_n(h,t)}{a_n(0,t)}-1\right|\to0.
\]
Then, uniformly over \(t\in T_n\),
\[
\begin{aligned}
&\int_{\|h\|\le H_n} a_n(h,t)
        \exp[\ell_n(t)'h-h'I_n(t)h/2+r_n(h,t)]\,dh \\
&\quad =a_n(0,t)(2\pi)^{d/2}\det(I_n(t))^{-1/2}
        \exp[\ell_n(t)'I_n(t)^{-1}\ell_n(t)/2][1+o(1)].
\end{aligned}
\]
The same conclusion holds in probability when the preceding bounds hold on events whose probabilities tend to one.
\end{lemma}

\begin{proof}[Proof of Lemma \ref{lem:deterministic-laplace}]
First consider a deterministic sequence satisfying the stated bounds. Write $\mu_n(t)=I_n(t)^{-1}\ell_n(t)$, $A_n=\{h:\|h\|\le H_n\}$, $\delta_n=\sup_{t,h\in A_n}|r_n(h,t)|$, and $\varepsilon_n=\sup_{t,h\in A_n}|a_n(h,t)/a_n(0,t)-1|$. Completing the square gives
\(\ell_n(t)'h-\frac12h'I_n(t)h=-\frac12(h-\mu_n(t))'I_n(t)(h-\mu_n(t))+\frac12\ell_n(t)'I_n(t)^{-1}\ell_n(t)\).
The eigenvalue and bounded-linear-term assumptions imply \(\|\mu_n(t)\|\le C/c\) and \(\operatorname{tr}(I_n(t)^{-1})\le d/c\), so, for \(N_t=N(\mu_n(t),I_n(t)^{-1})\), \(\sup_t\mathbb E_{N_t}\|h\|^2\le 2(C/c)^2+2d/c<\infty\). Markov's inequality then gives \(\sup_tN_t(A_n^c)\le H_n^{-2}\sup_t\mathbb E_{N_t}\|h\|^2\to0\).

Let $G_n(t)$ denote the full Gaussian integral over $\mathbb R^d$ of $\exp[\ell_n(t)'h-h'I_n(t)h/2]$. This gives
\(G_n(t)=(2\pi)^{d/2}\det(I_n(t))^{-1/2}\exp[\ell_n(t)'I_n(t)^{-1}\ell_n(t)/2]\).
The Gaussian integral over $A_n$ equals $G_n(t)N_t(A_n)$. On $A_n$, the amplitude and remainder bounds give
\((1-\varepsilon_n)e^{-\delta_n}a_n(0,t)\le a_n(h,t)e^{r_n(h,t)}\le(1+\varepsilon_n)e^{\delta_n}a_n(0,t)\),
for all large $n$, uniformly in $t$. Therefore the target integral is trapped between the same two multiplicative factors times $a_n(0,t)G_n(t)N_t(A_n)$. Since $\varepsilon_n\to0$, $\delta_n\to0$, and $\sup_t|N_t(A_n)-1|\to0$, the expansion follows uniformly in $t$. The stochastic version is obtained by applying this deterministic calculation on events whose probabilities tend to one.
\end{proof}

\medskip

For the generic Laplace reduction, write the original local parameter as \((\beta,\gamma)\), where \(\gamma\in\Gamma\) is the retained weak component and \(\beta=\beta_0(\gamma)+r_n^{-1}h\) collects the strong component, with \(r_n\to\infty\). Here \(\gamma\) is a generic retained weak coordinate, playing the role of \(\alpha_W\) after any strong coordinates have been profiled or integrated out; \(\beta\) is the strong coordinate removed by the Laplace approximation. Suppose that, on events whose probabilities tend to one, there is a sequence \(K_n\to\infty\) such that, uniformly over \(\gamma\), \(\theta\), \(\lambda\in\Lambda\), and \(\|h\|\le K_n\),
\[
\begin{aligned}
&q_{2,n}((\beta_0(\gamma)+r_n^{-1}h,\gamma);x)
   +\lambda q_1((\beta_0(\gamma)+r_n^{-1}h,\gamma),\theta;z) \\
&\qquad =q^{\mathrm{red}}_{2,n}(\gamma)+\lambda q^{\mathrm{red}}_1(\gamma,\theta;z)
   +\ell_n(\gamma)'h-\frac12h'I_n(\gamma)h+r_{n,\lambda}(h,\gamma,\theta),
\end{aligned}
\]
where \(\sup_{\lambda,\gamma,\theta,\|h\|\le K_n}|r_{n,\lambda}(h,\gamma,\theta)|=o_p(1)\), \(\sup_\gamma\|\ell_n(\gamma)\|=O_p(1)\), and the eigenvalues of \(I_n(\gamma)\) are bounded away from zero and infinity uniformly in \(\gamma\). Suppose also that the conditional prior-density ratio in the \(\beta\) component is \(1+o_p(1)\) on the same local set, and that the contribution of \(\|h\|>K_n\) and of the complement of the product local chart to the normalized original quasi-posterior is \(o_p(1)\), uniformly over \(\lambda\in\Lambda\). Define
\[
F_n^{\mathrm{Lap}}(\gamma)
=
 p_{\beta\mid\gamma}(\beta_0(\gamma)\mid \gamma)\det(I_n(\gamma))^{-1/2}
 \exp[\ell_n(\gamma)'I_n(\gamma)^{-1}\ell_n(\gamma)/2].
\]

\begin{lemma}\label{lem:weak-laplace-reduction}
In the generic local-reduction setting just described, suppose the displayed local expansion, conditional-prior ratio, uniform curvature, bounded-score, and tail conditions hold uniformly in \((\gamma,\theta,\lambda)\). Let \(P^{\mathrm{marg}}_{n,\lambda}\) be the marginal quasi-posterior for \((\gamma,\theta)\) obtained from the original quasi-posterior after integrating out \(\beta\). Let \(R_{n,\lambda}\) be the probability law with unnormalized measure
\[
\exp[q^{\mathrm{red}}_{2,n}(\gamma)+\lambda q^{\mathrm{red}}_1(\gamma,\theta;z)]
F_n^{\mathrm{Lap}}(\gamma)\,\pi_\Gamma(d\gamma)\pi_\theta(d\theta).
\]
Then \(\sup_{\lambda\in\Lambda}\|P^{\mathrm{marg}}_{n,\lambda}-R_{n,\lambda}\|_{TV}=o_p(1)\).
\end{lemma}

\begin{proof}[Proof of Lemma \ref{lem:weak-laplace-reduction}]
Work on events whose probabilities tend to one and on which all local-reduction bounds hold uniformly. For fixed \((\gamma,\theta)\), change variables from \(\beta\) to \(h\) by \(\beta=\beta_0(\gamma)+r_n^{-1}h\). The Jacobian contributes \(r_n^{-d_\beta}\), which is common to all \((\gamma,\theta)\). On \(\|h\|\le K_n\), the conditional integrand equals
\[
\begin{aligned}
&r_n^{-d_\beta}\exp[q^{\mathrm{red}}_{2,n}(\gamma)+\lambda q^{\mathrm{red}}_1(\gamma,\theta;z)]
        p_{\beta\mid\gamma}(\beta_0(\gamma)\mid\gamma) \\
&\quad\times
\exp[\ell_n(\gamma)'h-h'I_n(\gamma)h/2+r_{n,\lambda}(h,\gamma,\theta)]
        a_n(h,\gamma,\theta),
\end{aligned}
\]
where \(a_n\) is the conditional prior-density ratio and \(\sup_{\gamma,\theta,\|h\|\le K_n}|a_n(h,\gamma,\theta)-1|=o_p(1)\). Lemma~\ref{lem:deterministic-laplace}, with index \(t=(\gamma,\theta,\lambda)\), gives the local \(\beta\)-integral uniformly in \((\gamma,\theta,\lambda)\). The resulting local marginal density is
\[
        c_n\exp[q^{\mathrm{red}}_{2,n}(\gamma)+\lambda q^{\mathrm{red}}_1(\gamma,\theta;z)]
        F_n^{\mathrm{Lap}}(\gamma)[1+\eta_{n,\lambda}(\gamma,\theta)],
\]
where \(c_n=r_n^{-d_\beta}(2\pi)^{d_\beta/2}\) is common to all \((\gamma,\theta)\) and \(\sup_{\gamma,\theta,\lambda}|\eta_{n,\lambda}(\gamma,\theta)|=o_p(1)\). The constant \(c_n\) cancels from normalized probabilities. If \(B^{\mathrm{loc}}_{n,\lambda}\) and \(B^{\mathrm{tail}}_{n,\lambda}\) are the integrated local and tail masses before normalization, the assumed local-reduction tail bound gives \(\sup_\lambda B^{\mathrm{tail}}_{n,\lambda}/B^{\mathrm{loc}}_{n,\lambda}=o_p(1)\). Therefore the normalized tail mass is \(B^{\mathrm{tail}}_{n,\lambda}/(B^{\mathrm{loc}}_{n,\lambda}+B^{\mathrm{tail}}_{n,\lambda})\le B^{\mathrm{tail}}_{n,\lambda}/B^{\mathrm{loc}}_{n,\lambda}=o_p(1)\), uniformly in \(\lambda\).

It remains to remove the relative error \(1+\eta_{n,\lambda}\). Let \(Q_{n,\lambda}\) be the normalized law obtained by setting \(\eta_{n,\lambda}=0\), and let \(Q^{\eta}_{n,\lambda}\) be the normalized local law with the factor \(1+\eta_{n,\lambda}\). On events where \(\delta_n:=\sup_{\gamma,\theta,\lambda}|\eta_{n,\lambda}(\gamma,\theta)|<1/2\), write \(\bar\eta_{n,\lambda}=\int\eta_{n,\lambda}\,dQ_{n,\lambda}\). Then \(|\bar\eta_{n,\lambda}|\le\delta_n\), and for every measurable set \(A\),
\[
\left|Q^\eta_{n,\lambda}(A)-Q_{n,\lambda}(A)\right|
\le
\int_A\left|\frac{1+\eta_{n,\lambda}}{1+\bar\eta_{n,\lambda}}-1\right|dQ_{n,\lambda}
\le
\frac{2\delta_n}{1-\delta_n}.
\]
The bound is uniform in \(\lambda\). Combining this relative-density comparison with the \(o_p(1)\) mass outside the local chart proves the total-variation claim.
\end{proof}

\medskip

\subsection{Weak-experiment and prior primitives}

The results below verify one DQM primitive route to the reduced weak-GMM experiment and the hierarchy-induced prior path. The next lemmas record sufficient conditions for the empirical-process clauses in Assumption~\ref{ass:primitive-weak}; they are not additional assumptions.

\begin{lemma}\label{lem:primitive-donsker-route}
Suppose \(\Gamma_j\subset\mathbb R^{d_{\gamma j}}\) is compact, \(\gamma_j\mapsto\varphi_j(x,\gamma_j)\) is almost surely continuous, the class has a square-integrable envelope, and, for some \(\delta>0\), \(\|\varphi_j(X_{ji},\gamma_j)-\varphi_j(X_{ji},\tilde\gamma_j)\|\le L_j(X_{ji})\|\gamma_j-\tilde\gamma_j\|\), \(\mathbb E_{P_{j0}}L_j^{2+\delta}<\infty\), and \(\sup_{\gamma_j\in\Gamma_j}\mathbb E_{P_{j0}}\|\varphi_j(X_{ji},\gamma_j)\|^{2+\delta}<\infty\). Then \(\{\varphi_j(\cdot,\gamma_j):\gamma_j\in\Gamma_j\}\) is \(P_{j0}\)-Donsker, the map \(\gamma_j\mapsto\varphi_j(\cdot,\gamma_j)\) is continuous in \(L_2(P_{j0})\), and the Gaussian limit has a version with continuous sample paths.
\end{lemma}

\begin{proof}[Proof of Lemma \ref{lem:primitive-donsker-route}]
Work coordinate by coordinate and let \(F\) denote the coordinate envelope. The Lipschitz condition implies \(\|\varphi_{j\ell}(\cdot,\gamma)-\varphi_{j\ell}(\cdot,\tilde\gamma)\|_{P_{j0},2}\le\|L_j\|_{P_{j0},2}\|\gamma-\tilde\gamma\|\). For any Euclidean \(\eta\)-net of the compact set \(\Gamma_j\), the lower and upper functions \(\varphi_{j\ell}(\cdot,\gamma_k)\pm\eta L_j\) bracket every coordinate function assigned to \(\gamma_k\), and the \(L_2(P_{j0})\) width of each bracket is at most \(2\eta\|L_j\|_{P_{j0},2}\). After rescaling \(\eta\), the covering number of a compact subset of \(\mathbb R^{d_{\gamma j}}\) is bounded by \(C\eta^{-d_{\gamma j}}\), so the bracketing entropy integral \(\int_0^1\sqrt{\log N_{[]}(\rho,\mathcal F_{j\ell},L_2(P_{j0}))}\,d\rho\) is finite. The bracketing central limit theorem gives the Donsker property for each coordinate, and finite coordinate stacking gives the vector result. The same Lipschitz inequality gives \(L_2(P_{j0})\)-continuity. The canonical semimetric of the Gaussian limit is bounded by a constant multiple of the \(L_2(P_{j0})\) distance, and the same entropy bound gives a version with continuous sample paths on compact \(\Gamma_j\).
\end{proof}

\medskip

\begin{lemma}\label{lem:dqm-drift}
Fix group \(j\), and let \(\varphi_j:\mathcal X_j\times\Gamma_j\to\R^{q_j}\). Suppose \(\mathbb E_{P_{j0}}\varphi_j(X_{ji},\gamma_j)=0\) on compact \(\Gamma_j\), the group moment class is \(P_{j0}\)-Donsker with a continuous Gaussian version, \(\gamma_j\mapsto\varphi_j(\cdot,\gamma_j)\) is continuous into \(L_2(P_{j0})\), and the local alternatives satisfy Definition~\ref{def:dqm-path} with score \(f_j\in L_2(P_{j0})\). Define \(m_j(\gamma_j)=\mathbb E_{P_{j0}}[f_j(X_{ji})\varphi_j(X_{ji},\gamma_j)]\). Then \(m_j\) is continuous and \(n_j^{-1/2}\sum_{i=1}^{n_j}\varphi_j(X_{ji},\cdot) \dto \mathbb G_j(\cdot)+m_j(\cdot)\) in \(\ell^\infty(\Gamma_j,\R^{q_j})\) under \(P_{j,n_j,f_j}^{n_j}\), where \(\mathbb G_j\) is the reference-law Gaussian limit with covariance function \(\Sigma_j(\gamma_j,\tilde\gamma_j)=\mathbb E_{P_{j0}}[\varphi_j(X_{ji},\gamma_j)\varphi_j(X_{ji},\tilde\gamma_j)']\).
\end{lemma}

\begin{proof}[Proof of Lemma \ref{lem:dqm-drift}]
Write \(\mathbb G_{j,n}(\gamma_j)=n_j^{-1/2}\sum_{i=1}^{n_j}\varphi_j(X_{ji},\gamma_j)\) and \(\Delta_{j,n}=n_j^{-1/2}\sum_{i=1}^{n_j}f_j(X_{ji})\). If \(\gamma_{j,r}\to\gamma_j\), Cauchy--Schwarz gives
\[
\|m_j(\gamma_{j,r})-m_j(\gamma_j)\|
\le
\|f_j\|_{P_{j0},2}\,
\|\varphi_j(\cdot,\gamma_{j,r})-\varphi_j(\cdot,\gamma_j)\|_{P_{j0},2}\to0,
\]
so \(m_j\) is continuous. Adding the single square-integrable function \(f_j\) to the Donsker class preserves Donsker convergence. Under \(P_{j0}^{n_j}\), therefore, \((\mathbb G_{j,n},\Delta_{j,n})\) converges jointly in \(\ell^\infty(\Gamma_j,\R^{q_j})\times\R\) to a tight Gaussian vector \((\mathbb G_j,\Delta_j)\). The reference moments and the score have mean zero, and the cross-covariance is
\(\operatorname{Cov}(\mathbb G_j(\gamma_j),\Delta_j)=\mathbb E_{P_{j0}}[f_j(X_{ji})\varphi_j(X_{ji},\gamma_j)]=m_j(\gamma_j)\).

The DQM expansion in Definition~\ref{def:dqm-path} implies the local likelihood-ratio expansion
\(\log \frac{dP_{j,n_j,f_j}^{n_j}}{dP_{j0}^{n_j}}=\Delta_{j,n}-\frac12\|f_j\|_{P_{j0},2}^2+o_{P_{j0}}(1)\).
The right side converges under \(P_{j0}^{n_j}\) to \(\Delta_j-\|f_j\|_{P_{j0},2}^2/2\), and the limiting exponential has mean one. Le Cam's first lemma gives contiguity of \(P_{j,n_j,f_j}^{n_j}\) with respect to \(P_{j0}^{n_j}\). Le Cam's third lemma then shifts the reference-law Gaussian process by its covariance with \(\Delta_j\). Hence \(\mathbb G_{j,n}\dto\mathbb G_j+m_j\) in \(\ell^\infty(\Gamma_j,\R^{q_j})\) under \(P_{j,n_j,f_j}^{n_j}\). The continuity of \(m_j\) and the continuous version of \(\mathbb G_j\) give a continuous shifted limit.
\end{proof}

\medskip

\begin{proof}[Proof of Proposition \ref{prop:weak-process-main}]
For each \(j\in\mathcal J_W\), Assumption~\ref{ass:primitive-weak}\textup{(ii)} gives \(\mathbb E_{P_{j0}}\phi_j(X_{ji},\alpha_j)=0\) for all retained values \(\alpha_j\in\calW_j\). Assumption~\ref{ass:primitive-weak}\textup{(iii)} gives the Donsker, \(L_2(P_{j0})\)-continuity, and continuous-Gaussian-version conditions required by Lemma~\ref{lem:dqm-drift}. Combining these conditions with the DQM path in Assumption~\ref{ass:primitive-weak}\textup{(i)}, and applying Lemma~\ref{lem:dqm-drift} with \(\Gamma_j=\calW_j\) and \(\varphi_j=\phi_j\), gives
\[
        g_{j,n}(\cdot)
        =n_j^{-1/2}\sum_{i=1}^{n_j}\phi_j(X_{ji},\cdot)
        \dto
        \mathbb G_j(\cdot)+m_j(\cdot)
        \quad\text{in }\ell^\infty(\calW_j,\R^{k_j})
\]
under \(P_{j,n_j,f_j}^{n_j}\), where \(m_j(\alpha_j)=\mathbb E_{P_{j0}}[f_j(X_{ji})\phi_j(X_{ji},\alpha_j)]\). The continuity of \(m_j\) follows from the same Cauchy--Schwarz argument used in Lemma~\ref{lem:dqm-drift}.

Because the weak-study samples are independent under the product local law \(P^W_{n,f}\) and \(J\) is fixed, the finite product of the group convergences gives \((g_{j,n})_{j\in\mathcal J_W}\dto(\mathbb G_j+m_j)_{j\in\mathcal J_W}\) in \(\prod_{j\in\mathcal J_W}\ell^\infty(\calW_j,\R^{k_j})\). Work with the separable continuous versions of the limits. Define the restriction-and-stacking map \(R\) by \(R((x_j)_{j\in\mathcal J_W})(\alpha_W)=(x_j(\alpha_j)')_{j\in\mathcal J_W}'\). For two product elements \(x\) and \(y\),
\(\|R(x)-R(y)\|_{\ell^\infty(\calW)}\le\sum_{j\in\mathcal J_W}\|x_j-y_j\|_{\ell^\infty(\calW_j)}\),
so \(R\) is continuous. The continuous mapping theorem gives \(g_n\dto g=m+\mathbb G\) in \(\ell^\infty(\calW,\R^{k_W})\), and the limit has continuous sample paths on compact \(\calW\).

Assumption~\ref{ass:primitive-weak}\textup{(iv)} gives \(\widehat W_n\pto W\) in \(\ell^\infty(\calW,\R^{k_W\times k_W})\). Slutsky's theorem gives joint convergence of \((g_n,\widehat W_n)\) in the product sup-norm space. Define \(\Phi(x,V)(\alpha_W)=x(\alpha_W)'V(\alpha_W)x(\alpha_W)\). If \((x_r,V_r)\to(x,V)\) in sup norm, then \(\sup_r\|x_r\|_\infty<\infty\) and \(\sup_r\|V_r\|_\infty<\infty\) for all large \(r\), and
\[
\begin{aligned}
&\sup_{\alpha_W\in\calW}
|x_r(\alpha_W)'V_r(\alpha_W)x_r(\alpha_W)-x(\alpha_W)'V(\alpha_W)x(\alpha_W)| \\
&\quad\le
\|V_r\|_\infty\|x_r-x\|_\infty(\|x_r\|_\infty+\|x\|_\infty)
+\|x\|_\infty^2\|V_r-V\|_\infty,
\end{aligned}
\]
which tends to zero. Thus \(\Phi\) is continuous at \((g,W)\) almost surely, and the continuous mapping theorem gives \(Q_n=\Phi(g_n,\widehat W_n)\dto\Phi(g,W)=Q\).
\end{proof}

\medskip

The next lemma verifies that a hierarchy-induced prior path is proper and weakly continuous on a compact weak parameter space. The final sentence gives the prior score used in the local MSE derivative calculation in Section~\ref{sec:mse-improvement} when \(p_\lambda(h)\propto p_0(h)\exp[\lambda \tilde\ell_0(h)]\) on a compact local chart.

\begin{lemma}\label{lem:primitive-prior}
Suppose \(\Gamma_0\) is compact, \(\mu_0\) is finite with full support, and \(w_\lambda(\gamma)\) is strictly positive and jointly continuous on \(\Gamma_0\times[0,\bar\lambda]\). Define \(\pi_\lambda(d\gamma)\propto w_\lambda(\gamma)\mu_0(d\gamma)\). Then \((\pi_\lambda)_{0\le\lambda\le\bar\lambda}\) is proper and weakly continuous on every compact \(\Lambda\subset[0,\bar\lambda]\). The quadratic path \(p_\lambda(\gamma)\propto\exp[\lambda a_P(\gamma)]p_0(\gamma)\) satisfies the same conclusion on compact \(\Gamma_0\) when \(a_P\) is continuous and \(p_0\) is continuous, bounded, and bounded away from zero. If \(\calH\) is compact, \(\tilde\ell_0\) is bounded and measurable, and \(p_0\) is bounded above and bounded away from zero on \(\calH\), then \(p_\lambda(h)\propto p_0(h)\exp[\lambda \tilde\ell_0(h)]\) satisfies Assumption~\ref{ass:smooth-prior} with prior score \(\ell_0\) equal to \(\tilde\ell_0\) minus its \(p_0\)-mean at \(\lambda=0\).
\end{lemma}

\begin{proof}[Proof of Lemma \ref{lem:primitive-prior}]
Since \(\Gamma_0\times[0,\bar\lambda]\) is compact, joint continuity and strict positivity imply that \(w_\lambda(\gamma)\) is bounded above and bounded away from zero uniformly on this product set. The normalizing constant \(Z_\lambda=\int_{\Gamma_0}w_\lambda(u)\mu_0(du)\) is finite and strictly positive because \(\mu_0\) is finite with full support. For each \(f\in C_b(\Gamma_0)\), dominated convergence gives continuity in \(\lambda\) of both \(Z_\lambda\) and \(\int f(\gamma)w_\lambda(\gamma)\mu_0(d\gamma)\). Therefore \(\lambda\mapsto\int f(\gamma)\pi_\lambda(d\gamma)\) is continuous, which is weak continuity.

For the quadratic path, \(a_P\) is continuous and bounded on compact \(\Gamma_0\). The function \(\exp[\lambda a_P(\gamma)]p_0(\gamma)\) is jointly continuous and bounded above and below away from zero on \(\Gamma_0\times[0,\bar\lambda]\), so the preceding argument applies.

For the local path, write \(Z_\lambda=\int_\calH p_0(u)\exp[\lambda \tilde\ell_0(u)]du\). Boundedness of \(\tilde\ell_0\), compactness of \(\calH\), and the upper and lower bounds on \(p_0\) make \(Z_\lambda\) finite and bounded away from zero on compact \(\lambda\)-intervals. The density satisfies \(\log p_\lambda(h)=\log p_0(h)+\lambda \tilde\ell_0(h)-\log Z_\lambda\), and dominated convergence gives
\[
\partial_\lambda\log p_\lambda(h)
=
\tilde\ell_0(h)-
\frac{\int_\calH \tilde\ell_0(u)p_0(u)\exp[\lambda \tilde\ell_0(u)]du}
{\int_\calH p_0(u)\exp[\lambda \tilde\ell_0(u)]du}.
\]
This derivative is uniformly bounded on compact \(\lambda\)-intervals. At \(\lambda=0\), it is \(\tilde\ell_0(h)\) minus its \(p_0\)-mean, and Assumption~\ref{ass:smooth-prior} denotes this centered score by \(\ell_0\). Additive constants do not affect the quasi-posterior covariance calculations in Proposition~\ref{prop:small}. If \(\tilde\ell_0(h)=h_c'Ph-h'Ph/2\), then \(\ell_0(h)=c+h_c'Ph-h'Ph/2\) for a constant \(c\).
\end{proof}

\medskip

\begin{proof}[Proof of Proposition \ref{prop:primitive-prior-main}]
The retained spaces are compact, \(\Theta\) is compact, \(q_1\) is continuous on \(\calA\times\Theta\), and the baseline retained measure \(p_W^0\mu_W\) is finite with full support by construction and Assumption~\ref{ass:primitive-setup}. Hence
the map \((\alpha_W,\theta,\lambda)\mapsto\exp\{\lambda[q_{1,W}(\alpha_W,\theta;z_W)+q_{1,S}(\alpha_{S0},\theta;z_S)]\}\)
is bounded, strictly positive, and jointly continuous on \(\calW\times\Theta\times\Lambda\). Dominated convergence with respect to \(\pi_\theta\) gives joint continuity and strict positivity of \(r_\lambda(\alpha_W)\). Lemma~\ref{lem:primitive-prior}, applied with \(\Gamma_0=\calW\), \(\mu_0(d\alpha_W)=p_W^0(\alpha_W)\mu_W(d\alpha_W)\), and \(w_\lambda=r_\lambda\), gives properness and weak continuity of \((\pi_\lambda)_{\lambda\in\Lambda}\).

Lemma~\ref{lem:strong-primitive} gives \(\widehat\alpha_j^{\mathrm{un}}-\alpha_{j0}=o_p(1)\) for every \(j\in\mathcal J_S\). Since \(\mathcal J_S\) is finite, \(\widehat\alpha_{S,n}^{\mathrm{un}}-\alpha_{S0}=o_p(1)\). Uniform continuity of \(q_{1,S}\) on compact \(\calA_S\times\Theta\) gives
\[
        \Delta_n
        :=
        \sup_{\theta\in\Theta}
        \left|q_{1,S}(\widehat\alpha_{S,n}^{\mathrm{un}},\theta;z_S)-q_{1,S}(\alpha_{S0},\theta;z_S)\right|
        =o_p(1).
\]
For all \(\alpha_W\in\calW\) and \(\lambda\in\Lambda\),
\(e^{-\bar\lambda\Delta_n}\le r_{n,\lambda}(\alpha_W)/r_\lambda(\alpha_W)\le e^{\bar\lambda\Delta_n}\)
on events where \(\Delta_n\) is finite. Let \(\varepsilon_n=\bar\lambda\Delta_n\), write \(R_{n,\lambda}=r_{n,\lambda}/r_\lambda\), and set \(\bar R_{n,\lambda}=\int R_{n,\lambda}\,d\pi_\lambda\). The same bound gives \(e^{-\varepsilon_n}\le\bar R_{n,\lambda}\le e^{\varepsilon_n}\), so \(d\pi_{n,\lambda}^{\mathrm{plug}}/d\pi_\lambda=R_{n,\lambda}/\bar R_{n,\lambda}\in[e^{-2\varepsilon_n},e^{2\varepsilon_n}]\), uniformly over \(\alpha_W\) and \(\lambda\). Therefore \(\sup_{\lambda\in\Lambda}\|\pi_{n,\lambda}^{\mathrm{plug}}-\pi_\lambda\|_{TV}\le e^{2\varepsilon_n}-1=o_p(1)\).
\end{proof}

\medskip

\begin{proof}[Proof of Proposition \ref{prop:weak-marginal-main}]
First integrate out the strong studies conditionally on \((\alpha_W,\theta)\). For \(j\in\mathcal J_S\), Lemma~\ref{lem:strong-primitive} gives the local quadratic representation centered at \(\widehat\alpha_j^{\mathrm{un}}\), local prior flatness, hierarchy flatness, and a zeroth-moment tail bound. The last implication is immediate from the first-moment tail ratio, since \(0\le\int_{\|h\|>K_{j,n}}L_{j,n,\lambda}(h,\omega)dh/\int_{\|h\|\le K_{j,n}}L_{j,n,\lambda}(h,\omega)dh\le\int_{\|h\|>K_{j,n}}(1+\|h\|)L_{j,n,\lambda}(h,\omega)dh/\int_{\|h\|\le K_{j,n}}L_{j,n,\lambda}(h,\omega)dh\). Lemma~\ref{lem:deterministic-laplace} therefore yields, uniformly over the remaining variables and \(\lambda\in\Lambda\),
\[
        A_{j,n}
        \exp[\lambda q_{1,j}(\widehat\alpha_j^{\mathrm{un}},\theta;z_j)]
        [1+\eta_{j,n}(\theta,\lambda)],
        \qquad
        \sup_{\theta,\lambda}|\eta_{j,n}(\theta,\lambda)|=o_p(1),
\]
where
\[
A_{j,n}
=
n_j^{-d_j/2}
\exp[q_{2,j,n_j}(X_j,\widehat\alpha_j^{\mathrm{un}})]
p_j(\widehat\alpha_j^{\mathrm{un}})
(2\pi)^{d_j/2}\det(\widehat H_{j,n})^{-1/2}.
\]
The factor \(A_{j,n}\) depends on the strong-study data but not on \(\alpha_W\), \(\theta\), or \(\lambda\). Because \(\mathcal J_S\) is finite, multiplying the groupwise expansions preserves a uniform relative error: with \(\delta_n=\max_{j\in\mathcal J_S}\sup_{\theta,\lambda}|\eta_{j,n}(\theta,\lambda)|\), \(\prod_{j\in\mathcal J_S}[1+\eta_{j,n}(\theta,\lambda)]=1+\tilde\eta_n(\theta,\lambda)\) and \(\sup_{\theta,\lambda}|\tilde\eta_n(\theta,\lambda)|\le(1+\delta_n)^{|\mathcal J_S|}-1=o_p(1)\). The product \(\prod_{j\in\mathcal J_S}A_{j,n}\) cancels from normalized weak marginals.

By construction, the weak lower-level baseline prior contributes the measure \(p_W^0(\alpha_W)\mu_W(d\alpha_W)\) on \(\calW\). By \eqref{eq:weak-lower-level-link}, the weak lower-level objective is \(-Q_n(\alpha_W)/2\) up to a sample-dependent constant that does not depend on \(\alpha_W\). After integrating out the strong studies, the joint marginal law of \((\alpha_W,\theta)\) is therefore uniformly equivalent in total variation to the law with unnormalized measure
\[
\begin{aligned}
&\exp[-Q_n(\alpha_W)/2]
\exp\!\Bigl(\lambda\bigl[q_{1,W}(\alpha_W,\theta;z_W)
        +q_{1,S}(\widehat\alpha_{S,n}^{\mathrm{un}},\theta;z_S)\bigr]\Bigr)\\
&\qquad\times p_W^0(\alpha_W)\mu_W(d\alpha_W)\pi_\theta(d\theta).
\end{aligned}
\]
Integrating out \(\theta\) gives the weak law, up to the same uniformly vanishing relative error, with unnormalized measure \(\exp[-Q_n(\alpha_W)/2]r_{n,\lambda}(\alpha_W)p_W^0(\alpha_W)\mu_W(d\alpha_W)\).
Let \(\Pi^{\mathrm{plug}}_{n,\lambda}\) denote the normalized weak quasi-posterior obtained from \(\Pi_{n,\lambda}\) by replacing \(\pi_\lambda\) with \(\pi_{n,\lambda}^{\mathrm{plug}}\). The relative-density normalization argument in Lemma~\ref{lem:weak-laplace-reduction} gives
\[
        \sup_{\lambda\in\Lambda}
        \|\Pi^{\mathrm{full},W}_{n,\lambda}-\Pi^{\mathrm{plug}}_{n,\lambda}\|_{TV}=o_p(1).
\]

Proposition~\ref{prop:primitive-prior-main} gives \(\sup_{\lambda\in\Lambda}\|\pi_{n,\lambda}^{\mathrm{plug}}-\pi_\lambda\|_{TV}=o_p(1)\). On events where \(\sup_{\alpha_W\in\calW}|Q_n(\alpha_W)|\le M\), the likelihood factor \(\ell_n(\alpha_W)=\exp[-Q_n(\alpha_W)/2]\) satisfies \(m_M\le\ell_n\le M_M\) for positive constants depending only on \(M\). For a probability measure \(\nu\), write \(T_\ell\nu(A)=\int_A\ell_n\,d\nu/\int\ell_n\,d\nu\). If \(\nu_1\) and \(\nu_2\) are probability measures on \(\calW\), set \(Z_i=\int\ell_n\,d\nu_i\), so \(m_M\le Z_i\le M_M\). For every measurable \(A\), \(|\int_A\ell_n\,d(\nu_1-\nu_2)|\le 2M_M\|\nu_1-\nu_2\|_{TV}\) and \(|Z_1-Z_2|\le 2M_M\|\nu_1-\nu_2\|_{TV}\). The quotient identity \(a/Z_1-b/Z_2=(a-b)/Z_1+b(Z_2-Z_1)/(Z_1Z_2)\), with \(0\le b\le M_M\), gives \(|T_\ell\nu_1(A)-T_\ell\nu_2(A)|\le[2M_M/m_M+2M_M^2/m_M^2]\|\nu_1-\nu_2\|_{TV}\). Thus \(\|T_\ell\nu_1-T_\ell\nu_2\|_{TV}\le C_M\|\nu_1-\nu_2\|_{TV}\). Since Proposition~\ref{prop:weak-process-main} implies \(\sup_{\alpha_W}|Q_n(\alpha_W)|=O_p(1)\), this Lipschitz comparison gives \(\sup_{\lambda\in\Lambda}\|\Pi^{\mathrm{plug}}_{n,\lambda}-\Pi_{n,\lambda}\|_{TV}=o_p(1)\). The triangle inequality proves the stated total-variation equivalence.

Finally, \(\calW\) is compact, so the coordinate map \(\alpha_W\mapsto\alpha_W\) is bounded. If \(D_\calW=\sup_{\alpha_W\in\calW}\|\alpha_W\|\) and \(P,Q\) are probability laws on \(\calW\), then \(\|\int\alpha_W\,dP-\int\alpha_W\,dQ\|\le2D_\calW\|P-Q\|_{TV}\). The uniform total-variation bound therefore implies the stated uniform bound on weak marginal quasi-posterior means.
\end{proof}

\medskip

\subsection{Posterior continuity and main results}

The following lemma is stated for a generic compact weak parameter space. It applies to the study-level weak experiment in the main theorem with \(\Gamma=\calW\), and it also applies to the profiled weak experiment in Appendix~\ref{app:mixed-within-study-blocks} with \(\Gamma=\Gamma_0\).

\begin{lemma}\label{lem:posterior-continuity}
Let \(\Gamma\subset\R^d\) and \(\Lambda\subset[0,\bar\lambda]\) be compact. Let \((\pi_\lambda)_{\lambda\in\Lambda}\) be a path of proper probability measures on \(\Gamma\) such that \(\lambda\mapsto\int f(\gamma)\pi_\lambda(d\gamma)\) is continuous for every \(f\in C_b(\Gamma)\). Let \((x_n,V_n)\) and \((x,V)\) be random elements of \(\ell^\infty(\Gamma,\R^k)\times\ell^\infty(\Gamma,\R^{k\times k})\) such that \((x_n,V_n)\dto(x,V)\), and suppose the sample paths of \(x\) and \(V\) are continuous almost surely. Define \(Q_n(\gamma)=x_n(\gamma)'V_n(\gamma)x_n(\gamma)\), \(Q(\gamma)=x(\gamma)'V(\gamma)x(\gamma)\),
\[
        \Pi_{n,\lambda}^\Gamma(d\gamma)
        =
        \frac{\exp[-Q_n(\gamma)/2]\pi_\lambda(d\gamma)}
        {\int_\Gamma\exp[-Q_n(u)/2]\pi_\lambda(du)},
        \qquad
        \Pi_{\lambda}^\Gamma(d\gamma\mid x,V)
        =
        \frac{\exp[-Q(\gamma)/2]\pi_\lambda(d\gamma)}
        {\int_\Gamma\exp[-Q(u)/2]\pi_\lambda(du)} .
\]
Then \(Q_n\dto Q\) in \(\ell^\infty(\Gamma)\), and
\[
        \left(\int_\Gamma\gamma\,\Pi_{n,\lambda}^\Gamma(d\gamma)\right)_{\lambda\in\Lambda}
        \dto
        \left(\int_\Gamma\gamma\,\Pi_{\lambda}^\Gamma(d\gamma\mid x,V)\right)_{\lambda\in\Lambda}
\]
in \(\ell^\infty(\Lambda,\R^d)\).
\end{lemma}

\begin{proof}[Proof of Lemma \ref{lem:posterior-continuity}]
Let \(\Phi\) map a process--weight pair \((x,V)\) into \(\Phi(x,V)(\gamma)=x(\gamma)'V(\gamma)x(\gamma)\). If \((x_r,V_r)\to(x,V)\) in the product sup norm, then \((x_r)\) and \((V_r)\) are uniformly bounded for all large \(r\), and
\[
\begin{aligned}
&\sup_{\gamma\in\Gamma}
|x_r(\gamma)'V_r(\gamma)x_r(\gamma)-x(\gamma)'V(\gamma)x(\gamma)| \\
&\quad\le
\|V_r\|_\infty\|x_r-x\|_\infty(\|x_r\|_\infty+\|x\|_\infty)
+\|x\|_\infty^2\|V_r-V\|_\infty .
\end{aligned}
\]
The right side tends to zero, so \(\Phi\) is continuous at every bounded pair and the continuous mapping theorem gives \(Q_n=\Phi(x_n,V_n)\dto\Phi(x,V)=Q\). Since the limiting sample paths of \(x\) and \(V\) are continuous on compact \(\Gamma\), \(Q\) is continuous almost surely.

It remains to show that the quasi-posterior mean is continuous as a functional of the objective. For bounded \(q\in\ell^\infty(\Gamma)\), define
\[
\Psi_\lambda(q)
=
\frac{\int_\Gamma\gamma\exp[-q(\gamma)/2]\pi_\lambda(d\gamma)}
{\int_\Gamma\exp[-q(u)/2]\pi_\lambda(du)},
\qquad
\Psi(q)=(\Psi_\lambda(q))_{\lambda\in\Lambda} .
\]
Fix a continuous \(q\); for each \(\lambda\), the denominator is strictly positive. Since \(\pi_\lambda\) is a probability measure and \(q\) is bounded on compact \(\Gamma\), \(\int\exp[-q(\gamma)/2]\pi_\lambda(d\gamma)\ge\exp[-\|q\|_\infty/2]\), so the denominator is bounded away from zero uniformly over \(\lambda\). Weak continuity of \(\lambda\mapsto\pi_\lambda\) implies continuity in \(\lambda\) of both numerator and denominator, since the functions \(\gamma\mapsto\gamma_\ell\exp[-q(\gamma)/2]\), \(\ell=1,\ldots,d\), and \(\gamma\mapsto\exp[-q(\gamma)/2]\) are bounded and continuous. Hence \(\Psi(q)\in\ell^\infty(\Lambda,\R^d)\).

Suppose \(q_r\to q\) uniformly. Put \(\Delta_r=\sup_{\gamma\in\Gamma}|\exp[-q_r(\gamma)/2]-\exp[-q(\gamma)/2]|\), so \(\Delta_r\to0\). Let \(D_\Gamma=\sup_{\gamma\in\Gamma}\|\gamma\|\), and let \(c_q=\inf_{\lambda\in\Lambda}\int\exp[-q(\gamma)/2]\pi_\lambda(d\gamma)>0\). Uniformly over \(\lambda\), the numerator difference is bounded by \(D_\Gamma\Delta_r\), and the denominator difference is bounded by \(\Delta_r\). For all large \(r\), the denominator at \(q_r\) is at least \(c_q/2\). Writing the two numerators and denominators as \(N_{r,\lambda},N_\lambda,D_{r,\lambda},D_\lambda\), the identity \(N_{r,\lambda}/D_{r,\lambda}-N_\lambda/D_\lambda=(N_{r,\lambda}-N_\lambda)/D_{r,\lambda}+N_\lambda(D_\lambda-D_{r,\lambda})/(D_{r,\lambda}D_\lambda)\), together with \(\|N_\lambda\|\le D_\Gamma D_\lambda\), gives \(\sup_{\lambda\in\Lambda}\|\Psi_\lambda(q_r)-\Psi_\lambda(q)\|\le C_qD_\Gamma\Delta_r\to0\), where \(C_q<\infty\) depends only on \(c_q\). The map \(q\mapsto\Psi(q)\) is therefore continuous at every continuous \(q\). Applying the continuous mapping theorem to \(Q_n\dto Q\) proves the displayed convergence of quasi-posterior mean paths.
\end{proof}

\medskip

\begin{proof}[Proof of Theorem \ref{thm:main}]
Proposition~\ref{prop:weak-process-main} gives \((g_n,\widehat W_n)\dto(g,W)\), with continuous limiting sample paths, and Proposition~\ref{prop:primitive-prior-main} gives the proper weakly continuous prior path \((\pi_\lambda)_{\lambda\in\Lambda}\). Applying Lemma~\ref{lem:posterior-continuity} with \(\Gamma=\calW\), \(x_n=g_n\), \(V_n=\widehat W_n\), \(x=g\), and \(V=W\) yields
\((T^{W}_{n,\lambda})_{\lambda\in\Lambda}\dto(t_\lambda(g))_{\lambda\in\Lambda}\)
in \(\ell^\infty(\Lambda,\R^{d_W})\). Proposition~\ref{prop:weak-marginal-main} gives \(\sup_{\lambda\in\Lambda}\|\Pi^{\mathrm{full},W}_{n,\lambda}-\Pi_{n,\lambda}\|_{TV}=o_p(1)\). Using the bounded-coordinate inequality from the proof of that proposition, there is a finite constant \(C_\calW\) such that the difference between the corresponding weak marginal means is bounded by \(C_\calW\|\Pi^{\mathrm{full},W}_{n,\lambda}-\Pi_{n,\lambda}\|_{TV}\) for every \(\lambda\). The full weak marginal mean path and \((T^{W}_{n,\lambda})_{\lambda\in\Lambda}\) therefore differ by \(o_p(1)\) in sup norm. Slutsky's theorem gives the second displayed convergence in the theorem.
\end{proof}

\medskip

\begin{proof}[Proof of Proposition \ref{prop:posterior-risk}]
Fix \(\lambda\in\Lambda\) and a realized process \(g\) for which \(\Pi_\lambda(\cdot\mid g)\) is well defined, and write \(\Pi=\Pi_\lambda(\cdot\mid g)\) and \(\bar a=t_\lambda(g)=\int \alpha_W\Pi(d\alpha_W)\), which is finite by compactness of \(\calW\). For any finite action \(a\in\R^{d_W}\),
\(B(a-\alpha_W)=B(a-\bar a)+B(\bar a-\alpha_W)\).
Expanding the squared norm and integrating with respect to \(\Pi\) gives
\[
\begin{aligned}
&\int_{\calW}\|B(a-\alpha_W)\|^2\Pi(d\alpha_W)
-
\int_{\calW}\|B(\bar a-\alpha_W)\|^2\Pi(d\alpha_W) \\
&\qquad =\|B(a-\bar a)\|^2
+2 B(a-\bar a)'B\int_{\calW}(\bar a-\alpha_W)\Pi(d\alpha_W)
=\|B(a-\bar a)\|^2 .
\end{aligned}
\]
The conditional integrals are finite for every finite \(a\), because \(\calW\) is compact. The calculation does not require \(B\) to have full column rank; if \(B\) has a nontrivial null space, every action satisfying \(B(a-\bar a)=0\) has the same conditional posterior risk as \(\bar a\).

Applying the identity pathwise with \(a=t(g)\), the posterior risk of \(t_\lambda(g)\) is bounded uniformly in \(g\), since both \(t_\lambda(g)\) and \(\alpha_W\) lie in the convex hull of compact \(\calW\). Thus \(IPR_{\lambda,B,\mathsf M}(t_\lambda)<\infty\). Tonelli's theorem yields
\[
        IPR_{\lambda,B,\mathsf M}(t)
        =
        IPR_{\lambda,B,\mathsf M}(t_\lambda)
        +
        \int\|B(t(g)-t_\lambda(g))\|^2\mathsf M(dg),
\]
where the rightmost integral may be infinite. This proves the extended-valued identity in the proposition and shows that \(t_\lambda\) is a Bayes rule among measurable finite-valued rules. The set of Bayes rules is generally nonunique: any measurable finite-valued \(t\) with \(B(t(g)-t_\lambda(g))=0\) for \(\mathsf M\)-almost every \(g\) attains the same integrated posterior risk.
\end{proof}

\medskip

\subsection{Diffuse nuisance priors in the continuously updated weak-GMM limit}
\label{app:am-diffuse-limit}

This subsection records the connection between the weak quasi-posterior used in the main text and the diffuse-nuisance-prior construction in \citet{AM2022}. It applies only to the continuously updated weak criterion. This restriction is substantive: the proportional Gaussian-process prior on the nuisance mean in \citet{AM2022} yields the continuously updated weak-GMM quasi-likelihood, up to determinant factors that do not depend on the observed process. The preceding weak-limit and MSE results allow more general weight processes \(W\); the result below does not give those other criteria an Andrews--Mikusheva prior interpretation.

To keep notation short, this subsection writes \(\alpha\) for a candidate retained weak value \(\alpha_W\). Let \(\Sigma(\alpha,\tilde\alpha):=\Cov(\mathbb G(\alpha),\mathbb G(\tilde\alpha))\), and suppose \(\alpha\mapsto\Sigma(\alpha,\alpha)\) is continuous and uniformly nonsingular on \(\calW\). In this subsection \(W(\alpha)=\Sigma(\alpha,\alpha)^{-1}\), and the continuously updated weak criterion is \(Q^{\mathrm{CU}}_g(\alpha)=g(\alpha)'\Sigma(\alpha,\alpha)^{-1}g(\alpha)\). The corresponding weak quasi-posterior is
\[
        \Pi_\lambda^{\mathrm{CU}}(d\alpha\mid g)
        =
        \frac{\exp[-Q^{\mathrm{CU}}_g(\alpha)/2]\pi_\lambda(d\alpha)}
        {\int_{\calW}\exp[-Q^{\mathrm{CU}}_g(u)/2]\pi_\lambda(du)} .
\]
Let \(A\) be an Andrews--Mikusheva reparameterizing linear functional, which separates the structural restriction \(m(\alpha)=0\) from the remaining nuisance mean, and let \(L_{\kappa,A}(\alpha;g)\) denote the structural integrated likelihood obtained after integrating the nuisance mean under the proportional Gaussian-process prior with variance scale \(\kappa\). The scale \(\kappa\) is distinct from the pooling strength \(\lambda\): \(\lambda\) controls shrinkage across retained weak values, while \(\kappa\) controls the diffuseness of the nuisance-mean prior. The proportional-prior restriction corresponds to \citet[Theorem~3]{AM2022}, and the diffuse integrated-likelihood limit corresponds to \citet[Section~3.2, equation~(8)]{AM2022}. We use the following uniform version of that limit on \(\calW\): there is a finite, positive, continuous function \(c_A:\calW\to(0,\infty)\), independent of \(g\), such that, for every continuous process path \(g\),
\[
        \varepsilon_{\kappa,A}(g)
        :=
        \sup_{\alpha\in\calW}
        \left|
        \frac{L_{\kappa,A}(\alpha;g)}
        {c_A(\alpha)\exp[-Q^{\mathrm{CU}}_g(\alpha)/2]}
        -1
        \right|
        \to0
        \qquad\text{as } \kappa\to\infty .
\]
The factor \(c_A\) collects the determinant terms in \citet[Section~3.2, equation~(8)]{AM2022}. Because it does not depend on \(g\), it can be absorbed into the structural prior. For each \(\lambda\), define
\[
        \pi_{\lambda,A}^{\mathrm{AM}}(d\alpha)
        =
        \frac{c_A(\alpha)^{-1}\pi_\lambda(d\alpha)}
        {\int_{\calW}c_A(u)^{-1}\pi_\lambda(du)} ,
\]
and let the finite-\(\kappa\) structural posterior be
\[
        \Pi_{\lambda,\kappa,A}^{\mathrm{AM}}(d\alpha\mid g)
        =
        \frac{L_{\kappa,A}(\alpha;g)\pi_{\lambda,A}^{\mathrm{AM}}(d\alpha)}
        {\int_{\calW}L_{\kappa,A}(u;g)\pi_{\lambda,A}^{\mathrm{AM}}(du)} .
\]

\begin{proposition}\label{prop:am-diffuse-limit}
Suppose Assumptions~\ref{ass:primitive-setup}, \ref{ass:study-level-split}, \ref{ass:primitive-strong}, and \ref{ass:primitive-weak} hold, let \((\pi_\lambda)_{\lambda\in\Lambda}\) be the induced prior path in Proposition~\ref{prop:primitive-prior-main}, and suppose the continuously updated covariance and diffuse-likelihood conditions stated in this subsection hold. Then, for every continuous weak-process path \(g\),
\[
        \sup_{\lambda\in\Lambda}
        \left\|
        \Pi_{\lambda,\kappa,A}^{\mathrm{AM}}(\cdot\mid g)
        -
        \Pi_\lambda^{\mathrm{CU}}(\cdot\mid g)
        \right\|_{TV}
        \to0
        \qquad\text{as } \kappa\to\infty .
\]
Consequently, the corresponding posterior means converge uniformly over \(\lambda\in\Lambda\):
\[
        \sup_{\lambda\in\Lambda}
        \left\|
        \int_{\calW}\alpha\,\Pi_{\lambda,\kappa,A}^{\mathrm{AM}}(d\alpha\mid g)
        -
        \int_{\calW}\alpha\,\Pi_\lambda^{\mathrm{CU}}(d\alpha\mid g)
        \right\|
        \to0 .
\]
\end{proposition}

\begin{proof}[Proof of Proposition \ref{prop:am-diffuse-limit}]
Fix a continuous path \(g\), and write \(\ell(\alpha)=\exp[-Q^{\mathrm{CU}}_g(\alpha)/2]\). Since \(g\) is continuous, \(\calW\) is compact, and \(\Sigma(\alpha,\alpha)\) is continuous and uniformly nonsingular, \(Q^{\mathrm{CU}}_g\) is continuous and bounded. Hence \(0<\underline\ell\le \ell(\alpha)\le\bar\ell<\infty\) on \(\calW\). Let \(e_{\kappa,A}(\alpha;g)=L_{\kappa,A}(\alpha;g)/(c_A(\alpha)\ell(\alpha))-1\), so \(\|e_{\kappa,A}(\cdot;g)\|_\infty=\varepsilon_{\kappa,A}(g)\to0\).

For every measurable \(E\subseteq\calW\) and all sufficiently large \(\kappa\), the definition of \(\pi_{\lambda,A}^{\mathrm{AM}}\) cancels the factor \(c_A\), so
\[
        \Pi_{\lambda,\kappa,A}^{\mathrm{AM}}(E\mid g)
        =
        \frac{\int_E[1+e_{\kappa,A}(\alpha;g)]\ell(\alpha)\pi_\lambda(d\alpha)}
        {\int_{\calW}[1+e_{\kappa,A}(u;g)]\ell(u)\pi_\lambda(du)} .
\]
Let \(P_\lambda^g\) denote \(\Pi_\lambda^{\mathrm{CU}}(\cdot\mid g)\), and define \(\bar e_{\lambda,\kappa}=\int e_{\kappa,A}(\alpha;g)P_\lambda^g(d\alpha)\). Then \(|\bar e_{\lambda,\kappa}|\le\varepsilon_{\kappa,A}(g)\), and the preceding display gives
\[
        \Pi_{\lambda,\kappa,A}^{\mathrm{AM}}(E\mid g)-P_\lambda^g(E)
        =
        \frac{\int_E [e_{\kappa,A}(\alpha;g)-\bar e_{\lambda,\kappa}]P_\lambda^g(d\alpha)}
        {1+\bar e_{\lambda,\kappa}} .
\]
For \(\kappa\) large enough that \(\varepsilon_{\kappa,A}(g)<1\), the absolute value of the right side is bounded by \(2\varepsilon_{\kappa,A}(g)/(1-\varepsilon_{\kappa,A}(g))\), uniformly over \(E\) and \(\lambda\). Taking the supremum over measurable \(E\) gives the total-variation convergence. Since \(\calW\) is compact, \(D_\calW:=\sup_{\alpha\in\calW}\|\alpha\|<\infty\), and the mean difference is bounded by \(2D_\calW\|\Pi_{\lambda,\kappa,A}^{\mathrm{AM}}(\cdot\mid g)-\Pi_\lambda^{\mathrm{CU}}(\cdot\mid g)\|_{TV}\). The displayed convergence of posterior means follows uniformly over \(\lambda\).
\end{proof}

The proposition is limited to the continuously updated specialization. It does not alter the general weak-limit approximation or the pointwise MSE results, which are stated for the weight process \(W\) allowed in the main text. The result shows that, when the lower-level weak criterion is continuously updated, the hierarchy-induced structural prior \(\pi_\lambda\) can be used inside the Andrews--Mikusheva diffuse-nuisance-prior argument without changing the main estimator.

\medskip

\subsection{Mixed within-study blocks}
\label{app:mixed-within-study-blocks}

The main fixed-\(J\) theorem imposes a study-level identification split, and this subsection treats the case in which a single study contains both a DQM-regular block and a retained weak block. The regular block is integrated out conditional on the retained weak coordinate, and the resulting weak experiment is a reduced weak-GMM problem with the hierarchy evaluated on the population profile and with a smooth positive baseline factor.

\paragraph{Relation to the homogeneous cases.}
This subsection is a profiling extension of the study-level split used in the main text, not a third identification category. With the standard empty-coordinate conventions, the homogeneous cases are nested by the same notation. If a group has no retained weak coordinate, the retained space can be read as a singleton and the conditional regular-block approximation reduces to the strong-study approximation in Theorem~\ref{thm:strong-equivalence}. If a group has no regular coordinate, there is no Laplace step; the corresponding factor in \(C_S\) below is one, and the retained weak marginal reduces to the weak-study construction in Proposition~\ref{prop:weak-marginal-main}, with any deterministic baseline factor collected in \(C_W\). The main text states the homogeneous cases separately for readability, while this appendix records the profiled notation needed when both coordinate types appear in the same study.

For group \(j\), write \(\alpha_j=(\beta_j',\gamma_j')'\), where \(\beta_j\in\calB_j\subset\R^{d_{\beta j}}\) is the DQM-regular block and \(\gamma_j\in\Gamma_j\subset\R^{d_{\gamma j}}\) is the retained weak block. In this subsection, every occurrence of \(\gamma_j\) refers to this retained weak coordinate, while \(\theta\) remains the upper-level hierarchy parameter. In the split notation, write \(M_j(\beta,\gamma)\) for \(M_j(a)\) with \(a=(\beta',\gamma')'\). For each fixed \(\gamma\), let \(\beta_{j0}(\gamma)\) be the population maximizer of \(M_j(\cdot,\gamma)\) in the regular block, and define \(H_j(\gamma)=-\partial_{\beta\beta}^2M_j(\beta_{j0}(\gamma),\gamma)\). When a proof conditions on retained variables, \(\omega\) denotes all coordinates in the full quasi-posterior except the current regular coordinate \(\beta_j\); in particular, \(\gamma_j(\omega)\) is the retained weak coordinate of group \(j\), and the retained-coordinate set is compact by Assumption~\ref{ass:primitive-setup}.

\begin{assumption}\label{ass:mixed-strong}
For each group with \(d_{\beta j}>0\), the following conditions hold. All stochastic convergence statements are under the triangular-array law used for that group in the statement combining regular and weak rates.
\begin{enumerate}[label=(\roman*),leftmargin=2.1em]
\item For every \(\gamma\in\Gamma_j\), \(\beta_{j0}(\gamma)\) is the unique interior maximizer of \(M_j(\cdot,\gamma)\). The function \(M_j\) is continuous on \(\calB_j\times\Gamma_j\) and twice continuously differentiable in \(\beta\) on a neighborhood of the graph \(\{(\beta_{j0}(\gamma),\gamma):\gamma\in\Gamma_j\}\).
\item The curvature in the regular block is uniformly nonsingular: \(H_j(\gamma)\) is continuous in \(\gamma\), and there is a constant \(c_j>0\) such that \(v'H_j(\gamma)v\ge c_j\|v\|^2\) for all \(v\in\R^{d_{\beta j}}\) and all \(\gamma\in\Gamma_j\).
\item The sample objective and its first two \(\beta\)-derivatives converge to their population counterparts. Specifically, \(\sup_{(\beta,\gamma)\in\calB_j\times\Gamma_j}|\bar q_{2,j,n_j}(X_j,(\beta,\gamma))-M_j(\beta,\gamma)|\pto0\). There is \(\delta_j>0\) such that \(\mathcal U_j(\delta_j)=\{(\beta,\gamma):\gamma\in\Gamma_j,\ \|\beta-\beta_{j0}(\gamma)\|\le\delta_j\}\) lies in the differentiability neighborhood and, for \(r=1,2\),
\[
        \sup_{(\beta,\gamma)\in\mathcal U_j(\delta_j)}
        \left\|
        \nabla_\beta^r\bar q_{2,j,n_j}(X_j,(\beta,\gamma))
        -
        \nabla_\beta^r M_j(\beta,\gamma)
        \right\|
        \pto0.
\]
The sample score at the population profile has the regular order needed for the local expansion:
\[
        \sup_{\gamma\in\Gamma_j}
        \left\|
        n_j^{-1/2}
        \nabla_\beta q_{2,j,n_j}(X_j,(\beta_{j0}(\gamma),\gamma))
        \right\|
        =O_p(1).
\]
\item Write the lower-level prior in split coordinates as \(p_{\beta j\mid\gamma j}(\beta_j\mid\gamma_j)\mu_{\gamma j}(d\gamma_j)\), where \(\mu_{\gamma j}\) is the marginal weak-coordinate prior measure. The conditional density \(p_{\beta j\mid\gamma j}\) is positive, continuous in \((\beta_j,\gamma_j)\), and continuously differentiable in \(\beta_j\) on \(\mathcal U_j(\delta_j)\). For every fixed \(z\), the pooling objective \(q_1\), viewed as a function of \(\beta_j\) with all retained weak coordinates, other group coordinates, and \(\theta\) held fixed, has first and second partial derivatives on a neighborhood of \(\mathcal U_j(\delta_j)\); these derivatives are jointly continuous on the corresponding compact retained-coordinate set.
\end{enumerate}
\end{assumption}

Assumption~\ref{ass:mixed-strong} is the regular-block condition applied conditionally on each retained weak value. The baseline theorem uses the study-level split, while this subsection applies Assumption~\ref{ass:primitive-weak} to the profiled weak experiment after the regular coordinates have been integrated out. The reduced weak criterion replaces the study-level weak criterion from Section~\ref{sec:weak-components}, and any deterministic positive baseline factor from that reduction is collected in \(C_W\).

\begin{lemma}\label{lem:mixed-strong-primitive}
Suppose Assumptions~\ref{ass:primitive-setup} and \ref{ass:mixed-strong} hold for group~\(j\), and fix any deterministic sequence \(R_{j,n}\to\infty\) with \(R_{j,n}=o(\sqrt{n_j})\) and \(R_{j,n}^2/\log n_j\to\infty\). There are events with probability tending to one, measurable local maximizers \(\widehat\beta_j(\omega)\), and a sequence \(K_{j,n}\to\infty\) satisfying \(2K_{j,n}\le R_{j,n}\) and \(K_{j,n}=o(\sqrt{n_j})\), such that the following statements hold uniformly over retained \(\omega\). First,
\[
        \sup_\omega
        \sqrt{n_j}\|\widehat\beta_j(\omega)-\beta_{j0}(\gamma_j(\omega))\|=O_p(1),
        \qquad
        \sup_\omega\|\widehat H_{j,n}(\omega)-H_j(\gamma_j(\omega))\|=o_p(1),
\]
where \(\widehat H_{j,n}(\omega):=-\partial_{\beta\beta}^2\bar q_{2,j,n_j}(X_j,(\widehat\beta_j(\omega),\gamma_j(\omega)))\). Second,
\[
\begin{aligned}
&\sup_{\omega,\|h\|\le K_{j,n}}
\Big|q_{2,j,n_j}(X_j,(\widehat\beta_j(\omega)+h/\sqrt{n_j},\gamma_j(\omega)))
-q_{2,j,n_j}(X_j,(\widehat\beta_j(\omega),\gamma_j(\omega))) \\
&\hspace{6em}
+\frac12h'\widehat H_{j,n}(\omega)h\Big|=o_p(1).
\end{aligned}
\]
Third, the conditional lower-level prior-density ratio is \(1+o_p(1)\), and the hierarchy increment is \(o_p(1)\), uniformly over \(\lambda\in\Lambda\), retained \(\omega\), and \(\|h\|\le K_{j,n}\). Finally, the first-moment contribution of \(\|h\|>K_{j,n}\) to the conditional regular-block quasi-posterior kernel is \(o_p(1)\) relative to the contribution of \(\|h\|\le K_{j,n}\), uniformly over \(\lambda\in\Lambda\) and retained \(\omega\).
\end{lemma}

\begin{proof}[Proof of Lemma \ref{lem:mixed-strong-primitive}]
Write \(\gamma=\gamma_j(\omega)\), and let \(\Omega_j\) denote the compact retained index set over which \(\omega\) ranges. Only the coordinate \(\gamma\) enters the lower-level profile, but the prior and hierarchy bounds below are uniform over all retained variables in \(\Omega_j\). The maximum theorem, compactness of \(\Gamma_j\), uniqueness of the population profile, and continuity of \(M_j\) imply that \(\gamma\mapsto\beta_{j0}(\gamma)\) is continuous, so the profile graph is compact. For each \(\varepsilon>0\), let
\(S_\varepsilon=\{(\beta,\gamma)\in\calB_j\times\Gamma_j:\|\beta-\beta_{j0}(\gamma)\|\ge\varepsilon\}\).
If \(S_\varepsilon\) is nonempty, it is compact, and the function \((\beta,\gamma)\mapsto M_j(\beta_{j0}(\gamma),\gamma)-M_j(\beta,\gamma)\) is continuous and strictly positive on \(S_\varepsilon\). Hence
\(\Delta_j(\varepsilon):=\inf_{(\beta,\gamma)\in S_\varepsilon}\left[M_j(\beta_{j0}(\gamma),\gamma)-M_j(\beta,\gamma)\right]>0\). On the event \(\sup_{(\beta,\gamma)\in\calB_j\times\Gamma_j}|\bar q_{2,j,n_j}(X_j,(\beta,\gamma))-M_j(\beta,\gamma)|\le\Delta_j(\varepsilon)/4\), this gives \(\bar q_{2,j,n_j}(X_j,(\beta_{j0}(\gamma),\gamma))-\bar q_{2,j,n_j}(X_j,(\beta,\gamma))\ge\Delta_j(\varepsilon)/2\) for every \((\beta,\gamma)\in S_\varepsilon\), and hence the corresponding sample objective gap is at least \(n_j\Delta_j(\varepsilon)/2\). If \(S_\varepsilon\) is empty, no admissible regular coordinate is that far from its profile.

The score condition in Assumption~\ref{ass:mixed-strong} gives
\(\sup_{\gamma\in\Gamma_j}\left\|n_j^{-1/2}\nabla_\beta q_{2,j,n_j}(X_j,(\beta_{j0}(\gamma),\gamma))\right\|=O_p(1)\).
The uniform positive-definiteness condition implies that, after reducing \(\delta_j\) if necessary, \(-\partial_{\beta\beta}^2M_j(\beta,\gamma)\) is uniformly positive definite on \(\mathcal U_j(\delta_j)\), with population lower eigenvalue at least \(2c>0\). On the event \(\sup_{(\beta,\gamma)\in\mathcal U_j(\delta_j)}\|-\partial_{\beta\beta}^2\bar q_{2,j,n_j}(X_j,(\beta,\gamma))+\partial_{\beta\beta}^2M_j(\beta,\gamma)\|\le c\), every unit vector \(v\) satisfies \(v'[-\partial_{\beta\beta}^2\bar q_{2,j,n_j}(X_j,(\beta,\gamma))]v\ge c\), uniformly on \(\mathcal U_j(\delta_j)\). The derivative laws make this event have probability approaching one, so the sample objective is uniformly strictly concave in the regular coordinate on the local chart.

For any fixed \(M<\infty\), Taylor's theorem around the profile gives, uniformly over \(\gamma\in\Gamma_j\) and \(\|h\|\le M\),
\[
\begin{aligned}
&q_{2,j,n_j}(X_j,(\beta_{j0}(\gamma)+h/\sqrt{n_j},\gamma))
-q_{2,j,n_j}(X_j,(\beta_{j0}(\gamma),\gamma)) \\
&\qquad =S_{j,n}(\gamma)'h-\frac12h'H_j(\gamma)h+o_p(1),
\end{aligned}
\]
where \(S_{j,n}(\gamma)=n_j^{-1/2}\nabla_\beta q_{2,j,n_j}(X_j,(\beta_{j0}(\gamma),\gamma))\). Since \(\sup_\gamma\|S_{j,n}(\gamma)\|=O_p(1)\) and the eigenvalues of \(H_j(\gamma)\) are uniformly bounded below, fix \(B<\infty\) so that \(\sup_\gamma\|S_{j,n}(\gamma)\|\le B\) with probability arbitrarily close to one, and choose \(M\) so large that \(BM-cM^2/2+cM^2/8<0\). On the event where the score bound holds and the Taylor remainder is bounded by \(cM^2/8\) on \(\Gamma_j\times\{\|h\|\le M\}\), the boundary increment is at most \(BM-cM^2/2+cM^2/8<0\) for every \(\gamma\) and every \(\|h\|=M\). Continuity of the sample objective on the compact local ball then gives, for each retained value, a local maximizer inside that ball. A measurable selection exists by the measurable maximum theorem. Therefore
\(\sup_{\omega\in\Omega_j}\sqrt{n_j}\left\|\widehat\beta_j(\omega)-\beta_{j0}(\gamma_j(\omega))\right\|=O_p(1)\).
The selected local maximizer lies in the interior of the local chart with probability approaching one, so the first-order condition in the \(\beta\)-coordinate holds uniformly over retained values.

Define
\(\widehat H_{j,n}(\omega)=-\partial_{\beta\beta}^2\bar q_{2,j,n_j}(X_j,(\widehat\beta_j(\omega),\gamma_j(\omega)))\).
The uniform rate for \(\widehat\beta_j(\omega)\), the derivative laws, and continuity of the population Hessian give \(\sup_{\omega\in\Omega_j}\|\widehat H_{j,n}(\omega)-H_j(\gamma_j(\omega))\|=o_p(1)\). To see the uniformity explicitly, let \(\omega_H\) be the continuity modulus of \(-\partial_{\beta\beta}^2M_j(\beta,\gamma)\) on \(\mathcal U_j(\delta_j)\). On the local event,
\[
\begin{aligned}
\sup_\omega\|\widehat H_{j,n}(\omega)-H_j(\gamma_j(\omega))\|
&\le
\sup_{(\beta,\gamma)\in\mathcal U_j(\delta_j)}
\Bigl\|-\partial_{\beta\beta}^2\bar q_{2,j,n_j}(X_j,(\beta,\gamma)) \\
&\qquad\qquad\qquad\qquad
+\partial_{\beta\beta}^2M_j(\beta,\gamma)\Bigr\| \\
&\quad
+\omega_H\!\bigl(\sup_\omega\|\widehat\beta_j(\omega)-\beta_{j0}(\gamma_j(\omega))\|\bigr),
\end{aligned}
\]
and both terms are \(o_p(1)\). The eigenvalues of \(\widehat H_{j,n}(\omega)\) are therefore bounded away from zero and infinity uniformly in \(\omega\) on high-probability events. Since \(R_{j,n}=o(\sqrt{n_j})\), \(\sup_{\omega,\|h\|\le R_{j,n}}\|\widehat\beta_j(\omega)+h/\sqrt{n_j}-\beta_{j0}(\gamma_j(\omega))\|\le O_p(n_j^{-1/2})+R_{j,n}/\sqrt{n_j}=o_p(1)\). Thus the local chart \(\|h\|\le R_{j,n}\) remains inside \(\mathcal U_j(\delta_j)\) for all large \(n\) on those events.

On the same event define \(\eta_{j,n}(R_{j,n})=\sup_{\omega,\|u\|\le R_{j,n}}\|-\partial_{\beta\beta}^2\bar q_{2,j,n_j}(X_j,(\widehat\beta_j(\omega)+u/\sqrt{n_j},\gamma_j(\omega)))-\widehat H_{j,n}(\omega)\|\). The derivative law and continuity modulus give \(\eta_{j,n}(R_{j,n})=o_p(1)\); explicitly, it is bounded by the sum of two copies of the uniform Hessian-convergence error on \(\mathcal U_j(\delta_j)\) and \(\omega_H(\sup_{\omega,\|u\|\le R_{j,n}}\|\widehat\beta_j(\omega)+u/\sqrt{n_j}-\beta_{j0}(\gamma_j(\omega))\|)\). Choose deterministic \(b_n\downarrow0\) with \(P(\eta_{j,n}(R_{j,n})>b_n)\to0\), and set \(K_{j,n}=\min(R_{j,n}/2,b_n^{-1/4})\). Then \(K_{j,n}\uparrow\infty\), \(2K_{j,n}\le R_{j,n}\), \(K_{j,n}=o(\sqrt{n_j})\), and \(K_{j,n}^2\eta_{j,n}(R_{j,n})=o_p(1)\). Taylor's theorem around \(\widehat\beta_j(\omega)\), using the first-order condition, gives
\[
\begin{aligned}
&q_{2,j,n_j}(X_j,(\widehat\beta_j(\omega)+h/\sqrt{n_j},\gamma_j(\omega)))
-q_{2,j,n_j}(X_j,(\widehat\beta_j(\omega),\gamma_j(\omega))) \\
&\qquad =-\frac12h'\widehat H_{j,n}(\omega)h+r^q_{j,n}(h,\omega),
\end{aligned}
\]
with \(\sup_{\omega,\|h\|\le K_{j,n}}|r^q_{j,n}(h,\omega)|\le K_{j,n}^2\eta_{j,n}(R_{j,n})/2=o_p(1)\).

It remains to control the nonquadratic factors and the tails. Positivity and continuous differentiability of the conditional prior density in the regular coordinate on the compact profile neighborhood imply a uniform bound on \(\|\nabla_\beta\log p_{\beta j\mid\gamma j}(\beta\mid\gamma)\|\). Hence \(\sup_{\omega,\|h\|\le K_{j,n}}|\log p_{\beta j\mid\gamma j}(\widehat\beta_j(\omega)+h/\sqrt{n_j}\mid\gamma_j(\omega))-\log p_{\beta j\mid\gamma j}(\widehat\beta_j(\omega)\mid\gamma_j(\omega))|\le CK_{j,n}/\sqrt{n_j}=o(1)\), which gives the conditional prior-density ratio \(1+o_p(1)\) uniformly over \(\omega\) and \(\|h\|\le K_{j,n}\). The continuous first and second \(\beta_j\)-derivatives of \(q_1\) on the compact support give, uniformly over \(\lambda\in\Lambda\), retained variables, and \(\|h\|\le K_{j,n}\),
\[
        \left|\lambda[q_1(\widehat\beta_j(\omega)+h/\sqrt{n_j},\omega;z)-q_1(\widehat\beta_j(\omega),\omega;z)]\right|
        \le C K_{j,n}/\sqrt{n_j}+C K_{j,n}^2/n_j=o(1).
\]
This proves the local prior-flatness and hierarchy-control statements.

For the tail ratio, let \(L_{j,n,\lambda}(h,\omega)\) denote the conditional regular-block integrand in the \(h\)-coordinate after division by its value at \(h=0\). First consider \(K_{j,n}<\|h\|\le R_{j,n}/2\). The local denominator is bounded below uniformly by the same unit-ball calculation used above: the local expansion, uniform eigenvalue bound, and bounded nonquadratic factors give \(L_{j,n,\lambda}(h,\omega)\ge C^{-1}\exp[-C\|h\|^2]\) on \(\|h\|\le1\), so the denominator is at least \(C^{-1}\int_{\|h\|\le1}\exp[-C\|h\|^2]dh>0\). Taylor's theorem around \(\widehat\beta_j(\omega)\), the first-order condition, and the uniform Hessian lower bound give \(q_{2,j,n_j}(X_j,(\widehat\beta_j(\omega)+h/\sqrt{n_j},\gamma_j(\omega)))-q_{2,j,n_j}(X_j,(\widehat\beta_j(\omega),\gamma_j(\omega)))\le-c\|h\|^2\), after reducing \(c\) if needed, on the annulus uniformly over \(\omega\). The conditional prior density and hierarchy factor are bounded above and below on compact support, uniformly in \(\lambda\), so \(L_{j,n,\lambda}(h,\omega)\le C\exp[-c\|h\|^2]\) on the annulus. The contribution of this first tail region is therefore bounded by \(C\int_{\|h\|>K_{j,n}}(1+\|h\|)\exp[-c\|h\|^2]dh=o(1)\).

For \(\|h\|>R_{j,n}/2\), write \(\beta=\widehat\beta_j(\omega)+h/\sqrt{n_j}\). If \((\beta,\gamma_j(\omega))\in\mathcal U_j(\delta_j)\), the same Taylor inequality gives an objective-function gap at least \(cR_{j,n}^2/4\). If \((\beta,\gamma_j(\omega))\notin\mathcal U_j(\delta_j)\), then, for all large \(n\) on the high-probability events, \(\|\beta-\beta_{j0}(\gamma_j(\omega))\|\ge\delta_j/2\). If this nonlocal region is nonempty, take \(\varepsilon=\delta_j/2\) and write \(\Delta_j:=\Delta_j(\delta_j/2)>0\). The uniform sample separation above gives \(q_{2,j,n_j}(X_j,(\beta_{j0}(\gamma_j(\omega)),\gamma_j(\omega)))-q_{2,j,n_j}(X_j,(\beta,\gamma_j(\omega)))\ge n_j\Delta_j/2\). Since the selected local maximizer is taken over a ball containing \(\beta_{j0}(\gamma_j(\omega))\), \(q_{2,j,n_j}(X_j,(\widehat\beta_j(\omega),\gamma_j(\omega)))\ge q_{2,j,n_j}(X_j,(\beta_{j0}(\gamma_j(\omega)),\gamma_j(\omega)))\), and hence the same \(n_j\Delta_j/2\) gap holds relative to \(\widehat\beta_j(\omega)\), uniformly in \(\omega\). If the nonlocal region is empty, set \(\Delta_j=1\); that part of the tail bound is vacuous. Hence the lower-level objective-function gap outside \(\|h\|\le R_{j,n}/2\) is at least \(D_{j,n}=\min(cR_{j,n}^2/4,n_j\Delta_j/2)\), with \(D_{j,n}-(d_{\beta j}+1)\log n_j/2\to\infty\). Since the transformed support is contained in a ball of radius \(C\sqrt{n_j}\), \(\int(1+\|h\|)dh\) over that support is at most \(Cn_j^{(d_{\beta j}+1)/2}\). The remaining tail contribution is therefore bounded by \(C n_j^{(d_{\beta j}+1)/2}\exp[-D_{j,n}]=o(1)\). Combining the denominator bound and the two tail regions gives the displayed first-moment tail ratio uniformly over retained values and \(\lambda\in\Lambda\).
\end{proof}

\begin{proposition}\label{prop:mixed-strong-equivalence}
Suppose Assumptions~\ref{ass:primitive-setup} and \ref{ass:mixed-strong} hold for group \(j\). Let \(P^S_{j,n,\lambda}(\cdot\mid\omega)\) be the conditional law under the full QBHM quasi-posterior of \(h_j=\sqrt{n_j}(\beta_j-\widehat\beta_j(\omega))\), given the retained variables \(\omega\). Let \(P^{S,0}_{j,n}(\cdot\mid\omega)\) be the corresponding unpooled conditional law, and write \(H_j(\omega)=H_j(\gamma_j(\omega))\). For the bounded-Lipschitz metric \(d_{\mathrm{BL}}\), the usual metric for weak convergence, uniformly over \(\lambda\in\Lambda\) and retained \(\omega\),
\[
        d_{\mathrm{BL}}\!\big(P^S_{j,n,\lambda}(\cdot\mid\omega),
        N(0,H_j(\omega)^{-1})\big)=o_p(1),
        \qquad
        d_{\mathrm{BL}}\!\big(P^{S,0}_{j,n}(\cdot\mid\omega),
        N(0,H_j(\omega)^{-1})\big)=o_p(1),
\]
and
\[
        \sup_{\lambda\in\Lambda,\omega}
        \left\|
        \int h_j\,P^S_{j,n,\lambda}(dh_j\mid\omega)
        \right\|
        +
        \sup_\omega
        \left\|
        \int h_j\,P^{S,0}_{j,n}(dh_j\mid\omega)
        \right\|
        =o_p(1).
\]
If \(\widehat\beta^{\mathrm{un}}_j\) denotes the reported unpooled profiled estimator and
\[
        \sup_{\lambda\in\Lambda}
        \int
        \sqrt{n_j}
        \left\|
        \widehat\beta_j(\gamma_j)-\widehat\beta^{\mathrm{un}}_j
        \right\|
        \Pi^{\mathrm{full},W}_{n,\lambda}(d\gamma)
        =o_p(1),
\]
then
\[
        \sup_{\lambda\in\Lambda}
        \sqrt{n_j}
        \left\|
        \tilde\beta^{\mathrm{full}}_{j,n,\lambda}
        -
        \widehat\beta^{\mathrm{un}}_j
        \right\|
        =o_p(1).
\]
\end{proposition}

\paragraph{Interpreting the profile-variation condition.}
Conditional on a retained weak value, \(\beta_j\) has root-\(n_j\) Gaussian uncertainty around \(\widehat\beta_j(\gamma_j)\). Marginally, however, uncertainty about \(\gamma_j\) can move that center. The profile-variation condition in Proposition~\ref{prop:mixed-strong-equivalence} requires this movement to be \(o_p(n_j^{-1/2})\), on average under the retained weak marginal, relative to the reported profiled estimator \(\widehat\beta^{\mathrm{un}}_j\). It is automatic when the reported regular target is fixed across the retained weak space, or when the retained weak marginal places asymptotically all relevant mass on values for which the profile agrees with \(\widehat\beta^{\mathrm{un}}_j\) at the \(n_j^{-1/2}\) scale. If instead the profile varies at order one over retained weak values with nonnegligible posterior mass, then the reported object is mixed-rate: conditionally on \(\gamma_j\) its regular uncertainty is root-\(n_j\), but marginally its uncertainty is driven by the retained weak coordinate.

\begin{proof}[Proof of Proposition \ref{prop:mixed-strong-equivalence}]
Work on the high-probability events supplied by Lemma~\ref{lem:mixed-strong-primitive}. Conditional on a retained value \(\omega\), write \(h=\sqrt{n_j}(\beta_j-\widehat\beta_j(\omega))\). After subtracting the conditional log kernel at \(h=0\), the conditional log density of \(h\) under the pooled quasi-posterior has the form
\(-\frac12h'\widehat H_{j,n}(\omega)h+\rho_{n,\lambda}(h,\omega)\), for \(\|h\|\le K_{j,n}\),
where \(\sup_{\lambda\in\Lambda,\omega,\|h\|\le K_{j,n}}|\rho_{n,\lambda}(h,\omega)|=o_p(1)\). This is the local quadratic expansion in Lemma~\ref{lem:mixed-strong-primitive} plus the uniform prior-flatness and hierarchy-flatness terms. Lemma~\ref{lem:local-moment-ratio}, applied conditionally and uniformly over the retained index set, gives a local conditional first moment \(o_p(1)\) and bounded-Lipschitz distance \(o_p(1)\) from \(N(0,\widehat H_{j,n}(\omega)^{-1})\). Since \(\sup_\omega\|\widehat H_{j,n}(\omega)-H_j(\gamma_j(\omega))\|=o_p(1)\) and the eigenvalues are bounded away from zero uniformly, the Gaussian laws with covariance matrices \(\widehat H_{j,n}(\omega)^{-1}\) and \(H_j(\gamma_j(\omega))^{-1}\) are uniformly close in bounded-Lipschitz distance. Indeed, with \(Z\sim N(0,I)\), \(d_{\mathrm{BL}}(N(0,\widehat H^{-1}),N(0,H^{-1}))\le \mathbb E\| (\widehat H^{-1/2}-H^{-1/2})Z\|\le C\|\widehat H-H\|\) on the uniform eigenvalue-bounded event, because inversion and the symmetric square-root map are Lipschitz there. To pass from local laws to full conditional laws, let \(r_{n,\lambda,\omega}\) be the tail-to-local mass ratio from Lemma~\ref{lem:mixed-strong-primitive}; then \(\sup_{\lambda,\omega}r_{n,\lambda,\omega}=o_p(1)\), and for every \(\|\varphi\|_\infty\le1\), \(|\int\varphi\,dP^{\mathrm{full}}_{n,\lambda,\omega}-\int\varphi\,dP^{\mathrm{loc}}_{n,\lambda,\omega}|\le2r_{n,\lambda,\omega}/(1+r_{n,\lambda,\omega})\). The same decomposition with \(\varphi(h)=h\), using the displayed first-moment tail ratio, gives an \(o_p(1)\) difference between full and local first moments. The unpooled conditional law is the same argument with \(\lambda=0\).

Let \(\Pi^{\mathrm{full},R}_{n,\lambda}\) denote the full marginal quasi-posterior of the retained variables \(\omega\), after integrating out the current regular block. The quasi-posterior mean decomposition is
\[
\begin{aligned}
\sqrt{n_j}(\tilde\beta^{\mathrm{full}}_{j,n,\lambda}-\widehat\beta^{\mathrm{un}}_j)
&=
\int \sqrt{n_j}(\widehat\beta_j(\gamma_j(\omega))-\widehat\beta^{\mathrm{un}}_j)
        \Pi^{\mathrm{full},R}_{n,\lambda}(d\omega) \\
&\quad+
\int \mathbb E_{n,\lambda}[h\mid\omega]\,
        \Pi^{\mathrm{full},R}_{n,\lambda}(d\omega).
\end{aligned}
\]
The second term is \(o_p(1)\) uniformly in \(\lambda\) by the conditional first-moment bound. The integrand in the first term depends on retained variables only through \(\gamma_j(\omega)\), so its integral under \(\Pi^{\mathrm{full},R}_{n,\lambda}\) equals the corresponding integral under the retained weak marginal. The profile-variation condition in the proposition therefore makes the first term \(o_p(1)\), uniformly over \(\lambda\in\Lambda\), which proves the reported strong-block equivalence.
\end{proof}

Let \(\Gamma_0:=\prod_{j:d_{\gamma j}>0}\Gamma_j\), and let \(\mu_0:=\prod_{j:d_{\gamma j}>0}\mu_{\gamma j}\) denote the retained weak-coordinate prior measure. The variable \(\gamma\in\Gamma_0\) is the stacked retained weak value in the profiled problem. Let \(\alpha^\circ(\gamma)\) denote the vector obtained by substituting the regular population profiles \(\beta_{j0}(\gamma_j)\) and retaining \(\gamma\). Let \(q_1^\circ(\gamma,\theta;z)=q_1(\alpha^\circ(\gamma),\theta;z)\), \(r_\lambda^\circ(\gamma)=\int_\Theta\exp[\lambda q_1^\circ(\gamma,\theta;z)]\pi_\theta(d\theta)\), and define the deterministic profile factor
\(C_S(\gamma)=\prod_{j:d_{\beta j}>0}p_{\beta j\mid\gamma j}(\beta_{j0}(\gamma_j)\mid\gamma_j)\det(H_j(\gamma_j))^{-1/2}\),
with the empty product interpreted as one. Let \(C_W\) denote any deterministic positive continuous factor generated by the weak-GMM reduction, and put \(C(\gamma)=C_S(\gamma)C_W(\gamma)\). Stochastic \(\gamma\)-dependent score-square terms generated by sample profiling are part of the profiled weak-GMM criterion below, not part of \(C\).

\begin{proposition}\label{prop:mixed-primitive-prior}
Suppose Assumptions~\ref{ass:primitive-setup}, \ref{ass:mixed-strong}, and \ref{ass:primitive-weak} hold for the profiled weak experiment. Suppose \(C_W\) is continuous and bounded away from zero and infinity on \(\Gamma_0\). Then \(C\) is continuous and bounded away from zero and infinity on \(\Gamma_0\), and the mixed prior path \(\pi_\lambda^\circ(d\gamma)\propto C(\gamma)r_\lambda^\circ(\gamma)\mu_0(d\gamma)\) is proper and weakly continuous on \(\Lambda\).
\end{proposition}

\begin{proof}[Proof of Proposition \ref{prop:mixed-primitive-prior}]
The maximum-theorem argument in Lemma~\ref{lem:mixed-strong-primitive} gives continuity of each profile map \(\gamma_j\mapsto\beta_{j0}(\gamma_j)\). Assumption~\ref{ass:mixed-strong} gives continuity of \(H_j(\gamma_j)\) and a uniform lower eigenvalue bound. Since \(\Gamma_j\) is compact, \(\det(H_j(\gamma_j))^{-1/2}\) is continuous and bounded above and away from zero. The conditional-prior clause of Assumption~\ref{ass:mixed-strong} gives positivity and continuity of \(p_{\beta j\mid\gamma j}(\beta_{j0}(\gamma_j)\mid\gamma_j)\). Therefore each group factor in \(C_S\) is continuous, positive, and bounded above and away from zero, and the finite product \(C_S\) has the same properties on compact \(\Gamma_0\). Multiplication by the specified factor \(C_W\) gives the stated properties of \(C\).

The map \(\gamma\mapsto\alpha^\circ(\gamma)\) is continuous because its regular coordinates are continuous profile maps and its weak coordinates are identities. Continuity of \(q_1\) on the compact support therefore implies joint continuity of \((\gamma,\theta,\lambda)\mapsto\exp[\lambda q_1^\circ(\gamma,\theta;z)]\), and the function is bounded above and away from zero on \(\Gamma_0\times\Theta\times\Lambda\). Dominated convergence with respect to \(\pi_\theta\) gives joint continuity and strict positivity of \(r_\lambda^\circ(\gamma)\). Lemma~\ref{lem:primitive-prior}, applied with \(\Gamma_0\), \(\mu_0\), and \(w_\lambda(\gamma)=C(\gamma)r_\lambda^\circ(\gamma)\), gives properness and weak continuity of \((\pi_\lambda^\circ)_{\lambda\in\Lambda}\).
\end{proof}

\begin{proposition}\label{prop:mixed-weak-marginal}
Suppose the local-reduction conditions in Lemma~\ref{lem:weak-laplace-reduction} are generated by Assumption~\ref{ass:mixed-strong}. Suppose the profiled weak lower-level criterion, including any \(\gamma\)-dependent stochastic score-square term from regular-block profiling, has the reduced weak-GMM form \(-Q_n^\circ(\gamma)/2\) up to a constant independent of \(\gamma\), with any deterministic weak-GMM baseline factor collected in \(C_W\), and suppose Proposition~\ref{prop:mixed-primitive-prior} holds. Let
\[
        \Pi^{\circ}_{n,\lambda}(d\gamma\mid X,z)
        =
        \frac{\exp[-Q_n^\circ(\gamma)/2]C(\gamma)r_\lambda^\circ(\gamma)\mu_0(d\gamma)}
        {\int_{\Gamma_0}\exp[-Q_n^\circ(u)/2]C(u)r_\lambda^\circ(u)\mu_0(du)}
\]
be the resulting reduced weak quasi-posterior. Then the retained weak marginal of the full QBHM quasi-posterior is asymptotically equivalent, uniformly over \(\lambda\in\Lambda\), to \(\Pi^{\circ}_{n,\lambda}\):
\[
        \sup_{\lambda\in\Lambda}
        \|\Pi^{\mathrm{full},W}_{n,\lambda}-\Pi^{\circ}_{n,\lambda}\|_{TV}=o_p(1).
\]
\end{proposition}

\begin{proof}[Proof of Proposition \ref{prop:mixed-weak-marginal}]
Apply Lemma~\ref{lem:weak-laplace-reduction} to the finite collection of regular blocks generated by Lemma~\ref{lem:mixed-strong-primitive}. For a fixed retained value \(\gamma\) and hierarchy value \(\theta\), the local changes of variables in the regular coordinates give, up to a \(1+o_p(1)\) relative error uniform in \((\gamma,\theta,\lambda)\), the post-Laplace unnormalized kernel
\[
\begin{aligned}
&\exp[-Q_n^\circ(\gamma)/2]
  \exp[\lambda q_1(\widehat\alpha_n^\circ(\gamma),\theta;z)]
  C_W(\gamma) \\
&\quad\times
  \prod_{j:d_{\beta j}>0}
  p_{\beta j\mid\gamma j}(\beta_{j0}(\gamma_j)\mid\gamma_j)
  \det(H_j(\gamma_j))^{-1/2}
  \mu_0(d\gamma)\pi_\theta(d\theta),
\end{aligned}
\]
The replacement of sample profile factors by their population-profile limits is uniform: writing \(\widehat\beta_j(\gamma_j)\) and \(\widehat H_{j,n}(\gamma_j)\) for the sample profile and Hessian from Lemma~\ref{lem:mixed-strong-primitive} at retained value \(\gamma_j\), continuity and positivity of the conditional prior density give \(\sup_{\gamma_j\in\Gamma_j}|p_{\beta j\mid\gamma j}(\widehat\beta_j(\gamma_j)\mid\gamma_j)/p_{\beta j\mid\gamma j}(\beta_{j0}(\gamma_j)\mid\gamma_j)-1|=o_p(1)\), while \(\sup_{\gamma_j\in\Gamma_j}|\det(\widehat H_{j,n}(\gamma_j))^{-1/2}/\det(H_j(\gamma_j))^{-1/2}-1|=o_p(1)\) because the determinant map is Lipschitz on uniformly eigenvalue-bounded positive-definite matrices.
The kernel is multiplied by constants \(\prod_{j:d_{\beta j}>0}n_j^{-d_{\beta j}/2}(2\pi)^{d_{\beta j}/2}\) and sample objective constants that do not depend on \((\gamma,\theta,\lambda)\). These constants cancel from normalized retained weak marginals. The displayed deterministic \(\gamma\)-dependent prior-density and determinant terms are exactly \(C_S(\gamma)\), and multiplication by the deterministic weak-GMM baseline \(C_W(\gamma)\) gives \(C(\gamma)\). Thus, after the regular coordinates are integrated out, the retained weak marginal is uniformly equivalent to the law with unnormalized measure \(\exp[-Q_n^\circ(\gamma)/2]C(\gamma)\exp[\lambda q_1(\widehat\alpha_n^\circ(\gamma),\theta;z)]\mu_0(d\gamma)\pi_\theta(d\theta)\),
where \(\widehat\alpha_n^\circ(\gamma)\) denotes the vector formed from the sample regular profiles and the retained weak coordinate. Lemma~\ref{lem:mixed-strong-primitive} gives hierarchy flatness uniformly over the retained set. Equivalently, with \(\Delta_n=\sup_{\gamma,\theta}|q_1(\widehat\alpha_n^\circ(\gamma),\theta;z)-q_1(\alpha^\circ(\gamma),\theta;z)|\), we have \(\Delta_n=o_p(1)\), and hence \(e^{-\bar\lambda\Delta_n}\le \exp[\lambda q_1(\widehat\alpha_n^\circ(\gamma),\theta;z)]/\exp[\lambda q_1(\alpha^\circ(\gamma),\theta;z)]\le e^{\bar\lambda\Delta_n}\), uniformly over \((\gamma,\theta,\lambda)\). Thus replacing \(\widehat\alpha_n^\circ(\gamma)\) by \(\alpha^\circ(\gamma)\) changes the preceding unnormalized density by a multiplicative factor \(1+o_p(1)\), uniformly over \((\gamma,\theta,\lambda)\). The relative-density normalization argument in Lemma~\ref{lem:weak-laplace-reduction} then gives the same \(o_p(1)\) total-variation error after normalization.

Integrating out \(\theta\) leaves an unnormalized retained weak measure with density factor \(\exp[-Q_n^\circ(\gamma)/2]C(\gamma)r_\lambda^\circ(\gamma)\) with respect to \(\mu_0(d\gamma)\).
Normalizing this measure gives \(\Pi^\circ_{n,\lambda}\), and because the relative errors are uniform in \(\lambda\in\Lambda\) and the finite products preserve uniform relative error, the displayed total-variation convergence follows.
\end{proof}

\begin{corollary}\label{cor:mixed-reported-targets}
Suppose the conditions of Propositions~\ref{prop:mixed-strong-equivalence}, \ref{prop:mixed-primitive-prior}, and \ref{prop:mixed-weak-marginal} hold, including the profile-variation condition in Proposition~\ref{prop:mixed-strong-equivalence}. Suppose the profiled weak process satisfies \((g_n^\circ,\widehat W_n^\circ)\dto(g^\circ,W^\circ)\) in the product sup-norm space, with continuous limiting sample paths. Let \(T^\circ_{n,\lambda}=\int_{\Gamma_0}\gamma\,\Pi^\circ_{n,\lambda}(d\gamma\mid X,z)\), and let \(t_\lambda^\circ(g^\circ)\) be the corresponding limiting quasi-posterior mean obtained by replacing \(Q_n^\circ\) with \(Q^\circ(\gamma)=g^\circ(\gamma)'W^\circ(\gamma)g^\circ(\gamma)\). Let \(D_{n,S}\) be the diagonal rate matrix for the selected regular coordinates, with all diagonal entries diverging, let \(\beta^S_0\) stack their profile estimands at the retained weak state \(\gamma^*\), and suppose \(\gamma^*\) satisfies the profiled drift restriction \(m^\circ(\gamma^*)=0\). Write \(\tilde\beta^S_{n,\lambda}\) and \(\tilde\gamma^W_{n,\lambda}\) for the selected regular and retained weak components of the full QBHM quasi-posterior mean. If
\[
        (D_{n,S}(\widehat\beta^{S,\mathrm{un}}_n-\beta^S_0),g_n^\circ,\widehat W_n^\circ)
        \dto
        (Z_S,g^\circ,W^\circ),
\]
then, for fixed conformable matrices \(R_S\) and \(R_W\),
\[
        \left(
        \begin{pmatrix}
        R_SD_{n,S}(\tilde\beta^S_{n,\lambda}-\beta^S_0)\\
        R_W(\tilde\gamma^W_{n,\lambda}-\gamma^*)
        \end{pmatrix}
        \right)_{\lambda\in\Lambda}
        \dto
        \left(
        \begin{pmatrix}
        R_SZ_S\\
        R_W(t_\lambda^\circ(g^\circ)-\gamma^*)
        \end{pmatrix}
        \right)_{\lambda\in\Lambda}
\]
in \(\ell^\infty(\Lambda,\mathbb R^{r_S+r_W})\), where \(r_S\) and \(r_W\) are the row dimensions of \(R_S\) and \(R_W\). On the original level scale, the selected regular component is \(o_p(1)\), so any nondegenerate MSE comparison generated by pooling is carried by the retained weak component.
\end{corollary}

\begin{proof}[Proof of Corollary \ref{cor:mixed-reported-targets}]
Let \(U_n=D_{n,S}(\widehat\beta^{S,\mathrm{un}}_n-\beta^S_0)\). Proposition~\ref{prop:mixed-primitive-prior} gives the proper weakly continuous profiled prior path \((\pi_\lambda^\circ)_{\lambda\in\Lambda}\). Applying Lemma~\ref{lem:posterior-continuity} with \(\Gamma=\Gamma_0\), \(x_n=g_n^\circ\), \(V_n=\widehat W_n^\circ\), \(x=g^\circ\), \(V=W^\circ\), and prior path \(\pi_\lambda^\circ\) gives \((T^\circ_{n,\lambda})_{\lambda\in\Lambda}\dto(t_\lambda^\circ(g^\circ))_{\lambda\in\Lambda}\). Since the assumed convergence of \((U_n,g_n^\circ,\widehat W_n^\circ)\) is joint, the same continuous-mapping argument applied to the joint vector gives
\[
        (U_n,(T^\circ_{n,\lambda})_{\lambda\in\Lambda})
        \dto
        (Z_S,(t_\lambda^\circ(g^\circ))_{\lambda\in\Lambda}).
\]

For the selected regular block, finite stacking of Proposition~\ref{prop:mixed-strong-equivalence} gives \(\sup_{\lambda\in\Lambda}\|D_{n,S}(\tilde\beta^S_{n,\lambda}-\widehat\beta^{S,\mathrm{un}}_n)\|=o_p(1)\). For the retained weak block, Proposition~\ref{prop:mixed-weak-marginal} and compactness of \(\Gamma_0\) give the same conclusion for means: if \(D_\Gamma=\sup_{\gamma\in\Gamma_0}\|\gamma\|\), then \(\|\int\gamma\,dP-\int\gamma\,dQ\|\le2D_\Gamma\|P-Q\|_{TV}\), so \(\sup_{\lambda\in\Lambda}\|\tilde\gamma^W_{n,\lambda}-T^\circ_{n,\lambda}\|=o_p(1)\). Therefore the displayed path with regular and weak scaling differs in sup norm by \(o_p(1)\) from
\[
\left(
\begin{pmatrix}
R_SU_n\\
R_W(T^\circ_{n,\lambda}-\gamma^*)
\end{pmatrix}
\right)_{\lambda\in\Lambda} .
\]
The latter path converges to the stated limit by the preceding joint convergence and the continuity of finite stacking and multiplication by fixed matrices. Slutsky's theorem proves the joint convergence with regular and weak rates.

For the statement on the original parameter level, the diverging diagonal entries imply \(D_{n,S}^{-1}\to0\) element by element, and \(D_{n,S}(\tilde\beta^S_{n,\lambda}-\beta^S_0)=U_n+o_p(1)\) uniformly over \(\lambda\). The regular component scaled at its own rate is therefore tight, and \(\sup_{\lambda\in\Lambda}\|\tilde\beta^S_{n,\lambda}-\beta^S_0\|=o_p(1)\). The selected regular block vanishes on the original level scale, while the retained weak block has the nondegenerate limit displayed above.
\end{proof}

\medskip

\subsection{MSE transfer}

The next proposition gives the finite-sample-to-limit MSE transfer used to connect finite-sequence risk to the weak-limit risk. Beyond weak convergence of the estimator path, the proof uses uniform integrability of the squared estimand error, so fix a weak-GMM sequence satisfying Assumption~\ref{ass:primitive-weak}. Let \((T_{n,\lambda})_{\lambda\in\Lambda}\) be a feasible estimator path, and let \((t_\lambda(g))_{\lambda\in\Lambda}\) be a measurable weak-limit path. Write \(Y_{n,\lambda}=B(T_{n,\lambda}-\alpha_W^*)\) and \(Y_\lambda=B(t_\lambda(g)-\alpha_W^*)\), where \(B\) has \(r\) rows. The MSE-transfer requirements are path convergence \((Y_{n,\lambda})_{\lambda\in\Lambda}\dto(Y_\lambda)_{\lambda\in\Lambda}\) in \(\ell^\infty(\Lambda,\mathbb R^r)\) and uniform integrability of \(\sup_{\lambda\in\Lambda}\|Y_{n,\lambda}\|^2\) under the local triangular-array law.

\begin{proposition}\label{prop:transfer}
Fix a weak-GMM sequence satisfying Assumption~\ref{ass:primitive-weak}, an estimator path \((T_{n,\lambda})_{\lambda\in\Lambda}\), a measurable weak-limit path \((t_\lambda(g))_{\lambda\in\Lambda}\), and a matrix \(B\) with \(r\) rows. Define \(Y_{n,\lambda}=B(T_{n,\lambda}-\alpha_W^*)\) and \(Y_\lambda=B(t_\lambda(g)-\alpha_W^*)\). Suppose \((Y_{n,\lambda})_{\lambda\in\Lambda}\dto(Y_\lambda)_{\lambda\in\Lambda}\) in \(\ell^\infty(\Lambda,\mathbb R^r)\) and \(\sup_{\lambda\in\Lambda}\|Y_{n,\lambda}\|^2\) is uniformly integrable. Then
\(\sup_{\lambda\in\Lambda}|\mathcal R_{n,\lambda}-M_{B,\lambda}|\to0\),
where \(\mathcal R_{n,\lambda}=\mathbb E_{P_{n,f}}\|B(T_{n,\lambda}-\alpha_W^*)\|^2\) and \(M_{B,\lambda}=M_B(t_\lambda;\alpha_W^*,m)\). If a comparator path has the corresponding MSE convergence and there are \(\Lambda_1\subset\Lambda\) and \(\eta>0\) such that \(\inf_{\lambda\in\Lambda_1}[M^0_{B,\lambda}-M_{B,\lambda}]\ge\eta\), then, for all sufficiently large \(n\), \(\inf_{\lambda\in\Lambda_1}[\mathcal R^0_{n,\lambda}-\mathcal R_{n,\lambda}]\ge\eta/2\).
\end{proposition}

\begin{proof}[Proof of Proposition \ref{prop:transfer}]
Set \(Y_n=(Y_{n,\lambda})_{\lambda\in\Lambda}\) and \(Y=(Y_\lambda)_{\lambda\in\Lambda}\). The path convergence requirement gives \(Y_n\dto Y\) in \(\ell^\infty(\Lambda,\mathbb R^r)\). The limit law is tight and may be taken to have separable support. By the Skorokhod--Dudley representation theorem for separable Borel laws on metric spaces, there is a probability space carrying versions with the same marginal laws and satisfying \(\|Y_n-Y\|_\infty\to0\) almost surely. This representation is used only to prove convergence of expectations of functions of the marginal laws.

For \(K<\infty\), let \(\psi_K(x)=x\wedge K\). On the representation just described,
\[
        \sup_{\lambda\in\Lambda}
        \left|\psi_K(\|Y_{n,\lambda}\|^2)-\psi_K(\|Y_\lambda\|^2)\right|
        \to0
        \qquad \text{almost surely}.
\]
Since the supremum is bounded by \(K\), dominated convergence gives
\[
        \mathbb E\sup_{\lambda\in\Lambda}
        \left|\psi_K(\|Y_{n,\lambda}\|^2)-\psi_K(\|Y_\lambda\|^2)\right|
        \to0 .
\]
Therefore \(\sup_{\lambda\in\Lambda}|\mathbb E\psi_K(\|Y_{n,\lambda}\|^2)-\mathbb E\psi_K(\|Y_\lambda\|^2)|\to0\) for each fixed \(K\).

It remains to remove the truncation uniformly in \(\lambda\). Put \(X_n=\sup_{\lambda\in\Lambda}\|Y_{n,\lambda}\|^2\) and \(X=\sup_{\lambda\in\Lambda}\|Y_\lambda\|^2\). Uniform integrability of \((X_n)\) is assumed, and almost-sure convergence of \(Y_n\) to \(Y\) implies \(X_n\to X\) almost surely. Vitali's theorem gives \(\mathbb E|X_n-X|\to0\), so \(X\) is integrable and \(\mathbb E[X\one(X>K)]\to0\) as \(K\to\infty\). For every \(n\) and \(\lambda\),
\(0\le \|Y_{n,\lambda}\|^2-\psi_K(\|Y_{n,\lambda}\|^2)\le X_n\one(X_n>K)\),
and the same bound holds for \(Y\) with \(X\). Hence
\[
\begin{aligned}
\limsup_n\sup_{\lambda\in\Lambda}
\left|\mathbb E\|Y_{n,\lambda}\|^2-\mathbb E\|Y_\lambda\|^2\right|
&\le
\limsup_n\mathbb E[X_n\one(X_n>K)]
+
\mathbb E[X\one(X>K)] .
\end{aligned}
\]
Letting \(K\to\infty\) and using uniform integrability of \((X_n)\) proves the uniform MSE convergence.

For the comparator claim, apply the same argument to the comparator path. If the limiting margin is at least \(\eta\) on \(\Lambda_1\), then the two uniform MSE approximation errors are smaller than \(\eta/4\) for all sufficiently large \(n\). Subtracting the approximations gives \(\mathcal R^0_{n,\lambda}-\mathcal R_{n,\lambda}\ge\eta/2\) uniformly over \(\lambda\in\Lambda_1\).
\end{proof}

\medskip

Write \(Z_n=\sup_{\lambda\in\Lambda}\|B(T_{n,\lambda}-\alpha_W^*)\|\). The following two routes verify the tail condition in Proposition~\ref{prop:transfer}.

\begin{lemma}\label{lem:ui-primitive}
Let \(Z_n=\sup_{\lambda\in\Lambda}\|B(T_{n,\lambda}-\alpha_W^*)\|\) under the weak-GMM triangular-array law. The sequence \(Z_n^2\) is uniformly integrable if either \(Z_n\le C\) almost surely for all sufficiently large \(n\) and the finitely many initial \(Z_n^2\) are integrable, or, for some \(\delta>0\), \(\sup_n\mathbb E_{P_{n,f}}Z_n^{2+\delta}<\infty\).
\end{lemma}

\begin{proof}[Proof of Lemma \ref{lem:ui-primitive}]
Under the boundedness route, for every \(M>C^2\), \(Z_n^2\one(Z_n^2>M)=0\) for all sufficiently large \(n\). A finite collection of integrable random variables is uniformly integrable, so the initial indices do not affect the conclusion. Under the \(2+\delta\) moment route, for every \(M>0\), $\mathbb E[Z_n^2\one(Z_n^2>M)]
        \le
        M^{-\delta/2}\mathbb E Z_n^{2+\delta}$. Taking the supremum over \(n\) and then sending \(M\to\infty\) gives uniform integrability by the de la Vallee--Poussin criterion.
\end{proof}

\medskip

\subsection{Pointwise MSE}

The exact pointwise identity is algebraic and does not use a local quadratic approximation.

\begin{proof}[Proof of Proposition \ref{prop:identity}]
Fix \(\lambda\), let \(e_\lambda(g)=B(t_\lambda(g)-\alpha_W^*)\), and note that the two finite-MSE assumptions are exactly \(e_0,e_\lambda\in L^2(P_m)\).  Since \(Bd_\lambda=e_\lambda-e_0\), the vector \(Bd_\lambda\) is square integrable.  Cauchy--Schwarz gives
\[
        \mathbb E_m|e_0(g)'Bd_\lambda(g)|
        \le [\mathbb E_m\|e_0(g)\|^2]^{1/2}
             [\mathbb E_m\|Bd_\lambda(g)\|^2]^{1/2}<\infty .
\]
The pointwise quadratic identity
\(\|e_\lambda(g)\|^2-\|e_0(g)\|^2=2e_0(g)'Bd_\lambda(g)+\|Bd_\lambda(g)\|^2\)
therefore has integrable terms on the right side, and integrating gives the identity in the proposition.

For any square-integrable vector $Y$, \(\mathbb E\|Y\|^2=\operatorname{tr}(\operatorname{Var}(Y))+\|\mathbb EY\|^2\).  Applying this identity to $e_0(g)$ and $e_\lambda(g)$ gives \(M_{B,\lambda}-M_{B,0}=\operatorname{tr}[\operatorname{Var}(e_\lambda)-\operatorname{Var}(e_0)]+\|\mathbb E e_\lambda\|^2-\|\mathbb E e_0\|^2\), which is the bias--variance formulation in the main text.  Pooling lowers pointwise asymptotic MSE at the chosen $\lambda$ exactly when the right side of the identity in the proposition is negative.  The same algebra is uniform over a subset of $\Lambda$ whenever the strict inequality has a positive margin on that subset and the relevant MSEs are finite there.
\end{proof}

\medskip

The next proposition gives the two sufficient conditions discussed after Proposition \ref{prop:identity}.  The second part uses an abstract prior-concentration index $v$, because the parameter that sends the prior to a point mass need not be the same as the pooling-strength parameter $\lambda$.

Fix a weak-GMM limit experiment $(\alpha_W^*,m)$ and a comparator $t_0(g)$.  In local coordinates write $h=\alpha_W-\alpha_W^*$.  For a concentrating prior induced by the hierarchy, let $v\downarrow0$ be the concentration index, let $h_c$ be its local center, and let $\delta_v(g)$ be the weak-GMM quasi-posterior mean in local coordinates.

\begin{proposition}\label{prop:sufficient}
Fix a weak-GMM limit experiment \((\alpha_W^*,m)\), a comparator \(t_0(g)\), a matrix \(B\), a concentrating-prior center \(h_c\), and the local quasi-posterior mean \(\delta_v(g)\) defined from the concentrating prior \(\Pi_v\). For any measurable pooled rule \(t_\lambda\), if \(M_B(t_0;\alpha_W^*,m)=\infty\) and \(M_B(t_\lambda;\alpha_W^*,m)<\infty\), then \(t_\lambda\) has smaller pointwise asymptotic MSE than \(t_0\). For the concentrating-prior rule \(\delta_v\), if \(\delta_v\to h_c\) in \(L^2(P_m)\) and \(\|Bh_c\|^2<M_B(t_0;\alpha_W^*,m)\), then there exists \(v^*>0\) such that, for every \(0<v<v^*\), \(\mathbb{E}_m\|B\delta_v(g)\|^2<M_B(t_0;\alpha_W^*,m)\).
\end{proposition}

\begin{proof}[Proof of Proposition \ref{prop:sufficient}]
The first claim is an identity about extended nonnegative risks.  The risk of $t_0$ is $\infty$ by assumption, while the risk of $t_\lambda$ is a finite nonnegative number, so the latter is strictly smaller in the extended order.

For the second claim, convergence \(\delta_v\to h_c\) in $L^2(P_m)$ implies \(B\delta_v\to Bh_c\) in $L^2(P_m)$, because \(\|B(\delta_v-h_c)\|\le\|B\|_{\mathrm{op}}\|\delta_v-h_c\|\).  Putting $X_v=B\delta_v(g)-Bh_c$, the deterministic inequality
\(|\|Bh_c+X_v\|^2-\|Bh_c\|^2|\le \|X_v\|^2+2\|Bh_c\|\|X_v\|\)
and Cauchy--Schwarz imply
\[
        \left|\mathbb E_m\|B\delta_v(g)\|^2-\|Bh_c\|^2\right|
        \le \mathbb E_m\|X_v\|^2+2\|Bh_c\|[\mathbb E_m\|X_v\|^2]^{1/2}\to0 .
\]
Let \(\eta=M_B(t_0;\alpha_W^*,m)-\|Bh_c\|^2>0\), with the convention that \(\eta=\infty\) if the comparator risk is infinite.  For all sufficiently small $v$, the risk of $\delta_v$ is at most \(\|Bh_c\|^2+\eta/2\) when \(\eta<\infty\), and is finite when \(\eta=\infty\).  In both cases it is strictly smaller than \(M_B(t_0;\alpha_W^*,m)\).
\end{proof}

\medskip

For prior laws $\Pi_v$ on $\calH$, define the local quasi-posterior mean \(\delta_v(g)=\int_\calH h\exp[-Q_g(h)/2]\Pi_v(dh)/\int_\calH \exp[-Q_g(u)/2]\Pi_v(du)\).  The mean-squared prior concentration condition is \(\int_\calH\|h-h_c\|^2\Pi_v(dh)\to0\).  It is satisfied, for example, by truncated $N(h_c,vV)$ laws with fixed positive definite $V$ when $h_c$ is an interior point of $\calH$, and by laws of $h_c+vU$ restricted to $\calH$ whenever $\mathbb{E}\|U\|^2<\infty$ and the restriction probability is bounded away from zero.

\begin{lemma}\label{lem:concentration}
Let \(\Pi_v\) be prior laws on compact \(\calH\subset\mathbb R^{d_W}\), let \(h_c\in\calH\), and define \(\delta_v(g)=\int_\calH h\exp[-Q_g(h)/2]\Pi_v(dh)/\int_\calH \exp[-Q_g(u)/2]\Pi_v(du)\). Suppose \(Q_g\) is finite and continuous on \(\calH\) for \(P_m\)-almost every \(g\), and suppose \(\int_\calH\|h-h_c\|^2\Pi_v(dh)\to0\). Then \(\delta_v(g)\to h_c\) for \(P_m\)-almost every \(g\) and \(\delta_v\to h_c\) in \(L^2(P_m)\).
\end{lemma}

\begin{proof}[Proof of Lemma \ref{lem:concentration}]
Fix a process path $g$ for which $Q_g$ is finite and continuous on $\calH$, and write \(L_g(h)=\exp[-Q_g(h)/2]\).  Then $L_g$ is continuous on compact $\calH$, bounded above by some finite $\bar L_g$, and strictly positive at $h_c$.  Let \(D=\sup_{x,y\in\operatorname{co}(\calH)}\|x-y\|\), which is finite because $\calH$ is compact.  Given $\varepsilon>0$, choose $r>0$ with $r<\varepsilon$ such that \(\|h-h_c\|<r\) implies \(L_g(h)\ge L_g(h_c)/2\).  Markov's inequality and mean-squared prior concentration give \(\Pi_v(\|h-h_c\|\ge r)\le r^{-2}\int\|h-h_c\|^2\Pi_v(dh)\to0\).

For all sufficiently small $v$, the quasi-posterior denominator is at least \([L_g(h_c)/2]\Pi_v(\|h-h_c\|<r)\), which is bounded away from zero.  Decompose the quasi-posterior mean over the ball and its complement: the ball part is within $r<\varepsilon$ of $h_c$ after normalization, and the complement contributes at most its quasi-posterior probability times the diameter bound $D$.  Since both $\delta_v(g)$ and $h_c$ lie in the convex hull of $\calH$,
\[
        \|\delta_v(g)-h_c\|
        \le \varepsilon
        +D\frac{\int_{\|h-h_c\|\ge r}L_g(h)\Pi_v(dh)}
        {\int_\calH L_g(u)\Pi_v(du)}.
\]
The numerator in the fraction is at most \(\bar L_g\Pi_v(\|h-h_c\|\ge r)\), while the denominator is at least \([L_g(h_c)/2]\Pi_v(\|h-h_c\|<r)\) for all sufficiently small \(v\). The fraction therefore tends to zero, and since $\varepsilon$ is arbitrary, \(\delta_v(g)\to h_c\) on this path.  The bound \(\|\delta_v(g)-h_c\|^2\le D^2\) and bounded convergence over the law of $g$ give \(\delta_v\to h_c\) in $L^2(P_m)$.

For the examples, let \(Z_v\sim N(h_c,vV)\). If \(h_c\) is in the interior of \(\calH\), then \(P(Z_v\in\calH)\to1\), and the truncated law satisfies \(\mathbb E[\|Z_v-h_c\|^2\mid Z_v\in\calH]\le v\operatorname{tr}(V)/P(Z_v\in\calH)\to0\). If \(h=h_c+vU\) is restricted to \(\calH\) and the restriction probability is bounded away from zero, then \(\mathbb{E}[\|h-h_c\|^2\mid h\in\calH]\le v^2\mathbb{E}\|U\|^2/P(h_c+vU\in\calH)\), which tends to zero.  Both examples satisfy the mean-squared prior concentration condition.
\end{proof}

\medskip

\subsection{Local-pooling expansions}

The local results work with the weak-GMM limit experiment fixed and differentiate the quasi-posterior mean with respect to the pooling strength. The first proposition treats a general bounded prior score; the quadratic formulas follow by substituting the score generated by the hierarchy.

\begin{proposition}\label{prop:small}
Fix a weak-GMM limit experiment $(\alpha_W^*,m)$ generated by the reduced weak-GMM limit in Proposition~\ref{prop:weak-process-main}.  Suppose \(\calH\subset\mathbb R^{d_W}\) is compact with positive finite coordinate measure, $Q_g$ is finite and continuous on \(\calH\) for $P_m$-almost every $g$, and Assumption~\ref{ass:smooth-prior} holds on \(\calH\).  Then, as $\lambda\downarrow0$,
\[
        M_{\lambda,B}(\alpha_W^*,m)
        =M_{0,B}(\alpha_W^*,m)
        +2\lambda \mathbb{E}_m\left[(B\bar h(g))'B\operatorname{Cov}_{\nu_0^g}(h,\ell_0(h))\right]
        +o(\lambda).
\]
If the expectation in the display is strictly negative, then small positive pooling strictly lowers pointwise asymptotic MSE.
\end{proposition}

\begin{proof}[Proof of Proposition \ref{prop:small}]
Fix a process path $g$ in the full-probability set on which $Q_g$ is finite and continuous on compact $\calH$.  Write
\(Z_\lambda(g)=\int_\calH \exp[-Q_g(h)/2]p_\lambda(h)dh\), and \(N_\lambda(g)=\int_\calH h\exp[-Q_g(h)/2]p_\lambda(h)dh\).
For this fixed path, continuity of $Q_g$ and the uniform lower and upper bounds on $p_\lambda$ imply
\(0<\inf_{0\le\lambda\le\varepsilon}Z_\lambda(g)\le \sup_{0\le\lambda\le\varepsilon}Z_\lambda(g)<\infty\).
This lower bound is pathwise; the proof below does not require it to be uniform over process paths.  Let \(\ell_\lambda(h)=\partial_\lambda\log p_\lambda(h)\).  Since \(p_\lambda\) and \(\ell_\lambda\) are uniformly bounded on \([0,\varepsilon]\times\calH\), dominated convergence permits differentiation of $N_\lambda(g)$ and $Z_\lambda(g)$ under the integral sign.  The quotient rule gives
\[
\begin{aligned}
        \partial_\lambda\delta_\lambda(g)
        &=\frac{\int h\exp[-Q_g(h)/2]p_\lambda(h)\ell_\lambda(h)dh}{Z_\lambda(g)}
          -\delta_\lambda(g)
          \frac{\int \exp[-Q_g(h)/2]p_\lambda(h)\ell_\lambda(h)dh}{Z_\lambda(g)} \\
        &=\operatorname{Cov}_{\nu_\lambda^g}(h,\ell_\lambda(h)).
\end{aligned}
\]
At $\lambda=0$, this derivative is \(\dot\delta_0(g)=\operatorname{Cov}_{\nu_0^g}(h,\ell_0(h))\).  The pointwise convergence \(\ell_\lambda(h)\to \ell_0(h)\), the uniform bound on \(\ell_\lambda\), and the same domination argument imply \(\partial_\lambda\delta_\lambda(g)\to\dot\delta_0(g)\) as $\lambda\downarrow0$.

The derivative is bounded by a constant depending only on $\calH$ and the score bound.  If \(D_\calH=\sup_{h,u\in\calH}\|h-u\|\) and \(L=\sup_{\lambda,h}|\ell_\lambda(h)|\), then for every probability law supported on $\calH$,
\[
        \|\operatorname{Cov}(h,\ell_\lambda(h))\|
        =\|\mathbb E[(h-\mathbb Eh)(\ell_\lambda(h)-\mathbb E\ell_\lambda)]\|
        \le 2D_\calH L .
\]
Hence \((\delta_\lambda(g)-\delta_0(g))/\lambda\to\dot\delta_0(g)\) and the quotient is uniformly bounded.  Because $\calH$ is compact, \(\delta_\lambda(g)\) and \(\delta_0(g)\) are uniformly bounded.  Dividing
\[
        \|B\delta_\lambda(g)\|^2-\|B\delta_0(g)\|^2
        =2(B\delta_0(g))'B[\delta_\lambda(g)-\delta_0(g)]
          +\|B[\delta_\lambda(g)-\delta_0(g)]\|^2
\]
by $\lambda$ and applying dominated convergence over the law of $g$ gives the right derivative of \(M_{\lambda,B}(\alpha_W^*,m)\) at zero, which yields the displayed first-order expansion.  If this derivative is strictly negative, the $o(\lambda)$ remainder is dominated by half of the negative linear term for all sufficiently small positive $\lambda$, so small positive pooling lowers pointwise asymptotic MSE.
\end{proof}

\medskip

For quadratic hierarchy scores, define \(A_B(P)=\mathbb{E}_m[(B\bar h(g))'B[\Omega(g)P\bar h(g)+\tau_P(g)/2]]\) and \(L_B(P)=\mathbb{E}_m[P\Omega(g)B'B\bar h(g)]\).

\begin{proposition}\label{prop:small-quadratic}
Fix the weak-GMM limit experiment and local coordinates as in Proposition~\ref{prop:small}: \(\calH\subset\mathbb R^{d_W}\) is compact with positive finite coordinate measure, \(Q_g\) is finite and continuous on \(\calH\) for \(P_m\)-almost every \(g\), and Assumption~\ref{ass:smooth-prior} holds on \(\calH\).  Suppose the local hierarchy score is \(\ell_0(h)=c+h_c'Ph-h'Ph/2\), with \(P=P'\succeq0\).  Then
\[
        M_{\lambda,B}(\alpha_W^*,m)
        =M_{0,B}(\alpha_W^*,m)+2\lambda[h_c'L_B(P)-A_B(P)]+o(\lambda).
\]
Small positive pooling lowers asymptotic MSE whenever $h_c'L_B(P)<A_B(P)$.
\end{proposition}

\begin{proof}[Proof of Proposition \ref{prop:small-quadratic}]
Constants have zero covariance, so the scalar $c$ does not affect the derivative.  Fix $g$ and write $\bar h=\bar h(g)$, $\Omega=\Omega(g)$, and $u=h-\bar h$, with expectations in this paragraph taken under $\nu_0^g$.  Then \(\mathbb E_{\nu_0^g}u=0\).  Since $P=P'$, \(\operatorname{Cov}_{\nu_0^g}(h,h_c'Ph)=\mathbb E_{\nu_0^g}[uu']P h_c=\Omega P h_c\).  Also \(h'Ph=(\bar h+u)'P(\bar h+u)=\bar h'P\bar h+2\bar h'Pu+u'Pu\).  The constant term again drops out, and \(\operatorname{Cov}_{\nu_0^g}(h,h'Ph)=\mathbb E_{\nu_0^g}[u(2\bar h'Pu+u'Pu)]=2\Omega P\bar h+\tau_P(g)\).  Therefore \(\operatorname{Cov}_{\nu_0^g}(h,\ell_0(h))=\Omega(g)P[h_c-\bar h(g)]-\tau_P(g)/2\).  Substituting this expression into Proposition \ref{prop:small} gives
\[
\begin{aligned}
        \left.\partial_\lambda M_{\lambda,B}(\alpha_W^*,m)\right|_{0+}
        &=2\mathbb E_m[(B\bar h)'B\Omega P h_c]
          -2\mathbb E_m\left[(B\bar h)'B\left[\Omega P\bar h+\frac12\tau_P\right]\right] \\
        &=2[h_c'L_B(P)-A_B(P)].
\end{aligned}
\]
The second equality uses symmetry of $\Omega$ and $P$ to write $(B\bar h)'B\Omega P h_c=h_c'P\Omega B'B\bar h$ and then takes expectations.  The stated quadratic expansion and the strict-improvement condition follow.
\end{proof}

\medskip

The derivative condition can be made uniform around a pooling manifold.  The following theorem states this for a quadratic hierarchy that penalizes departures from a linear or affine pooling restriction.  Let $P=P'\succeq0$ and let $\bar\alpha_W\in\R^{d_W}$ be a center.  The associated pooling manifold is \(\calM_{P,\bar\alpha_W}=\{\alpha_W\in\calW:P(\alpha_W-\bar\alpha_W)=0\}\).  For the common-mean hierarchy with equal weights, $P=I_J-J^{-1}\one\one'$ and $\bar\alpha_W=0$, so $\|P\alpha_W\|$ is the size of the cross-group deviation.  A fixed-center hierarchy is the special case in which $P$ is positive definite and \(\calM_{P,\bar\alpha_W}=\{\bar\alpha_W\}\).

To isolate the uniform pointwise implication, take the prior density to be proportional to \(\exp[\lambda a_P(\alpha_W)]p_0(\alpha_W)\), where the normalized form is \(p_\lambda(\alpha_W)=\exp[\lambda a_P(\alpha_W)]p_0(\alpha_W)/\int_{\calW}\exp[\lambda a_P(u)]p_0(u)\mu_W(du)\).  Here \(a_P(\alpha_W):=-(\alpha_W-\bar\alpha_W)'P(\alpha_W-\bar\alpha_W)/2\), \(0\le\lambda\le\bar\lambda\), and $p_0$ is the unpooled prior density.  The normalizing constant in $p_\lambda$ is irrelevant for the derivative calculation below, because the quasi-posterior under $p_\lambda$ is an exponential reweighting of the unpooled quasi-posterior.  For a weak-GMM limit experiment $e=(\alpha_W^*,m)$, write \(M_{\lambda,B}(e)=\mathbb{E}_m\|B(t_\lambda(g)-\alpha_W^*)\|^2\).  Under the unpooled quasi-posterior $\Pi_0(d\alpha_W\mid g)$, define \(r_e(g)=t_0(g)-\alpha_W^*\), \(\Omega_e(g)=\int_{\calW}(\alpha_W-t_0(g))(\alpha_W-t_0(g))'\Pi_0(d\alpha_W\mid g)\), and \(\tau_{e,P}(g)=\int_{\calW}(\alpha_W-t_0(g))(\alpha_W-t_0(g))'P(\alpha_W-t_0(g))\Pi_0(d\alpha_W\mid g)\).  Let \(A_e(P)=\mathbb{E}_m[(B r_e(g))'B[\Omega_e(g)P r_e(g)+\tau_{e,P}(g)/2]]\) and \(L_e(P)=\mathbb{E}_m[\Omega_e(g)B'B r_e(g)]\).  The scalar $A_e(P)$ is the first-order derivative gain that would obtain if the truth lay exactly on the pooling manifold.  It need not be positive in a fully nonlinear weak-GMM problem; positivity is a substantive shrinkage-gain condition.  The vector $L_e(P)$ determines how much that gain is offset when the truth is close to, but not exactly on, the manifold.

The proof of Theorem~\ref{thm:uniform-local} uses compactness of $\calW$ to bound all quasi-posterior moments uniformly over the weak-GMM limit experiments specified in the theorem, while the quadratic form of $a_P$ gives a uniformly bounded exponential reweighting on every compact interval for $\lambda$.

\begin{proof}[Proof of Theorem \ref{thm:uniform-local}]
Fix \(e=(\alpha_W^*,m)\in\mathcal E\) and a realized process $g$ for which the unpooled quasi-posterior is well defined.  The normalizing constant in the prior $p_\lambda$ is a scalar depending on $\lambda$ but not on $\alpha_W$, so it cancels from the quasi-posterior ratio.  Hence
\[
        \Pi_\lambda(d\alpha_W\mid g)
        =\frac{\exp[\lambda a_P(\alpha_W)]\Pi_0(d\alpha_W\mid g)}
        {\int_{\calW}\exp[\lambda a_P(u)]\Pi_0(du\mid g)},
        \qquad
        a_P(\alpha_W)=-\frac12(\alpha_W-\bar\alpha_W)'P(\alpha_W-\bar\alpha_W).
\]
Because $\calW$ is compact and $a_P$ is continuous, \(A_0:=\sup_{\alpha_W\in\calW}|a_P(\alpha_W)|<\infty\).  For a bounded measurable scalar function $f$, define
\[
        T_\lambda(f)=\frac{\int f(\alpha_W)\exp[\lambda a_P(\alpha_W)]\Pi_0(d\alpha_W\mid g)}
        {\int \exp[\lambda a_P(\alpha_W)]\Pi_0(d\alpha_W\mid g)} .
\]
Let \(\mathbb E_\lambda\) denote expectation under \(\Pi_\lambda(\cdot\mid g)\).  Boundedness of $a_P$ justifies differentiating numerator and denominator.  The quotient rule gives
\[
        \partial_\lambda T_\lambda(f)
        =\mathbb E_\lambda[f(\alpha_W)a_P(\alpha_W)]-T_\lambda(f)T_\lambda(a_P)
        =\operatorname{Cov}_\lambda(f(\alpha_W),a_P(\alpha_W)).
\]
Applying the same identity to $fa_P$, to $f$, and to $a_P$ gives
\[
\begin{aligned}
        \partial_\lambda^2T_\lambda(f)
        &=\operatorname{Cov}_\lambda(f(\alpha_W)a_P(\alpha_W),a_P(\alpha_W))
          -T_\lambda(a_P)\operatorname{Cov}_\lambda(f(\alpha_W),a_P(\alpha_W)) \\
        &\quad -T_\lambda(f)\operatorname{Var}_\lambda(a_P(\alpha_W)) \\
        &=\mathbb E_\lambda\left([f(\alpha_W)-T_\lambda(f)][a_P(\alpha_W)-T_\lambda(a_P)]^2\right).
\end{aligned}
\]
The last equality is the centered-square identity \(\operatorname{Cov}_\lambda(fa_P,a_P)-T_\lambda(a_P)\operatorname{Cov}_\lambda(f,a_P)-T_\lambda(f)\operatorname{Var}_\lambda(a_P)=\mathbb E_\lambda[(f-T_\lambda f)(a_P-T_\lambda a_P)^2]\). Thus \(|\partial_\lambda^2T_\lambda(f)|\le8\|f\|_\infty A_0^2\), uniformly over \(\lambda\in[0,\bar\lambda]\), $g$, and $e$.  Applying this bound coordinate by coordinate to $f(\alpha_W)=\alpha_W$ gives
\(t_\lambda(g)=t_0(g)+\lambda\dot t_0(g)+r_\lambda(g)\), with \(\sup_{e\in\mathcal E}\sup_g\|r_\lambda(g)\|\le C_t\lambda^2\),
where \(\dot t_0(g)=\partial_\lambda t_\lambda(g)|_{\lambda=0}\) and $C_t<\infty$ depends only on the objects listed in the theorem.

To compute \(\dot t_0(g)\), use \(t_\lambda(g)=T_\lambda(\alpha_W)\), so \(\dot t_0(g)=\operatorname{Cov}_{\Pi_0(\cdot\mid g)}(\alpha_W,a_P(\alpha_W))\).  Put \(h=\alpha_W-\alpha_W^*\), \(d_e=\alpha_W^*-\bar\alpha_W\), \(r_e(g)=t_0(g)-\alpha_W^*\), and \(u=h-r_e(g)=\alpha_W-t_0(g)\).  Since \(a_P(\alpha_W^*+h)=-(d_e+h)'P(d_e+h)/2\), the covariance-relevant part is \(-d_e'Ph-h'Ph/2\).  Under \(\Pi_0(\cdot\mid g)\), \(\operatorname{Cov}(h,d_e'Ph)=\mathbb E[uu']Pd_e=\Omega_e(g)Pd_e\). Also, since \(h'Ph=(r_e(g)+u)'P(r_e(g)+u)\), \(\operatorname{Cov}(h,h'Ph)=\mathbb E[u(2r_e(g)'Pu+u'Pu)]=2\Omega_e(g)Pr_e(g)+\tau_{e,P}(g)\). Therefore
\(\dot t_0(g)=-\Omega_e(g)P[d_e+r_e(g)]-\tau_{e,P}(g)/2\).

Substituting the Taylor expansion into the loss, and writing \(x_g=B r_e(g)\), \(y_g=B\dot t_0(g)\), and \(z_g=Br_\lambda(g)\), gives
\(\|x_g+\lambda y_g+z_g\|^2-\|x_g\|^2=2\lambda x_g'y_g+2x_g'z_g+\lambda^2\|y_g\|^2+2\lambda y_g'z_g+\|z_g\|^2\).
All factors in this identity are uniformly bounded: $\calW$ is compact, $t_0(g)$ lies in the convex hull of $\calW$, and every covariance and third central moment is taken under a probability measure supported on $\calW$.  Thus \(\|x_g\|\le C_x\), \(\|y_g\|\le C_y\), and \(\|z_g\|\le C_z\lambda^2\), uniformly over \(e\in\mathcal E\), $g$, and \(\lambda\in[0,\bar\lambda]\). Consequently \(|2x_g'z_g|\le2C_xC_z\lambda^2\), \(\lambda^2\|y_g\|^2\le C_y^2\lambda^2\), \(|2\lambda y_g'z_g|\le2C_yC_z\lambda^3\), and \(\|z_g\|^2\le C_z^2\lambda^4\). Hence all terms except \(2\lambda x_g'y_g\) are bounded in absolute value by \(C\lambda^2\).  Taking \(\mathbb E_m\) and substituting the expression for \(\dot t_0(g)\) gives
\[
\begin{aligned}
        M_{\lambda,B}(e)-M_{0,B}(e)
        &=-2\lambda \mathbb{E}_m\left[(B r_e(g))'B[\Omega_e(g)P r_e(g)+\tau_{e,P}(g)/2]\right] \\
        &\quad -2\lambda [P(\alpha_W^*-\bar\alpha_W)]'\mathbb{E}_m[\Omega_e(g)B'B r_e(g)]+R_e(\lambda) \\
        &=-2\lambda\left[A_e(P)+[P(\alpha_W^*-\bar\alpha_W)]'L_e(P)\right]+R_e(\lambda),
\end{aligned}
\]
where \(|R_e(\lambda)|\le C\lambda^2\) uniformly over \(e\in\mathcal E\) and \(\lambda\in[0,\bar\lambda]\).  This proves the expansion, and it remains to prove the local-improvement implication.  Let \(D_{\calW}=\sup_{u,v\in\calW}\|u-v\|\).  Since $t_0(g)$ is a quasi-posterior average over $\calW$, \(\|r_e(g)\|\le D_{\calW}\).  For any unit vector $v$, \(v'\Omega_e(g)v=\operatorname{Var}_{\Pi_0}(v'\alpha_W\mid g)\le D_{\calW}^2\), and therefore \(\|\Omega_e(g)\|_{\mathrm{op}}\le D_{\calW}^2\).  Hence
\(\|L_e(P)\|\le\mathbb{E}_m\|\Omega_e(g)B'B r_e(g)\|\le D_{\calW}^3\|B\|_{\mathrm{op}}^2\le K\).
If \(A_e(P)\ge a\) and \(\|P(\alpha_W^*-\bar\alpha_W)\|\le a/(2K)\), then Cauchy--Schwarz gives
\(A_e(P)+[P(\alpha_W^*-\bar\alpha_W)]'L_e(P)\ge a-\|P(\alpha_W^*-\bar\alpha_W)\|\|L_e(P)\|\ge a/2\).
The expansion gives \(M_{\lambda,B}(e)-M_{0,B}(e)\le-a\lambda+C\lambda^2\) uniformly on the restricted class.  With \(\lambda^*=\min[\bar\lambda,a/(2\max(C,1))]\), the right side is at most \(-a\lambda/2\) for every \(\lambda\in(0,\lambda^*]\).  Taking the supremum over \(e\in\mathcal C\) proves the stated uniform strict-improvement conclusion.
\end{proof}

\medskip

The expansion becomes a local uniform dominance result over restricted classes of weak-GMM limit experiments.  Let $D_{\calW}=\sup_{u,v\in\calW}\|u-v\|$ and $K=\max(1,D_{\calW}^3\|B\|_{\mathrm{op}}^2)$.  If $a>0$ and a class of weak-GMM limit experiments satisfies \(A_e(P)\ge a\) and \(\|P(\alpha_W^*-\bar\alpha_W)\|\le a/(2K)\), then the proof gives, uniformly on that class, \(M_{\lambda,B}(e)-M_{0,B}(e)\le -a\lambda+C\lambda^2\) for \(0\le\lambda\le\bar\lambda\).  Thus, with $\lambda^*=\min[\bar\lambda,a/(2\max(C,1))]$, the pooled rule has strictly smaller pointwise asymptotic MSE for every $\lambda\in(0,\lambda^*]$ throughout the class.  The first restriction, $A_e(P)\ge a$, is a first-order gain margin.  If the true parameter satisfies $P(\alpha_W^*-\bar\alpha_W)=0$, the expansion reduces to \(M_{\lambda,B}(e)-M_{0,B}(e)=-2\lambda A_e(P)+O(\lambda^2)\).  Thus $A_e(P)>0$ is the primitive local condition that the infinitesimal shrinkage direction induced by $P$ is MSE-reducing at $\lambda=0$.  The uniform margin is needed only because one $\lambda^*$ is chosen for a class of experiments.

The second restriction, $\|P(\alpha_W^*-\bar\alpha_W)\|\le a/(2K)$, is the bias restriction.  It measures misspecification only in the directions penalized by the hierarchy.  When $P$ is singular, movement along the manifold $\calM_{P,\bar\alpha_W}$ is unrestricted: for a common-mean hierarchy, the common level may be far from zero, provided the cross-group deviations are small.  More generally, if $P\ne0$ and $\rho_P$ is the smallest positive eigenvalue of $P$, then \(\operatorname{dist}(\alpha_W,\bar\alpha_W+\mathcal N(P))\le\rho_P^{-1}\|P(\alpha_W-\bar\alpha_W)\|\), where \(\mathcal N(P)=\{x:Px=0\}\).  The constant $K$ is only a worst-case compactness bound on $\|L_e(P)\|$.  In an application or a smaller theoretical subclass, it may be replaced by any sharper bound $K_{\mathcal C}\ge\sup_{e\in\mathcal C}\|L_e(P)\|$, which correspondingly enlarges the allowed neighborhood of the pooling manifold.

The term $A_e(P)$ has two components: the part involving $\Omega_e(g)P r_e(g)$ is the shrinkage effect of moving the quasi-posterior mean in penalized directions.  The vector $\tau_{e,P}(g)$ is a quasi-posterior skewness correction generated by nonquadratic or asymmetric weak-GMM objectives.  If the unpooled quasi-posterior is centrally symmetric around $t_0(g)$, then $\tau_{e,P}(g)=0$.  In that symmetric case the gain margin is easiest to read.

A checkable sufficient condition for the gain margin is available when the hierarchy penalizes a projection and the estimand is the penalized component itself. Work under the conditions of Theorem~\ref{thm:uniform-local} and suppose $P=P'=P^2$ and $B=P$. For a class $\mathcal E$ of weak-GMM limit experiments, the pointwise covariance and skewness conditions are, for each $e\in\mathcal E$ and $P_m$-almost every $g$, that \(P\Omega_e(g)P\succeq v_e(g)P\) and \(|(P r_e(g))'P\tau_{e,P}(g)|\le2(1-\rho)v_e(g)\|P r_e(g)\|^2\), where $v_e(g)\ge0$ and $\rho\in(0,1]$. The average gain margin is \(\inf_{e\in\mathcal E}\rho\mathbb E_m[v_e(g)\|P(t_0(g)-\alpha_W^*)\|^2]\ge a>0\). The covariance lower bound requires the unpooled quasi-posterior to retain uncertainty in the directions the hierarchy will shrink. The skewness condition allows nonquadratic and asymmetric weak-GMM objectives, provided quasi-posterior skewness does not offset the full variance-reduction gain. Exact quasi-posterior symmetry in the penalized directions is the special case $\rho=1$ with $(P r_e)'P\tau_{e,P}=0$.

\begin{corollary}\label{cor:projection-gain}
Work in the setting of Theorem~\ref{thm:uniform-local}, with the same class \(\mathcal E\) and the same definitions of \(A_e(P)\) and \(K\).  Suppose \(P=P'=P^2\), \(B=P\), and, for every \(e\in\mathcal E\) and \(P_m\)-almost every \(g\), \(P\Omega_e(g)P\succeq v_e(g)P\) and \(|(P r_e(g))'P\tau_{e,P}(g)|\le2(1-\rho)v_e(g)\|P r_e(g)\|^2\), where \(v_e(g)\ge0\) and \(\rho\in(0,1]\).  If \(\inf_{e\in\mathcal E}\rho\mathbb E_m[v_e(g)\|P(t_0(g)-\alpha_W^*)\|^2]\ge a>0\), then \(A_e(P)\ge a\) for all \(e\in\mathcal E\).  Hence the conclusion of Theorem~\ref{thm:uniform-local} holds uniformly on the subclass of \(\mathcal E\) satisfying \(\|P(\alpha_W^*-\bar\alpha_W)\|\le a/(2K)\).
\end{corollary}

\begin{proof}[Proof of Corollary \ref{cor:projection-gain}]
Under $P=P'=P^2$ and $B=P$, one has $B'B=P$.  Fix $e$ and $g$ and put $r=r_e(g)$ and $x=Pr$.  The pointwise integrand in $A_e(P)$ is
\((Pr)'[\Omega_e(g)Pr+\tau_{e,P}(g)/2]=r'P\Omega_e(g)Pr+r'P\tau_{e,P}(g)/2\).
The covariance lower bound gives \(r'P\Omega_e(g)Pr\ge v_e(g)r'Pr=v_e(g)\|Pr\|^2\), because $P$ is an orthogonal projection.  The skewness bound gives \(r'P\tau_{e,P}(g)/2\ge-(1-\rho)v_e(g)\|Pr\|^2\).  Therefore the integrand is bounded below by \(\rho v_e(g)\|P r_e(g)\|^2\).  Taking expectations and then the infimum over $e\in\mathcal E$ yields
\[
        \inf_{e\in\mathcal E}A_e(P)
        \ge \inf_{e\in\mathcal E}\rho\mathbb{E}_m[v_e(g)\|P(t_0(g)-\alpha_W^*)\|^2]
        \ge a .
\]
The hypotheses of Theorem \ref{thm:uniform-local} are therefore satisfied on the subclass of $\mathcal E$ obeying $\|P(\alpha_W^*-\bar\alpha_W)\|\le a/(2K)$, and the stated uniform improvement follows from that theorem.
\end{proof}

\medskip

\paragraph{Scalar fixed-center example.}

Let $d_W=1$, $B=1$, $P=1$, and $\bar\alpha_W$ be the hierarchy center.  Write $\Omega_e(g)$ for the unpooled quasi-posterior variance and $\tau_e(g)=\mathbb{E}_{\Pi_0}[(\alpha_W-t_0(g))^3\mid g]$ for the quasi-posterior third central moment.  Then \(A_e(1)=\mathbb{E}_m[\Omega_e(g)r_e(g)^2+r_e(g)\tau_e(g)/2]\), and \(L_e(1)=\mathbb{E}_m[\Omega_e(g)r_e(g)]\).  If the unpooled quasi-posterior is symmetric, $\tau_e(g)=0$ and $A_e(1)=\mathbb{E}_m[\Omega_e(g)r_e(g)^2]\ge0$; more generally, the small-skewness bound in Corollary \ref{cor:projection-gain} permits a nonzero third central moment.  A uniform margin holds whenever the unpooled quasi-posterior mean has nonnegligible weak-limit error in states where quasi-posterior variance remains nonnegligible.  The theorem implies that all weak-GMM limit experiments with $|\alpha_W^*-\bar\alpha_W|$ small enough have lower asymptotic MSE under a common small positive amount of pooling toward $\bar\alpha_W$.

\paragraph{Coordinatewise ridge example.}

When \(P=\operatorname{diag}(p_1,\ldots,p_{d_W})\), \(B'B=\operatorname{diag}(b_1^2,\ldots,b_{d_W}^2)\), and \(\Omega_e(g)=\operatorname{diag}(\omega_1(g),\ldots,\omega_{d_W}(g))\), conditional symmetry reduces the gain term to \(A_e(P)=\mathbb{E}_m\sum_{j=1}^{d_W} b_j^2\omega_j(g)p_jr_{e,j}(g)^2\).  Only components with $p_j>0$ are pooled, so the gain condition requires weak-limit variance and unpooled error in those components, while the neighborhood condition requires $p_j((\alpha_W^*)_j-\bar\alpha_{W,j})$ to be jointly small.  Components with $p_j=0$ are not restricted by the pooling-manifold condition.

\paragraph{Common-mean pooling example.}

Let $M_J=I_J-J^{-1}\one\one'$ and take $P=M_J$, the projection onto deviations from the group average.  Taking $\bar\alpha_W=0$ gives the pooling manifold $M_J\alpha_W=0$, i.e. $(\alpha_W)_1=\cdots=(\alpha_W)_J$.  If the estimand is the vector of deviations, $B=M_J$, and the unpooled quasi-posterior is conditionally symmetric with covariance $\Omega_e(g)=\omega(g)I_J$, then \(A_e(M_J)=\mathbb{E}_m[\omega(g)\|M_J r_e(g)\|^2]\) and \(P(\alpha_W^*-\bar\alpha_W)=M_J\alpha_W^*\).  Therefore a common-mean hierarchy uniformly improves the asymptotic MSE of the deviation estimand for all experiments whose true cross-group deviations $\|M_J\alpha_W^*\|$ are small enough, provided the unpooled weak quasi-posterior mean has enough variation in those same deviation directions.  No condition is imposed on the common level $J^{-1}\one'\alpha_W^*$, because that direction lies on the pooling manifold and is not shrunk.  For a contrast estimand $B=b'$ with $b'\one=0$, the same positive-gain conclusion holds under the same covariance symmetry, because $A_e(M_J)=\mathbb{E}_m[\omega(g)(b'r_e(g))^2]$.

The scalar weak-IV calculation in the main text gives the finite-\(\lambda\) analogue for the IV ratio under weak first-stage variation: the unpooled rule divides by the weak first stage, while the pooled rule replaces the unstable inverse with a ridge denominator.  Proposition~\ref{prop:iv} gives the exact conditional comparison, and Theorem~\ref{thm:uniform-local} gives the corresponding small-pooling local logic for the quasi-posterior mean.

\subsection{Weak-IV and testing details}

The scalar weak-IV proof is a conditional second-moment calculation.  The notation \(G_x\) is used for the weak first-stage limit and should not be confused with the weak-GMM process $g$.

\begin{proof}[Proof of Proposition \ref{prop:iv}]
Condition on a value of $G_x\ne0$ at which the displayed conditional moments exist, and fix $\lambda>0$.  The local truth is $h=0$, so conditional MSE is the conditional second moment.  The unpooled rule is $H=G_u/G_x$, and therefore
\(\mathbb{E}_m(H^2\mid G_x)=G_x^{-2}\mathbb{E}_m(G_u^2\mid G_x)=\frac{b(G_x)^2+\sigma^2(G_x)}{G_x^2}\).
For the pooled rule,
\[
        \mathbb{E}_m(\Delta_\lambda\mid G_x)
        =\frac{G_xb(G_x)+\lambda h_c}{G_x^2+\lambda},
        \qquad
        \operatorname{Var}_m(\Delta_\lambda\mid G_x)
        =\frac{G_x^2\sigma^2(G_x)}{(G_x^2+\lambda)^2}.
\]
Thus
\(\mathbb{E}_m(\Delta_\lambda^2\mid G_x)=\frac{G_x^2\sigma^2(G_x)+[G_xb(G_x)+\lambda h_c]^2}{(G_x^2+\lambda)^2}\).
The conditional unpooled MSE minus the conditional pooled MSE has denominator \(G_x^2(G_x^2+\lambda)^2>0\).  Its numerator is
\[
\begin{aligned}
&[b(G_x)^2+\sigma^2(G_x)](G_x^2+\lambda)^2
-G_x^2\left[G_x^2\sigma^2(G_x)+[G_xb(G_x)+\lambda h_c]^2\right] \\
&\quad=\lambda\left[(2G_x^2+\lambda)(b(G_x)^2+\sigma^2(G_x))
        -2G_x^3b(G_x)h_c-\lambda G_x^2h_c^2\right]
=\lambda D_\lambda(G_x).
\end{aligned}
\]
Because \(\lambda>0\), pooling lowers conditional MSE if and only if \(D_\lambda(G_x)>0\).  Uniform conditional improvement over a set \(\Lambda_1\) with positive pooling precisions is the same statement with \(\inf_{\lambda\in\Lambda_1}D_\lambda(G_x)>0\).  If $b(G_x)=0$, the condition becomes \((2G_x^2+\lambda)\sigma^2(G_x)>\lambda G_x^2h_c^2\), equivalently \(h_c^2<\sigma^2(G_x)(G_x^{-2}+2\lambda^{-1})\).

For the infinite-MSE claim, use \(P_m(G_x\ne0)=1\).  The law of iterated expectations gives
\(\mathbb{E}_mH^2=\mathbb{E}_m\left[\frac{b(G_x)^2+\sigma^2(G_x)}{G_x^2}\right]=\infty\).
For the pooled rule,
\[
        \Delta_\lambda^2
        \le \frac{2G_x^2G_u^2}{(G_x^2+\lambda)^2}
             +\frac{2\lambda^2h_c^2}{(G_x^2+\lambda)^2}
        \le \frac{2G_x^2G_u^2}{\lambda^2}+2h_c^2 .
\]
With \(\underline\lambda=\inf\Lambda_{IV}>0\),
\(\sup_{\lambda\in\Lambda_{IV}}\Delta_\lambda^2\le 2\underline\lambda^{-2}G_x^2G_u^2+2h_c^2\).
The right side has finite expectation when \(\mathbb{E}_m(G_x^2G_u^2)<\infty\) and $h_c$ is finite.  Thus the pooled asymptotic MSE is finite uniformly over \(\lambda\in\Lambda_{IV}\), while the unpooled weak-IV risk is infinite.
\end{proof}

\medskip

The proof of the testing propositions has two parts.  The first is the conditioning argument in \citet{AM2016conditional}: after conditioning on the residual process, the null law of the remaining Gaussian coordinate is pivotal.  The second is the conditional Neyman--Pearson argument for the mixture likelihood ratio weighted by the hierarchy.

\begin{proof}[Proof of Proposition \ref{prop:testing}]
Work with a separable version of the Gaussian process, so regular conditional laws given the residual path exist.  Under \(H_{0,a_0}\), the nuisance mean satisfies \(m(a_0)=0\).  The residual definition gives the decomposition \(g(\alpha_W)=h_{a_0}(\alpha_W)+\Sigma(\alpha_W,a_0)\Sigma_{00}^{-1}g(a_0)\).  For every \(\alpha_W\), \(\operatorname{Cov}(h_{a_0}(\alpha_W),g(a_0))=\Sigma(\alpha_W,a_0)-\Sigma(\alpha_W,a_0)\Sigma_{00}^{-1}\Sigma_{00}=0\).  Every finite-dimensional projection of the centered Gaussian residual process is therefore uncorrelated with \(g(a_0)\), and joint Gaussianity implies independence.  The deterministic nuisance mean is included in the observed residual path; after conditioning on that path, it does not alter the distribution of the remaining Gaussian coordinate.  Also \(g(a_0)\sim N(0,\Sigma_{00})\) under the null.  Therefore, for any realized residual path $h$, the conditional null law of the full process is the law of \(h(\alpha_W)+V_{a_0}(\alpha_W)\xi^*\), where \(\xi^*\sim N(0,\Sigma_{00})\).  This conditional law depends on the nuisance mean only through the realized value of $h$.

By the assumed construction of the randomized critical rule, for every \(h\in\calH_{0,a_0}\) and every \(\lambda\in\Lambda\),
\[
        P_0^h[T_\lambda(g)>c_{\zeta,\lambda}(h)]
        +\rho_{\zeta,\lambda}(h)P_0^h[T_\lambda(g)=c_{\zeta,\lambda}(h)]=\zeta .
\]
Hence \(\mathbb{E}_{a_0,m}[\varphi_{\lambda,a_0}(g)\mid h_{a_0}=h]=\zeta\) simultaneously for all \(\lambda\in\Lambda\) on a residual-path set that has probability one under every null nuisance mean.  Iterated expectations give \(\mathbb{E}_{a_0,m}\varphi_{\lambda,a_0}(g)=\zeta\) for every \(\lambda\).  This proves conditional size and unconditional size under every null nuisance mean.
\end{proof}

\medskip

\begin{proposition}\label{prop:testing-wap}
For fixed $h$ and $\lambda$, suppose $W_{\lambda,a_0}$ is a nonzero finite measure on the alternative index set and $P_\eta^h\ll P_0^h$ for $W_{\lambda,a_0}$-almost every $\eta$.  Let \(M_{\lambda,a_0}^h(A)=\int P_\eta^h(A)W_{\lambda,a_0}(d\eta)\), and let \(\mathfrak L_{\lambda,a_0}^h=dM_{\lambda,a_0}^h/dP_0^h\), assuming this density is finite $P_0^h$-almost surely.  If \(\varphi^*(g)=\one[\mathfrak L_{\lambda,a_0}^h(g)>c]+\rho\one[\mathfrak L_{\lambda,a_0}^h(g)=c]\) has exact conditional size $\zeta$, then \(\varphi^*\) maximizes \(\int \mathbb{E}_\eta^h\varphi(g)W_{\lambda,a_0}(d\eta)\) among all tests with conditional size at most $\zeta$.
\end{proposition}

\begin{proof}[Proof of Proposition \ref{prop:testing-wap}]
Fix $h$ and $\lambda$, and suppress $(\lambda,a_0,h)$ from the notation.  Let $P_0$ be the conditional null law, let $W$ be the finite alternative weight, and define the finite mixture measure \(M(A)=\int P_\eta(A)W(d\eta)\).  Since \(P_\eta\ll P_0\) for $W$-almost every $\eta$, $M\ll P_0$.  Let \(\mathfrak L=dM/dP_0\), which is nonnegative $P_0$-almost surely.  For every measurable test $0\le\varphi\le1$, Tonelli's theorem gives
\(\int \mathbb{E}_\eta\varphi(g)W(d\eta)=\int \varphi(g)M(dg)=\mathbb{E}_0[\varphi(g)\mathfrak L(g)]\).
Multiplying $W$, and hence $\mathfrak L$, by a positive constant does not change a likelihood-ratio critical region, so the argument covers both finite and probability weights.

Let \(\varphi^*(g)=\one[\mathfrak L(g)>c]+\rho\one[\mathfrak L(g)=c]\) be an exact conditional-size likelihood-ratio test, so \(\mathbb{E}_0\varphi^*=\zeta\).  For any other test $\varphi$ with \(\mathbb{E}_0\varphi\le\zeta\), the difference in weighted conditional power is
\[
        \mathbb{E}_0[(\varphi^* -\varphi)\mathfrak L]
        =\mathbb{E}_0[(\varphi^* -\varphi)(\mathfrak L-c)]
          +c\mathbb{E}_0(\varphi^* -\varphi).
\]
The first term is nonnegative pointwise: \((\varphi^*-\varphi)(\mathfrak L-c)=(1-\varphi)(\mathfrak L-c)\ge0\) on \([\mathfrak L>c]\), it equals \(\varphi(c-\mathfrak L)\ge0\) on \([\mathfrak L<c]\), and it is zero on \([\mathfrak L=c]\).  The second term is nonnegative because $c\ge0$ and \(\mathbb{E}_0\varphi^*=\zeta\ge\mathbb{E}_0\varphi\).  Therefore \(\varphi^*\) maximizes weighted conditional power among all tests with conditional size at most $\zeta$, conditionally on $h$ and for each $\lambda$.
\end{proof}

\medskip

For feasible conditional testing, apply the weak moment-process conditions under \(H_{0,a_0}\), assume nonsingularity of \(\Sigma(a_0,a_0)\), and suppose the covariance-kernel estimator satisfies \((g_n,\widehat\Sigma_n)\dto(g,\Sigma)\) in the product sup-norm space under the null. For \(\lambda\in\Lambda\), a residual path $h$, and a covariance function $K$, define \(K_{00}=K(a_0,a_0)\), \(V_K(\alpha_W)=K(\alpha_W,a_0)K_{00}^{-1}\), let $F_{\lambda,h,K}$ be the distribution of \(T_\lambda(h+V_K\xi,K)\) for \(\xi\sim N(0,K_{00})\), and let \(c_{\zeta,\lambda}(h,K)=\inf\{t:F_{\lambda,h,K}(t)\ge1-\zeta\}\). Define \(h_{n,a_0}(\alpha_W)=g_n(\alpha_W)-\widehat\Sigma_n(\alpha_W,a_0)\widehat\Sigma_n(a_0,a_0)^{-1}g_n(a_0)\), with \(c_{n,\zeta,\lambda}=c_{\zeta,\lambda}(h_{n,a_0},\widehat\Sigma_n)\). The statistic-continuity requirement is that, whenever $x_r\to x$ and $K_r\to K$ uniformly at the relevant null and simulated limit paths, \(\sup_{\lambda\in\Lambda}|T_\lambda(x_r,K_r)-T_\lambda(x,K)|\to0\). The limit-quantile isolation requirement is that, for every $\varepsilon>0$,
\[
\begin{aligned}
        &\inf_{\lambda\in\Lambda}
        \left[F_{\lambda,h_{a_0},\Sigma}
        (c_{\zeta,\lambda}(h_{a_0},\Sigma)+\varepsilon)-(1-\zeta)\right]>0,\\
        &\inf_{\lambda\in\Lambda}
        \left[(1-\zeta)-F_{\lambda,h_{a_0},\Sigma}
        (c_{\zeta,\lambda}(h_{a_0},\Sigma)-\varepsilon)\right]>0
\end{aligned}
\]
almost surely. Together with the joint convergence of \((g_n,\widehat\Sigma_n)\), these requirements are the feasible conditional-testing conditions.

\begin{lemma}\label{lem:conditional-quantile-continuity}
Let $\Lambda$ be compact.  For each $r$, let $Y_{r,\lambda}$ and $Y_\lambda$, $\lambda\in\Lambda$, be real random variables on a common probability space such that \(\sup_{\lambda\in\Lambda}|Y_{r,\lambda}-Y_\lambda|\to0\) in probability.  Let \(c_{r,\lambda}=\inf\{t:F_{r,\lambda}(t)\ge1-\zeta\}\) and \(c_\lambda=\inf\{t:F_\lambda(t)\ge1-\zeta\}\), where $F_{r,\lambda}$ and $F_\lambda$ are the corresponding distribution functions.  If, for every $\varepsilon>0$,
\[
        \inf_{\lambda\in\Lambda}\left[F_\lambda(c_\lambda+\varepsilon)-(1-\zeta)\right]>0,
        \qquad
        \inf_{\lambda\in\Lambda}\left[(1-\zeta)-F_\lambda(c_\lambda-\varepsilon)\right]>0,
\]
then \(\sup_{\lambda\in\Lambda}|c_{r,\lambda}-c_\lambda|\to0\).
\end{lemma}

\begin{proof}[Proof of Lemma \ref{lem:conditional-quantile-continuity}]
Fix $\varepsilon>0$; by the isolation condition applied with $\varepsilon/2$, there is $\delta>0$ such that
\(F_\lambda(c_\lambda+\varepsilon/2)>1-\zeta+2\delta\) and \(F_\lambda(c_\lambda-\varepsilon/2)<1-\zeta-2\delta\)
for every $\lambda\in\Lambda$.  Let \(a_r=P[\sup_{\lambda\in\Lambda}|Y_{r,\lambda}-Y_\lambda|>\varepsilon/2]\), so $a_r\to0$.  For every $\lambda$,
\[
        F_{r,\lambda}(c_\lambda+\varepsilon)
        \ge F_\lambda(c_\lambda+\varepsilon/2)-a_r,
        \qquad
        F_{r,\lambda}(c_\lambda-\varepsilon)
        \le F_\lambda(c_\lambda-\varepsilon/2)+a_r .
\]
For all large $r$, both bounds are uniform in $\lambda$ and imply
\(F_{r,\lambda}(c_\lambda+\varepsilon)>1-\zeta+\delta\) and \(F_{r,\lambda}(c_\lambda-\varepsilon)<1-\zeta-\delta\).
By the definition of the left quantile, the first inequality gives \(c_{r,\lambda}\le c_\lambda+\varepsilon\), and the second gives \(c_{r,\lambda}\ge c_\lambda-\varepsilon\), uniformly in $\lambda$.  Since $\varepsilon$ is arbitrary, \(\sup_{\lambda\in\Lambda}|c_{r,\lambda}-c_\lambda|\to0\) as deterministic numbers.
\end{proof}

\medskip

\begin{proposition}\label{prop:feasible-testing}
Suppose \((g_n,\widehat\Sigma_n)\dto(g,\Sigma)\) under \(H_{0,a_0}\), \(\Sigma(a_0,a_0)\) is nonsingular, the statistic-continuity requirement holds uniformly over \(\lambda\in\Lambda\), the displayed limit-quantile isolation inequalities hold almost surely for every \(\varepsilon>0\), and the feasible critical values are computed from \(h_{n,a_0}\) and \(\widehat\Sigma_n\) on the event where \(\widehat\Sigma_n(a_0,a_0)\) is nonsingular, with arbitrary values assigned off that event. Then the statistic and critical-value paths converge weakly in \(\ell^\infty(\Lambda,\mathbb R^2)\) to their conditional Gaussian limits. With \(W_{n,\lambda}=(T_\lambda(g_n,\widehat\Sigma_n),c_{n,\zeta,\lambda})\) and \(W_\lambda=(T_\lambda(g,\Sigma),c_{\zeta,\lambda}(h_{a_0},\Sigma))\), this is \((W_{n,\lambda})_{\lambda\in\Lambda}\dto(W_\lambda)_{\lambda\in\Lambda}\). If
\[
        P_0\left[\inf_{\lambda\in\Lambda}|T_\lambda(g,\Sigma)-c_{\zeta,\lambda}(h_{a_0},\Sigma)|=0\right]=0,
\]
then the feasible nonrandomized tests satisfy
\[
        \sup_{\lambda\in\Lambda}
        \left|P_0[T_\lambda(g_n,\widehat\Sigma_n)>c_{n,\zeta,\lambda}]-\zeta\right|
        \to0.
\]
For a finite grid \(\Lambda_G\subset\Lambda\), the same conclusion holds on \(\Lambda_G\) without the uniform no-boundary condition if each grid point satisfies the corresponding quantile-isolation inequalities and \(P_0[T_\lambda(g,\Sigma)=c_{\zeta,\lambda}(h_{a_0},\Sigma)]=0\).  A sufficient grid-point condition is that each distribution function is continuous at its critical value and strictly crosses \(1-\zeta\) there.  If the critical values on \(\Lambda_G\) are simulated with \(M_n\) independent conditional draws and \(M_n\to\infty\), the conclusion is unchanged.
\end{proposition}

\begin{proof}[Proof of Proposition \ref{prop:feasible-testing}]
Because \(\Sigma(a_0,a_0)\) is nonsingular and matrix inversion is continuous on the nonsingular matrices, convergence of \(\widehat\Sigma_n(a_0,a_0)\) implies that \(\widehat\Sigma_n(a_0,a_0)\) is nonsingular with probability tending to one.  Values assigned off this event do not affect weak limits.  On the nonsingular event, the residual-path map \((x,K)\mapsto x(\cdot)-K(\cdot,a_0)K(a_0,a_0)^{-1}x(a_0)\) is continuous at $(g,\Sigma)$ in the uniform norm. Indeed, if \((x_r,K_r)\to(x,K)\) uniformly and \(K_{r,00}=K_r(a_0,a_0)\), then
\[
\begin{aligned}
&\left\|x_r(\cdot)-K_r(\cdot,a_0)K_{r,00}^{-1}x_r(a_0)-x(\cdot)+K(\cdot,a_0)K_{00}^{-1}x(a_0)\right\|_\infty \\
&\quad\le \|x_r-x\|_\infty+\|K_r(\cdot,a_0)\|_\infty\|K_{r,00}^{-1}\|\|x_r(a_0)-x(a_0)\| \\
&\qquad +\|K_r(\cdot,a_0)-K(\cdot,a_0)\|_\infty\|K_{r,00}^{-1}\|\|x(a_0)\| \\
&\qquad +\|K(\cdot,a_0)\|_\infty\|K_{r,00}^{-1}-K_{00}^{-1}\|\|x(a_0)\|,
\end{aligned}
\]
which tends to zero because inversion is continuous at the nonsingular matrix \(K_{00}=K(a_0,a_0)\). The projection map \(K\mapsto K(\cdot,a_0)K(a_0,a_0)^{-1}\) is continuous by the same two matrix terms, with the \(x(a_0)\) factors omitted. Therefore \((h_{n,a_0},\widehat\Sigma_n)\dto(h_{a_0},\Sigma)\) jointly with \((g_n,\widehat\Sigma_n)\dto(g,\Sigma)\).

We next verify continuity of the simulated conditional law.  Fix a deterministic continuity point $(h,K)$ with \(K_{00}=K(a_0,a_0)\) nonsingular and take \((h_r,K_r)\to(h,K)\) uniformly.  The projection matrices satisfy
\[
\|V_{K_r}-V_K\|_\infty
\le
\|K_r(\cdot,a_0)-K(\cdot,a_0)\|_\infty\|K_{r,00}^{-1}\|
+\|K(\cdot,a_0)\|_\infty\|K_{r,00}^{-1}-K_{00}^{-1}\|,
\]
so \(V_{K_r}\to V_K\) uniformly.  On a common auxiliary probability space, write \(\xi_r=K_{r,00}^{1/2}Z\) and \(\xi=K_{00}^{1/2}Z\), where \(Z\sim N(0,I_{k_W})\).  Continuity of the symmetric matrix square-root map on positive definite matrices gives \(\xi_r\to\xi\) almost surely, and
\[
\|h_r+V_{K_r}\xi_r-h-V_K\xi\|_\infty
\le
\|h_r-h\|_\infty+\|V_{K_r}-V_K\|_\infty\|\xi_r\|+\|V_K\|_\infty\|\xi_r-\xi\|.
\]
Hence \(h_r+V_{K_r}\xi_r\to h+V_K\xi\) in \(\ell^\infty(\calW)\) almost surely.  The statistic-continuity requirement gives
\(\sup_{\lambda\in\Lambda}|T_\lambda(h_r+V_{K_r}\xi_r,K_r)-T_\lambda(h+V_K\xi,K)|\to0\)
in probability under the simulation draw.  Lemma \ref{lem:conditional-quantile-continuity} then gives \(\sup_{\lambda\in\Lambda}|c_{\zeta,\lambda}(h_r,K_r)-c_{\zeta,\lambda}(h,K)|\to0\).  Applying these deterministic continuity statements on the almost-sure support of \((h_{a_0},\Sigma)\), the continuous mapping theorem gives joint convergence in \(\ell^\infty(\Lambda,\mathbb R^2)\) of the statistic and critical-value paths to \(\lambda\mapsto(T_\lambda(g,\Sigma),c_{\zeta,\lambda}(h_{a_0},\Sigma))\).

It remains to prove uniform size, so define \(R_{n,\lambda}=T_\lambda(g_n,\widehat\Sigma_n)-c_{n,\zeta,\lambda}\) and \(R_\lambda=T_\lambda(g,\Sigma)-c_{\zeta,\lambda}(h_{a_0},\Sigma)\).  The preceding convergence gives \(R_n\dto R\) in \(\ell^\infty(\Lambda)\).  Because the limit law is tight and may be taken to have separable support, the Skorokhod--Dudley representation theorem applies along any subsequence.  On such a representation, with the same marginal laws, \(\|R_n-R\|_\infty\to0\) almost surely.  For any \(\eta>0\) and any \(\lambda\), the events \([R_{n,\lambda}>0]\) and \([R_\lambda>0]\) can differ only if \(\|R_n-R\|_\infty>\eta\) or \(|R_\lambda|\le\eta\).  Hence
\[
        \sup_{\lambda\in\Lambda}|P(R_{n,\lambda}>0)-P(R_\lambda>0)|
        \le P(\|R_n-R\|_\infty>\eta)
        +P\left(\inf_{\lambda\in\Lambda}|R_\lambda|\le\eta\right).
\]
The first probability tends to zero for every fixed \(\eta\).  The no-boundary condition implies \(P(\inf_{\lambda\in\Lambda}|R_\lambda|\le\eta)\to0\) as \(\eta\downarrow0\).  Since the subsequence was arbitrary, \(\sup_{\lambda\in\Lambda}|P(R_{n,\lambda}>0)-P(R_\lambda>0)|\to0\).

By Proposition \ref{prop:testing}, conditionally on the residual path and covariance, \(c_{\zeta,\lambda}(h_{a_0},\Sigma)\) is the left \((1-\zeta)\)-quantile of the conditional null law of \(T_\lambda(g,\Sigma)\).  Let \(F_\lambda(t)=P(T_\lambda(g,\Sigma)\le t\mid h_{a_0},\Sigma)\) and \(c_\lambda=c_{\zeta,\lambda}(h_{a_0},\Sigma)\).  The no-boundary condition implies \(P(R_\lambda=0)=0\) for each fixed \(\lambda\), hence \(P(R_\lambda=0\mid h_{a_0},\Sigma)=0\) almost surely. Thus the conditional distribution has no atom at \(c_\lambda\); since \(c_\lambda\) is the left \((1-\zeta)\)-quantile, \(F_\lambda(c_\lambda)=1-\zeta\). Hence \(P(R_\lambda>0\mid h_{a_0},\Sigma)=P(T_\lambda(g,\Sigma)>c_\lambda\mid h_{a_0},\Sigma)=1-F_\lambda(c_\lambda)=\zeta\) almost surely for every \(\lambda\).  Taking expectations gives \(P(R_\lambda>0)=\zeta\), and the preceding uniform bound proves the asserted uniform asymptotic size.

On a finite reported grid \(\Lambda_G\), the same argument is pointwise at each \(\lambda\in\Lambda_G\).  The grid-point no-boundary condition gives \(P(R_\lambda=0)=0\) for each reported \(\lambda\), and the grid-point quantile-isolation condition gives convergence of the corresponding critical value.  Since the grid is finite, taking the maximum over \(\Lambda_G\) preserves convergence, so these pointwise conditions are sufficient on the grid.

Finally, for simulated critical values on \(\Lambda_G\), let \(M_n\) be the number of conditional simulation draws, let \(F_{n,\lambda}\) be the conditional distribution function of \(T_\lambda(h_{n,a_0}+V_{\widehat\Sigma_n}\xi,\widehat\Sigma_n)\), and let \(\widehat F_{n,\lambda}\) be the empirical distribution function from the \(M_n\) draws.  Fix \(\varepsilon>0\); by the gridwise quantile-isolation condition and convergence of the feasible conditional laws, there are events \(A_n\) with \(P(A_n)\to1\) and a deterministic \(\delta_\varepsilon>0\) such that, on \(A_n\), for every \(\lambda\in\Lambda_G\),
\[
       F_{n,\lambda}(c_{n,\zeta,\lambda}+\varepsilon)>1-\zeta+2\delta_\varepsilon,
       \qquad
       F_{n,\lambda}(c_{n,\zeta,\lambda}-\varepsilon)<1-\zeta-2\delta_\varepsilon .
\]
Conditional on the data, the Dvoretzky--Kiefer--Wolfowitz inequality and a union bound give
\[
       P\!\left(\max_{\lambda\in\Lambda_G}\sup_t|\widehat F_{n,\lambda}(t)-F_{n,\lambda}(t)|>\delta_\varepsilon\mid X\right)
       \le 2|\Lambda_G|\exp(-2M_n\delta_\varepsilon^2)=o(1).
\]
On the intersection of these events, \(\widehat F_{n,\lambda}(c_{n,\zeta,\lambda}+\varepsilon)>1-\zeta+\delta_\varepsilon\) and \(\widehat F_{n,\lambda}(c_{n,\zeta,\lambda}-\varepsilon)<1-\zeta-\delta_\varepsilon\) for every \(\lambda\in\Lambda_G\). By the definition \(\widehat c_{n,\zeta,\lambda}=\inf\{t:\widehat F_{n,\lambda}(t)\ge1-\zeta\}\), the first inequality gives \(\widehat c_{n,\zeta,\lambda}\le c_{n,\zeta,\lambda}+\varepsilon\), while the second gives \(\widehat c_{n,\zeta,\lambda}\ge c_{n,\zeta,\lambda}-\varepsilon\). Therefore \(c_{n,\zeta,\lambda}-\varepsilon\le\widehat c_{n,\zeta,\lambda}\le c_{n,\zeta,\lambda}+\varepsilon\) for all \(\lambda\in\Lambda_G\).  Since \(\varepsilon\) is arbitrary, \(\max_{\lambda\in\Lambda_G}|\widehat c_{n,\zeta,\lambda}-c_{n,\zeta,\lambda}|=o_p(1)\).  Slutsky's theorem applied on the finite grid gives the same size conclusion for the simulated critical values.
\end{proof}

\medskip

\subsection{Auxiliary calculations used in the text}

\paragraph{Finite weak-parameter-set likelihood.} When the weak parameter space is finite, the calculation used in the main text follows from a direct normal integration.  Here \(\gamma^k\) is a candidate retained weak value from the finite weak parameter set. Fix a candidate \(\gamma^k\) and partition \(g=(g_k',g_{-k}')'\).  Under this candidate, the mean vector has the form \(m=(0,\mu_{-k}')'\).  Partition the covariance matrix into blocks \(\Sigma_{kk}\), \(\Sigma_{k,-k}\), \(\Sigma_{-k,k}\), and \(\Sigma_{-k,-k}\), and define \(\Omega_{-k\mid k}=\Sigma_{-k,-k}-\Sigma_{-k,k}\Sigma_{kk}^{-1}\Sigma_{k,-k}\).
Then the joint normal density factors as \(f(g\mid \gamma^k,\mu_{-k}) = f_k(g_k) f_{-k\mid k}(g_{-k}\mid g_k,\mu_{-k}),\)
where \(f_k(g_k)=(2\pi)^{-q_k/2}|\Sigma_{kk}|^{-1/2} \exp[-g_k'\Sigma_{kk}^{-1}g_k/2]\)
and \(f_{-k\mid k}(g_{-k}\mid g_k,\mu_{-k}) =\varphi_{\Omega_{-k\mid k}} (g_{-k}-\mu_{-k}-\Sigma_{-k,k}\Sigma_{kk}^{-1}g_k).\)
Here \(\varphi_\Omega\) denotes the centered normal density with covariance \(\Omega\).  Integrating the conditional density with respect to Lebesgue measure over \(\mu_{-k}\) gives \(\int \varphi_{\Omega_{-k\mid k}} (g_{-k}-\mu_{-k}-\Sigma_{-k,k}\Sigma_{kk}^{-1}g_k)d\mu_{-k}=1,\)
by the change of variables \(u=g_{-k}-\mu_{-k}-\Sigma_{-k,k}\Sigma_{kk}^{-1}g_k\).  Therefore, after dropping constants common across candidates of the same dimension, \(\ell_\infty(\gamma^k;g) \propto |\Sigma_{kk}|^{-1/2} \exp[-g_k'\Sigma_{kk}^{-1}g_k/2],\)
which is the proportionality used in the main text.

\paragraph{Normal prior.} The normal prior formula is obtained by completing the square.  Multiplying a \(N(\mu_\alpha,V_\alpha)\) prior density by \(\exp(-\lambda(\alpha-h_c)'P(\alpha-h_c)/2)\) gives the exponent \(-\frac12(\alpha-\mu_\alpha)'V_\alpha^{-1}(\alpha-\mu_\alpha)-\frac\lambda2(\alpha-h_c)'P(\alpha-h_c)=-\frac12\alpha'A_\lambda\alpha+b_\lambda'\alpha+C_\lambda\),
where \(A_\lambda=V_\alpha^{-1}+\lambda P\), \(b_\lambda=V_\alpha^{-1}\mu_\alpha+\lambda Ph_c\), and \(C_\lambda\) does not depend on \(\alpha\).  If \(A_\lambda\) is positive definite, then \(-\frac12\alpha'A_\lambda\alpha+b_\lambda'\alpha =-\frac12(\alpha-M_\lambda)'A_\lambda(\alpha-M_\lambda) +\frac12b_\lambda'A_\lambda^{-1}b_\lambda\), with \(M_\lambda=A_\lambda^{-1}b_\lambda\).
Thus the covariance and mean are \(V_\lambda=A_\lambda^{-1}\) and \(M_\lambda=V_\lambda(V_\alpha^{-1}\mu_\alpha+\lambda Ph_c)\).  With a flat prior density limit and positive definite \(P\), set \(V_\alpha^{-1}=0\); the same completion-of-squares calculation gives mean \(h_c\) and covariance \((\lambda P)^{-1}\).

\paragraph{Common-mean marginal prior.} The common-mean marginal prior follows from writing \(\alpha_W=\theta\one+\varepsilon\), where \(\varepsilon\mid\theta\sim N(0,\lambda^{-1}D_w^{-1})\) and is independent of \(\theta\sim N(\theta_0,\sigma_\theta^2)\).  Therefore \(\mathbb{E}\alpha_W=\theta_0\one\) and \(\text{Var}(\alpha_W)=\sigma_\theta^2\one\one'+\lambda^{-1}D_w^{-1}\).  In the equal-weight case \(D_w=I_J\), the Woodbury identity gives \([\sigma_\theta^2\one\one'+\lambda^{-1}I_J]^{-1} =\lambda I_J- \frac{\lambda^2\sigma_\theta^2}{1+\lambda\sigma_\theta^2J}\one\one'.\)
As \(\sigma_\theta^2\to\infty\), this precision converges to \(\lambda(I_J-J^{-1}\one\one')\).  Thus the finite part of the diffuse-common-mean precision penalizes only deviations from the span of \(\one\).

\paragraph{Distance to the pooling manifold.} The distance bound following Theorem~\ref{thm:uniform-local} is a spectral calculation: let \(P=P'\succeq0\), let \(\mathcal N(P)=\{x:Px=0\}\), and let \(\rho_P\) be the smallest positive eigenvalue of \(P\).  Decompose \(x=\alpha_W-\bar\alpha_W\) as \(x=x_0+x_1\), with \(x_0\in\mathcal N(P)\) and \(x_1\perp\mathcal N(P)\).  Then \(\operatorname{dist}(\alpha_W,\bar\alpha_W+\mathcal N(P))=\|x_1\|\), while \(\|Px\|=\|Px_1\|\ge\rho_P\|x_1\|\).  Hence \(\operatorname{dist}(\alpha_W,\bar\alpha_W+\mathcal N(P))\le\rho_P^{-1}\|P(\alpha_W-\bar\alpha_W)\|\).

\section{Additional moment details}\label{app:additional-moment-maps}

\paragraph{Minimum-distance and auxiliary-statistic maps.}
Let \(\hat h_{j,n_j}\) be an empirical auxiliary object and let \(h_j(\alpha_j)\) be its model counterpart, both taking values in a finite-dimensional Euclidean space or in an empirical discrepancy space with a specified inner product.  Let \(r_{j,n_j}\) be the rate at which the auxiliary object is localized, equal to \(\sqrt{n_j}\) in the regular finite-dimensional case.  With positive semidefinite weight \(W_{j,n_j}\), set
\[
G^{\mathrm{md}}_{j,n_j}(x_j,\alpha_j)
=
r_{j,n_j}(\hat h_{j,n_j}-h_j(\alpha_j)),
\qquad
q^{\mathrm{md}}_{2,j,n_j}(x_j,\alpha_j)
=
-\frac12\|G^{\mathrm{md}}_{j,n_j}(x_j,\alpha_j)\|_{W_{j,n_j}}^2 .
\]
This construction is the usual GMM construction when \(\hat h_{j,n_j}\) is a vector of sample moments, and it is minimum distance when \(\hat h_{j,n_j}\) is a fitted distribution, quantile function, impulse response, or other auxiliary estimate.  No likelihood or additional sampling model is introduced.

The verification follows from the scaled map: in the regular case \(r_{j,n_j}=\sqrt{n_j}\).  For strong identification, if \(G^{\mathrm{md}}_{j,n_j}(\alpha_{j0}+h/\sqrt{n_j}) =G^{\mathrm{md}}_{j,n_j}(\alpha_{j0})-D_jh+o_p(1)\)
uniformly for bounded \(h\), \(G^{\mathrm{md}}_{j,n_j}(\alpha_{j0})=O_p(1)\), \(W_{j,n_j}\to W_j\), and \(D_j^\top W_jD_j\succ0\), then \(q^{\mathrm{md}}_{2,j,n_j}(\alpha_{j0}+h/\sqrt{n_j})-q^{\mathrm{md}}_{2,j,n_j}(\alpha_{j0})=h'D_j^\top W_jG^{\mathrm{md}}_{j,n_j}(\alpha_{j0})-h'D_j^\top W_jD_jh/2+o_p(1)\), uniformly for bounded \(h\). Thus the local score is \(D_j^\top W_jG^{\mathrm{md}}_{j,n_j}(\alpha_{j0})\) and the curvature is \(D_j^\top W_jD_j\).  For weak identification, if along the weak parameterization \(G^{\mathrm{md}}_{j,n_j}[\mathcal T_j(\gamma_j,n_j^{-1/2}h)] =G^{0,\mathrm{md}}_{j,n_j}(\gamma_j)-D_j^\perp(\gamma_j)h+o_p(1)\), and \(G^{0,\mathrm{md}}_{j,n_j}\Rightarrow \mathbb G_j+\tau_j\) in \(\ell^\infty(\Gamma_j)\), then completing the same square in \(h\) gives the Schur-complement reduced objective function in Lemma~\ref{lem:weak-laplace-reduction}.  A minimum-distance problem therefore enters the QBHM theory only through the map, weight, and local expansion conditions already used for GMM.

\paragraph{Score and likelihood connections.}
Likelihood and quasi-likelihood enter through score moments or through the local quadratic experiment generated by the likelihood.  For observation log likelihood \(\ell_{j,n_j}(X_{ji},\alpha_j)\), define the score moment as \(G^{\mathrm{score}}_{j,n_j}=n_j^{-1/2}\sum_{i=1}^{n_j}\partial_{\alpha_j}\ell_{j,n_j}(X_{ji},\alpha_j)\), and use the score-based GMM objective function \(-\|G^{\mathrm{score}}_{j,n_j}\|_{W_{j,n_j}}^2/2\).  A global log-likelihood quasi-posterior is not the same object unless the researcher chooses the score-based GMM objective; the connection is local, so suppress the group index and write \(h=\sqrt n(\alpha-\alpha_0)\).  If the local asymptotic normality (LAN) expansion has quadratic part \(h^\top\Delta_n-h^\top I h/2\), uniformly for bounded \(h\), with \(I\succ0\) and \(\Delta_n=n^{-1/2}\sum_i\partial_\alpha\ell_n(X_i,\alpha_0)\Rightarrow N(0,I)\), then the score expansion is \(G^{\mathrm{score}}_n(\alpha_0+h/\sqrt n)=\Delta_n-Ih+o_p(1)\).  With efficient score-based GMM weight \(W=I^{-1}\), the algebra \(-\frac12(\Delta_n-Ih)'I^{-1}(\Delta_n-Ih)=h'\Delta_n-h'Ih/2-\Delta_n'I^{-1}\Delta_n/2\) shows that the corresponding GMM objective function equals the same local quadratic likelihood shape plus the additive constant \(-\Delta_n'I^{-1}\Delta_n/2\) and an \(o_p(1)\) remainder.
The final term is free of \(h\) and cancels from normalized local quasi-posteriors.  Hence the likelihood and efficient score-based GMM quasi-posteriors have the same local parameter law, local quasi-posterior mean, and local quadratic shape.  Under misspecification, the same calculation uses the quasi-likelihood score and curvature in place of \((\Delta_n,I)\).

\section{Monte Carlo simulation details}\label{app:mc}

The Monte Carlo designs in Section~\ref{sec:mc} use the same group-specific objective function for the unpooled estimator and for the lower level of the QBHM.  This appendix reports the simulation parameter values, hierarchical priors, the $\lambda$-grid results, and the nonconvex production-function simulation calibrated to the empirical application.

Table~\ref{tab:simple_design_parameters} reports details for the instrumental-variable simulations in Section~\ref{sec:mc}.  The model estimates the parameters of a Cobb--Douglas production function.  Three designs are run, each varying a different component of the setting.  The instrument-strength grid varies instrument strength, the \(\lambda\)-grid varies the pooling strength, and the sample-size grid varies group sample size.

Table~\ref{tab:lowtech_lambda_sweep_overall} reports results from the \(\lambda\)-grid, where the pooling strength varies in the strong-instrument, large-group-size setting.  As \(\lambda\) decreases, the RMSE of the QBHM converges to the performance of the unpooled estimator.

\begin{table}[htbp]
\centering
\footnotesize
\setlength{\tabcolsep}{6pt}
\renewcommand{\arraystretch}{1.14}
\caption{Simple simulation designs: common parameters, priors, and varied settings}
\label{tab:simple_design_parameters}
\begin{tabular}{@{}C{0.28\textwidth} C{0.64\textwidth}@{}}
\toprule
\textbf{Component} & \textbf{Specification} \\
\midrule
Model & Simple Cobb--Douglas production function, $Y_{ig}=a_{1g}K_{ig}^{\rho_g}\exp(\varepsilon_{ig})$ \\
Instrument vector & $Z_i=(1,\mathrm{treatment}_i)$ \\
Synthetic groups & 10 groups \\
Baseline resampling & Empirical resampling from the real full-model dataset with log-normal capital jitter, $\texttt{jitter\_sd}=0.02$ \\
Capital transition & $\log K^{post}_{ig}=\bar{\log K}^{pre}+\delta_T T_{ig}+\delta_P(\log K^{pre}_{ig}-\bar{\log K}^{pre})+\nu_{ig}$ \\
Capital persistence & $\delta_P=0.7$ \\
Capital shock & $\nu_{ig}\sim\mathcal{N}(0,\sigma_\nu^2)$ \\
Revenue shocks & Student-$t$ with 3 degrees of freedom, scaled to $\sigma_{\varepsilon,\mathrm{pre}}=\sigma_{\varepsilon,\mathrm{post}}=0.3$ \\
Data-generating process for $\log a_{1g}$ & $\log a_{1g}\sim\mathcal{N}(\log(0.8),0.3^2)$ \\
Data-generating process for $\rho_g$ & $\operatorname{logit}(\rho_g)\sim\mathcal{N}(\operatorname{logit}(0.3),0.4^2)$ \\
\midrule
\multicolumn{2}{@{}l}{\textbf{Varied designs}} \\
\midrule
Instrument-strength grid & $\lambda=1$, $n_g=1000$ for all groups; $(\delta_T,\sigma_\nu)\in\{(0.50,0.40),(0.06,1.10),(0.03,1.40)\}$ for very strong, weak, and very weak instruments \\
$\lambda$-grid & Very strong instruments fixed at $(\delta_T,\sigma_\nu)=(0.50,0.40)$ with $n_g=1000$; pooling strength $\lambda\in\{1,1/2,1/3\}$ \\
Sample-size grid & Very strong instruments fixed at $(\delta_T,\sigma_\nu)=(0.50,0.40)$ with $\lambda=1$; $n_g\in\{1000,50,20\}$ for all 10 groups \\
\midrule
\multicolumn{2}{@{}l}{\textbf{QBHM priors}} \\
\midrule
Prior for $\log a_{1,\mu}$ & $\log a_{1,\mu}\sim\mathcal{N}(\log(1.0),1.0^2)$ \\
Prior for $\sigma_{a_1}$ & $\sigma_{a_1}\sim\mathcal{N}^{+}(0,0.7^2)$ \\
Prior for $\mu_\rho$ & $\mu_\rho\sim\mathcal{N}(\operatorname{logit}(0.35),1.0^2)$ \\
Prior for $\sigma_\rho$ & $\sigma_\rho\sim\mathcal{N}^{+}(0,0.5^2)$ \\
Group-level raw effects & $\log a_{1,g}^{raw}\sim\mathcal{N}(0,1)$ and $\rho_g^{raw}\sim\mathcal{N}(0,1)$ \\
Stan sampling & Full runs use 2 chains, 1000 warmup iterations, and 1000 sampling iterations per chain \\
\bottomrule
\end{tabular}
\end{table}

\begin{table}[htbp]
\centering
\footnotesize
\setlength{\tabcolsep}{5pt}
\renewcommand{\arraystretch}{1.12}
\caption{Simple scenario RMSE estimator performance for varying $\lambda$}
\label{tab:lowtech_lambda_sweep_overall}
\begin{tabular}{@{}C{1.5cm} C{1.7cm} S[table-format=1.3] S[table-format=1.3] S[table-format=1.3] S[table-format=1.3]@{}}
\toprule
\textbf{$\lambda$} & \textbf{Parameter} & {\shortstack{\textbf{RMSE}\\\textbf{QBHM}}} & {\shortstack{\textbf{RMSE}\\\textbf{Unpooled}}} & {\shortstack{\textbf{Coverage}\\\textbf{QBHM}}} & {\shortstack{\textbf{Coverage}\\\textbf{Unpooled}}} \\
\midrule
$1$ & $a_1$    & 0.564 & 0.269 & 0.962 & 0.940 \\
    & $\rho$ & 0.056 & 0.038 & 0.962 & 0.958 \\
\addlinespace
$1/2$ & $a_1$    & 0.380 & 0.269 & 0.970 & 0.940 \\
      & $\rho$ & 0.043 & 0.038 & 0.970 & 0.958 \\
\addlinespace
$1/3$ & $a_1$    & 0.308 & 0.269 & 0.971 & 0.940 \\
      & $\rho$ & 0.038 & 0.038 & 0.970 & 0.958 \\
\bottomrule
\end{tabular}
\end{table}

The following robustness table reports the prior-sensitivity check. The same simulation is rerun with priors centered at zero and prior standard deviations set to more than 150 times the true data-generating-process scale.  Panel A reports the varying-instrument-strength scenario, where the estimator comparison is unchanged although the wider priors increase QBHM RMSE.  Panel B reports the varying-sample-size scenario, where the QBHM has higher RMSE than the unpooled estimator for parameter \(a_1\) at \(\lambda=1\), while Panel C shows that the QBHM advantage returns at \(\lambda=1/2\).  The robustness results reflect the same bias--variance tradeoff emphasized in the main text: very diffuse priors reduce information about the hyperparameters and can change the amount of pooling favored by MSE.  The robustness check therefore supports reporting the full \(\lambda\)-path rather than treating one default pooling value as definitive.

\begin{table}[htbp]
\centering
\scriptsize
\caption{Model performance under very diffuse priors for varying instrument strength, group size, and pooling weight}
\label{tab:lowtech_instrument_verydiffuse_overall}
\begin{tabular}{ccccccc}
\toprule
Panel & Setting & Parameter & \shortstack{RMSE\\QBHM} & \shortstack{RMSE\\Unpooled} & \shortstack{Coverage\\QBHM} & \shortstack{Coverage\\Unpooled} \\
\midrule
A. Instrument strength & Very strong & $a_1$ & 0.385 & 0.207 & 0.981 & 0.949 \\
                       & Weak        & $a_1$ & 1.857 & 2.092 & 1.000 & 0.876 \\
                       & Very weak   & $a_1$ & 2.219 & 2.381 & 1.000 & 0.887 \\
                       & Very strong & $\rho$ & 0.052 & 0.035 & 0.980 & 0.961 \\
                       & Weak        & $\rho$ & 0.100 & 0.232 & 1.000 & 0.954 \\
                       & Very weak   & $\rho$ & 0.109 & 0.247 & 1.000 & 0.944 \\
\addlinespace
B. Group size, $\lambda=1$ & 1000 & $a_1$ & 0.396 & 0.213 & 0.983 & 0.946 \\
                           & 50   & $a_1$ & 2.470 & 1.930 & 0.999 & 0.872 \\
                           & 20   & $a_1$ & 3.178 & 2.610 & 0.999 & 0.834 \\
                           & 1000 & $\rho$ & 0.051 & 0.036 & 0.982 & 0.953 \\
                           & 50   & $\rho$ & 0.109 & 0.197 & 0.999 & 0.978 \\
                           & 20   & $\rho$ & 0.123 & 0.265 & 0.999 & 0.963 \\
\addlinespace
C. Group size, $\lambda=1/2$ & 1000 & $a_1$ & 0.275 & 0.213 & 0.981 & 0.946 \\
                           & 50   & $a_1$ & 1.507 & 1.930 & 0.996 & 0.872 \\
                           & 20   & $a_1$ & 2.254 & 2.610 & 0.998 & 0.834 \\
                           & 1000 & $\rho$ & 0.039 & 0.036 & 0.980 & 0.953 \\
                           & 50   & $\rho$ & 0.093 & 0.197 & 0.996 & 0.978 \\
                           & 20   & $\rho$ & 0.110 & 0.265 & 0.997 & 0.963 \\
\bottomrule
\end{tabular}
\end{table}

Table~\ref{tab:fullmodel_4strongiv_design_priors} reports details for the nonconvex production-function simulation described in Section~\ref{sec:structural}.  The capital transition equation is group-specific, and its parameters are calibrated to match the four source studies.  Thus the simulated first-stage strength of the treatment instrument approximately matches the first-stage strength in the real data.  Results are reported in Tables~\ref{tab:fullmodel_4strongiv_overall_rmse_coverage} and~\ref{tab:fullmodel_4strongiv_switching_accuracy}: Table~\ref{tab:fullmodel_4strongiv_overall_rmse_coverage} reports parameter-estimation performance, while Table~\ref{tab:fullmodel_4strongiv_switching_accuracy} reports the implied capital switching thresholds and the implied proportion of the sample below each switching threshold.  Columns 3 and 4 report averages across replications, and columns 5 and 6 report RMSE.  In some cases the unpooled estimator has an average estimate closer to the true value than the QBHM average estimate but still has greater RMSE because its estimates are more dispersed across replications.

\begin{table}[H]
\centering
\scriptsize
\setlength{\tabcolsep}{6pt}
\renewcommand{\arraystretch}{1.05}
\caption{Nonconvex production function simulation: common parameters, priors, and varied settings}
\label{tab:fullmodel_4strongiv_design_priors}
\begin{tabular}{@{}C{0.28\textwidth} C{0.64\textwidth}@{}}
\toprule
\textbf{Component} & \textbf{Specification} \\
\midrule
Scenario & Full-model synthetic design calibrated to the four retained strong-IV studies \\
Source studies & GE, AB, FY, RC \\
Sample sizes & GE = 3354; AB = 562; FY = 665; RC = 493 \\
Number of synthetic groups & 4 \\
Replications & 98 Monte Carlo replications. In each replication, households are resampled within study; the capital transition shock, revenue shocks, and group-level production parameters are redrawn.  \\
Baseline resampling & Bootstrap within each source study, preserving empirical baseline capital and treatment distribution \\
Baseline jitter & Log-normal jitter with \texttt{jitter\_sd = 0.02} \\
Instrument vector & $(1,\ \mathrm{treatment})$ \\
Pooling strength & $\lambda = 1$ \\
Capital transition mode & Group-specific \\
Capital transition equation & $\log K_{ig,\mathrm{post}} = \mu_g + \delta_{Tg}T_{ig} + \delta_{Pg}\left(\log K_{ig,\mathrm{pre}} - \overline{\log K}_{g,\mathrm{pre}}\right) + \nu_{ig}$, with $\nu_{ig} \sim \mathcal{N}(0,\sigma_{\nu g}^2)$ \\
GE transition & $\mu_g = 2.678279$, $\delta_{Tg} = 0.171453$, $\delta_{Pg} = 0.641673$, $\sigma_{\nu g} = 1.511258$ \\
AB transition & $\mu_g = 7.904821$, $\delta_{Tg} = 0.382707$, $\delta_{Pg} = 0.430722$, $\sigma_{\nu g} = 0.850527$ \\
FY transition & $\mu_g = 6.489923$, $\delta_{Tg} = 0.232051$, $\delta_{Pg} = 0.685259$, $\sigma_{\nu g} = 0.865686$ \\
RC transition & $\mu_g = 7.137737$, $\delta_{Tg} = 0.302037$, $\delta_{Pg} = 0.684101$, $\sigma_{\nu g} = 0.788914$ \\
Revenue shock distribution & Student-$t$ with 3 degrees of freedom \\
Revenue shock scale & $\sigma_{\varepsilon,\mathrm{pre}} = 0.3$, $\sigma_{\varepsilon,\mathrm{post}} = 0.3$ \\
Data-generating process for $\log a_{1g}$ & $\log a_{1g} \sim \mathcal{N}(\log(0.8),\,0.3^2)$ \\
Data-generating process for $\rho_g$ & $\operatorname{logit}(\rho_g) \sim \mathcal{N}(\operatorname{logit}(0.3),\,0.4^2)$ \\
Data-generating process for $\log a_{2g}$ & $\log a_{2g} \sim \mathcal{N}(\log(2.0),\,0.3^2)$ \\
Data-generating process for $K_g$ & $K_g \sim \mathcal{N}(250,\,25^2)$ \\
Stan settings & 2 chains; 1000 warmup and 1000 sampling iterations per chain \\
\midrule
\multicolumn{2}{@{}l}{\textbf{QBHM priors}} \\
\midrule
Prior for $\log a_{1,\mu}$ & $\log a_{1,\mu} \sim \mathcal{N}(\log(1.0),\,1.0^2)$ \\
Prior for $\sigma_{a_1}$ & $\sigma_{a_1} \sim \mathcal{N}^{+}(0,\,0.7^2)$ \\
Group-level raw effects for $a_1$ & $z_{a_1,g} \sim \mathcal{N}(0,\,1)$ \\
Prior for $\mu_\rho$ & $\mu_\rho \sim \mathcal{N}(\operatorname{logit}(0.35),\,1.0^2)$ \\
Prior for $\sigma_\rho$ & $\sigma_\rho \sim \mathcal{N}^{+}(0,\,0.5^2)$ \\
Group-level raw effects for $\rho$ & $z_{\rho,g} \sim \mathcal{N}(0,\,1)$ \\
Prior for $\log a_{2,\mu}$ & $\log a_{2,\mu} \sim \mathcal{N}(\log(2.5),\,1.0^2)$ \\
Prior for $\sigma_{a_2}$ & $\sigma_{a_2} \sim \mathcal{N}^{+}(0,\,0.7^2)$ \\
Group-level raw effects for $a_2$ & $z_{a_2,g} \sim \mathcal{N}(0,\,1)$ \\
Prior for $K_\mu$ & $K_\mu \sim \mathcal{N}^{+}(250,\,250^2)$ \\
Prior for $\sigma_K$ & $\sigma_K \sim \mathcal{N}^{+}(0,\,125^2)$ \\
Group-level raw effects for $K$ & $z_{K,g} \sim \mathcal{N}^{+}(0,\,1)$ \\
\bottomrule
\end{tabular}
\end{table}

\begin{table}[H]
\centering
\footnotesize
\setlength{\tabcolsep}{5pt}
\renewcommand{\arraystretch}{1.12}
\caption{Full-model synthetic estimator performance over 98 completed real-calibrated four-study replications}
\label{tab:fullmodel_4strongiv_overall_rmse_coverage}
\begin{tabular}{@{}C{2.5cm} S[table-format=2.3] S[table-format=2.3] S[table-format=1.3] S[table-format=1.3]@{}}
\toprule
\textbf{Parameter} & {\shortstack{\textbf{RMSE}\\\textbf{QBHM}}} & {\shortstack{\textbf{RMSE}\\\textbf{Unpooled}}} & {\shortstack{\textbf{Coverage}\\\textbf{QBHM}}} & {\shortstack{\textbf{Coverage}\\\textbf{Unpooled}}} \\
\midrule
$a_1$ & 0.358 & 1.061 & 0.995 & 0.268 \\
$\rho$ & 0.097 & 0.121 & 1.000 & 0.240 \\
$a_2$ & 0.472 & 0.674 & 0.982 & 0.704 \\
$K_{\mathrm{thresh}}$ & 48.668 & 51.930 & 0.964 & 0.727 \\
\bottomrule
\end{tabular}
\end{table}

\begin{table}[H]
\centering
\footnotesize
\setlength{\tabcolsep}{4pt}
\renewcommand{\arraystretch}{1.12}
\caption{Synthetic full-model switching-point accuracy over 98 completed real-calibrated four-study replications}
\label{tab:fullmodel_4strongiv_switching_accuracy}
\begin{tabular}{@{}C{1.4cm} C{1.7cm} S[table-format=3.3] S[table-format=3.3] S[table-format=1.3] S[table-format=2.3] S[table-format=1.3]@{}}
\toprule
\textbf{Group} & \textbf{Estimator} & {\textbf{$K_{\mathrm{thresh}}$}} & {\textbf{Switching capital}} & {\textbf{Prop. below}} & {\textbf{RMSE switch}} & {\textbf{RMSE prop.}} \\
\midrule
GE & True & 251.797 & 254.493 & 0.943 & 0.000 & 0.000 \\
GE & QBHM & 221.153 & 223.244 & 0.934 & 45.791 & 0.014 \\
GE & Unpooled & 232.840 & 235.415 & 0.937 & 46.487 & 0.015 \\
AB & True & 247.056 & 249.785 & 0.006 & 0.000 & 0.000 \\
AB & QBHM & 215.837 & 217.028 & 0.004 & 46.240 & 0.003 \\
AB & Unpooled & 243.356 & 243.609 & 0.006 & 48.139 & 0.003 \\
FY & True & 248.053 & 250.969 & 0.204 & 0.000 & 0.000 \\
FY & QBHM & 206.287 & 207.784 & 0.167 & 54.533 & 0.048 \\
FY & Unpooled & 196.804 & 198.565 & 0.159 & 63.209 & 0.056 \\
RC & True & 246.000 & 248.589 & 0.089 & 0.000 & 0.000 \\
RC & QBHM & 208.871 & 210.222 & 0.073 & 52.245 & 0.023 \\
RC & Unpooled & 218.935 & 220.079 & 0.079 & 51.713 & 0.021 \\
\bottomrule
\end{tabular}
\end{table}

\section{Empirical priors and sensitivity tables}\label{app:emp}

This appendix reports the hierarchical priors and sensitivity checks for the empirical application in Section~\ref{sec:structural}.  Table~\ref{tab:real_fullmodel_priors} reports the priors used for the nonconvex production-function QBHM in the primary specification.  The parameter estimates from this model are reported in Table~\ref{tab:real_fullmodel_4strongiv_parameter_estimates}.  The QBHM interval reported in the tables is the central 95\% quasi-posterior interval, which is a quasi-posterior summary of the hierarchical estimator, while the unpooled estimator uses a Wald interval; these appendix intervals should be read together with the weak-identification inference discussion in Section~\ref{sec:testing}.

\begin{table}[H]
\centering
\footnotesize
\setlength{\tabcolsep}{6pt}
\renewcommand{\arraystretch}{1.14}
\caption{Prior specification for the empirical full-model QBHM (\(\lambda = 1\))}
\label{tab:real_fullmodel_priors}
\begin{tabular}{@{}L{0.34\textwidth} L{0.56\textwidth}@{}}
\toprule
\textbf{Component} & \textbf{Prior} \\
\midrule
Prior for \(\log a_{1,\mu}\) & \(\log a_{1,\mu} \sim \mathcal{N}(\log(1.0), 1.0^2)\) \\
Prior for \(\sigma_{a_1}\) & \(\sigma_{a_1} \sim \mathcal{N}^{+}(0, 0.7^2)\) \\
Prior for \(\mu_{\rho}\) & \(\mu_{\rho} \sim \mathcal{N}(\mathrm{logit}(0.35), 1.0^2)\) \\
Prior for \(\sigma_{\rho}\) & \(\sigma_{\rho} \sim \mathcal{N}^{+}(0, 0.5^2)\) \\
Prior for \(\log a_{2,\mu}\) & \(\log a_{2,\mu} \sim \mathcal{N}(\log(2.5), 1.0^2)\) \\
Prior for \(\sigma_{a_2}\) & \(\sigma_{a_2} \sim \mathcal{N}^{+}(0, 0.7^2)\) \\
Prior for \(K_{\mu}\) & \(K_{\mu} \sim \mathcal{N}(250, 250^2)\) \\
Prior for \(K_{\sigma}\) & \(K_{\sigma} \sim \mathcal{N}^{+}(0, 125^2)\) \\
\bottomrule
\end{tabular}
\end{table}
\normalsize

\begin{table}[H]
\centering
\scriptsize
\setlength{\tabcolsep}{3.5pt}
\renewcommand{\arraystretch}{1.12}
\caption{Real-data full-model parameter estimates for the four retained studies}
\label{tab:real_fullmodel_4strongiv_parameter_estimates}
\begin{tabular}{@{}ccS[table-format=3.3] C{2.6cm} S[table-format=3.3] C{2.6cm}@{}}
\toprule
\textbf{Group} & \textbf{Parameter} & {\textbf{QBHM}} & \textbf{95\% interval} & {\textbf{Unpooled}} & \textbf{95\% interval} \\
\midrule
GE & $a_1$ & 91.772 & [82.186, 102.171] & 150.018 & [138.223, 161.812] \\
GE & $\rho$ & 0.631 & [0.596, 0.667] & 0.403 & [0.384, 0.421] \\
GE & $a_2$ & 0.685 & [0.048, 3.516] & 14.776 & [11.585, 17.966] \\
GE & $K_{\mathrm{thresh}}$ & 435.655 & [66.821, 908.342] & 71.078 & [67.922, 74.234] \\
AB & $a_1$ & 181.845 & [73.639, 400.824] & 0.057 & [0.057, 0.057] \\
AB & $\rho$ & 0.613 & [0.507, 0.708] & 0.350 & [0.350, 0.350] \\
AB & $a_2$ & 0.456 & [0.058, 2.516] & 6.659 & [6.181, 7.137] \\
AB & $K_{\mathrm{thresh}}$ & 439.799 & [66.146, 918.579] & 0.330 & [0.330, 0.330] \\
FY & $a_1$ & 2.967 & [1.178, 9.955] & 31.885 & [0.000, 146.818] \\
FY & $\rho$ & 0.592 & [0.327, 0.726] & 0.044 & [0.029, 0.059] \\
FY & $a_2$ & 0.152 & [0.069, 0.217] & 0.161 & [0.000, 0.328] \\
FY & $K_{\mathrm{thresh}}$ & 446.301 & [66.819, 955.756] & 16.151 & [14.675, 17.627] \\
RC & $a_1$ & 12.952 & [5.903, 40.650] & 0.369 & [0.369, 0.369] \\
RC & $\rho$ & 0.580 & [0.383, 0.665] & 0.240 & [0.240, 0.240] \\
RC & $a_2$ & 0.196 & [0.038, 0.389] & 0.487 & [0.357, 0.618] \\
RC & $K_{\mathrm{thresh}}$ & 435.031 & [67.457, 915.490] & 0.038 & [0.038, 0.038] \\
\bottomrule
\end{tabular}
\end{table}

The four sensitivity tables below report the estimation results from varying the pooling strength $\lambda$.  The primary model, with $\lambda=1$, is rerun for $\lambda=2$ and $\lambda=0.5$.  The unpooled estimator results stay the same, while the QBHM results shift slightly but not substantively, and the presence and approximate location of nonconvexities remain the same across $\lambda$ values.

\begin{table}[H]
\centering
\scriptsize
\setlength{\tabcolsep}{3.5pt}
\renewcommand{\arraystretch}{1.12}
\caption{Real-data full-model parameter estimates ($\lambda = 2$)}
\label{tab:real_fullmodel_4strongiv_parameter_estimates_lambda_2}
\begin{tabular}{@{}ccS[table-format=3.3] C{2.6cm} S[table-format=3.3] C{2.6cm}@{}}
\toprule
\textbf{Group} & \textbf{Parameter} & {\textbf{QBHM mean}} & \textbf{95\% interval} & {\textbf{Unpooled}} & \textbf{95\% interval} \\
\midrule
GE & $a_1$ & 92.249 & [79.136, 106.163] & 150.018 & [138.223, 161.812] \\
GE & $\rho$ & 0.630 & [0.582, 0.681] & 0.403 & [0.384, 0.421] \\
GE & $a_2$ & 0.650 & [0.054, 3.887] & 14.776 & [11.585, 17.966] \\
GE & $K_{\mathrm{thresh}}$ & 431.394 & [63.551, 915.566] & 71.078 & [67.922, 74.234] \\
AB & $a_1$ & 188.298 & [63.214, 519.959] & 0.057 & [0.057, 0.057] \\
AB & $\rho$ & 0.614 & [0.476, 0.728] & 0.350 & [0.350, 0.350] \\
AB & $a_2$ & 0.480 & [0.051, 2.755] & 6.659 & [6.181, 7.137] \\
AB & $K_{\mathrm{thresh}}$ & 434.411 & [68.303, 946.398] & 0.330 & [0.330, 0.330] \\
FY & $a_1$ & 3.031 & [1.045, 10.096] & 31.885 & [0.000, 146.818] \\
FY & $\rho$ & 0.590 & [0.324, 0.744] & 0.044 & [0.029, 0.059] \\
FY & $a_2$ & 0.154 & [0.049, 0.233] & 0.161 & [0.000, 0.328] \\
FY & $K_{\mathrm{thresh}}$ & 438.112 & [66.243, 936.691] & 16.151 & [14.675, 17.627] \\
RC & $a_1$ & 12.755 & [4.895, 38.933] & 0.369 & [0.369, 0.369] \\
RC & $\rho$ & 0.586 & [0.395, 0.689] & 0.240 & [0.240, 0.240] \\
RC & $a_2$ & 0.203 & [0.036, 0.404] & 0.487 & [0.357, 0.618] \\
RC & $K_{\mathrm{thresh}}$ & 433.743 & [57.239, 922.095] & 0.038 & [0.038, 0.038] \\
\bottomrule
\end{tabular}
\end{table}

\begin{table}[H]
\centering
\footnotesize
\setlength{\tabcolsep}{4pt}
\renewcommand{\arraystretch}{1.12}
\caption{Implied switching capital by study and estimator ($\lambda = 2$)}
\label{tab:real_fullmodel_4strongiv_switching_thresholds_lambda_2}
\begin{tabular}{@{}C{1.4cm} C{1.8cm} S[table-format=3.3] S[table-format=7.3] S[table-format=1.3]@{}}
\toprule
\textbf{Group} & \textbf{Estimator} & {\textbf{$K_{\mathrm{thresh}}$}} & {\textbf{Switching capital}} & {\textbf{Prop. below}} \\
\midrule
GE & QBHM & 431.394 & 648416.369 & 1.000 \\
GE & Unpooled & 71.078 & 146.796 & 0.880 \\
AB & QBHM & 434.411 & 5172586.488 & 1.000 \\
AB & Unpooled & 0.330 & 0.336 & 0.000 \\
FY & QBHM & 438.112 & 2361.562 & 0.806 \\
FY & Unpooled & 16.151 & 269.617 & 0.227 \\
RC & QBHM & 433.743 & 23031.923 & 0.998 \\
RC & Unpooled & 0.038 & 0.743 & 0.000 \\
\bottomrule
\end{tabular}
\end{table}

\begin{table}[H]
\centering
\scriptsize
\setlength{\tabcolsep}{3.5pt}
\renewcommand{\arraystretch}{1.12}
\caption{Real-data full-model parameter estimates ($\lambda = 0.5$)}
\label{tab:real_fullmodel_4strongiv_parameter_estimates_lambda_0_5}
\begin{tabular}{@{}ccS[table-format=3.3] C{2.6cm} S[table-format=3.3] C{2.6cm}@{}}
\toprule
\textbf{Group} & \textbf{Parameter} & {\textbf{QBHM mean}} & \textbf{95\% interval} & {\textbf{Unpooled}} & \textbf{95\% interval} \\
\midrule
GE & $a_1$ & 91.328 & [84.451, 99.004] & 150.018 & [138.223, 161.812] \\
GE & $\rho$ & 0.632 & [0.606, 0.659] & 0.403 & [0.384, 0.421] \\
GE & $a_2$ & 0.728 & [0.052, 4.076] & 14.776 & [11.585, 17.966] \\
GE & $K_{\mathrm{thresh}}$ & 429.394 & [68.162, 919.402] & 71.078 & [67.922, 74.234] \\
AB & $a_1$ & 177.677 & [83.737, 315.897] & 0.057 & [0.057, 0.057] \\
AB & $\rho$ & 0.617 & [0.535, 0.695] & 0.350 & [0.350, 0.350] \\
AB & $a_2$ & 0.456 & [0.053, 2.650] & 6.659 & [6.181, 7.137] \\
AB & $K_{\mathrm{thresh}}$ & 430.969 & [68.120, 933.604] & 0.330 & [0.330, 0.330] \\
FY & $a_1$ & 2.964 & [1.182, 10.399] & 31.885 & [0.000, 146.818] \\
FY & $\rho$ & 0.595 & [0.327, 0.727] & 0.044 & [0.029, 0.059] \\
FY & $a_2$ & 0.150 & [0.071, 0.210] & 0.161 & [0.000, 0.328] \\
FY & $K_{\mathrm{thresh}}$ & 434.546 & [68.777, 923.550] & 16.151 & [14.675, 17.627] \\
RC & $a_1$ & 13.669 & [6.643, 41.163] & 0.369 & [0.369, 0.369] \\
RC & $\rho$ & 0.575 & [0.379, 0.650] & 0.240 & [0.240, 0.240] \\
RC & $a_2$ & 0.188 & [0.039, 0.382] & 0.487 & [0.357, 0.618] \\
RC & $K_{\mathrm{thresh}}$ & 431.345 & [68.195, 916.697] & 0.038 & [0.038, 0.038] \\
\bottomrule
\end{tabular}
\end{table}

\begin{table}[H]
\centering
\footnotesize
\setlength{\tabcolsep}{4pt}
\renewcommand{\arraystretch}{1.12}
\caption{Implied switching capital by study and estimator ($\lambda = 0.5$)}
\label{tab:real_fullmodel_4strongiv_switching_thresholds_lambda_0_5}
\begin{tabular}{@{}C{1.4cm} C{1.8cm} S[table-format=3.3] S[table-format=7.3] S[table-format=1.3]@{}}
\toprule
\textbf{Group} & \textbf{Estimator} & {\textbf{$K_{\mathrm{thresh}}$}} & {\textbf{Switching capital}} & {\textbf{Prop. below}} \\
\midrule
GE & QBHM & 429.394 & 513391.716 & 1.000 \\
GE & Unpooled & 71.078 & 146.796 & 0.880 \\
AB & QBHM & 430.969 & 5801600.082 & 1.000 \\
AB & Unpooled & 0.330 & 0.336 & 0.000 \\
FY & QBHM & 434.546 & 2498.667 & 0.817 \\
FY & Unpooled & 16.151 & 269.617 & 0.227 \\
RC & QBHM & 431.345 & 25056.059 & 0.998 \\
RC & Unpooled & 0.038 & 0.743 & 0.000 \\
\bottomrule
\end{tabular}
\end{table}

\medskip

\bibliography{overall_bibliography}

@article{AM2016conditional,
  author       = {Andrews, Isaiah and Mikusheva, Anna},
  title        = {Conditional Inference with a Functional Nuisance Parameter},
  journal      = {Econometrica},
  volume       = {84},
  number       = {4},
  pages        = {1571--1612},
  year         = {2016},
  doi          = {10.3982/ECTA12868}
}

@article{AM2022,
  author       = {Andrews, Isaiah and Mikusheva, Anna},
  title        = {Optimal Decision Rules for Weak {GMM}},
  journal      = {Econometrica},
  volume       = {90},
  number       = {2},
  pages        = {715--748},
  year         = {2022},
  doi          = {10.3982/ECTA18678}
}

@misc{AndrewsMikusheva2023Inadmissible,
  author       = {Andrews, Isaiah and Mikusheva, Anna},
  title        = {{GMM} is Inadmissible Under Weak Identification},
  year         = {2023},
  doi          = {10.48550/arXiv.2204.12462},
  url          = {https://arxiv.org/abs/2204.12462},
  eprint       = {2204.12462},
  archivePrefix = {arXiv},
  primaryClass  = {econ.EM},
  note         = {arXiv preprint, revised May 2023; first circulated 2022}
}

@article{AndrewsMoreiraStock2006,
  author       = {Andrews, Donald W. K. and Moreira, Marcelo J. and Stock, James H.},
  title        = {Optimal Two-Sided Invariant Similar Tests for Instrumental Variables Regression},
  journal      = {Econometrica},
  volume       = {74},
  number       = {3},
  pages        = {715--752},
  year         = {2006},
  doi          = {10.1111/j.1468-0262.2006.00680.x}
}

@article{Angrist2010,
  author       = {Angrist, Joshua D. and Pischke, J{\"o}rn-Steffen},
  title        = {The Credibility Revolution in Empirical Economics: How Better Research Design Is Taking the Con out of Econometrics},
  journal      = {Journal of Economic Perspectives},
  volume       = {24},
  number       = {2},
  pages        = {3--30},
  year         = {2010},
  doi          = {10.1257/jep.24.2.3}
}

@article{BissiriHolmesWalker2016,
  author       = {Bissiri, Pier Giovanni and Holmes, Chris C. and Walker, Stephen G.},
  title        = {A General Framework for Updating Belief Distributions},
  journal      = {Journal of the Royal Statistical Society: Series B},
  volume       = {78},
  number       = {5},
  pages        = {1103--1130},
  year         = {2016},
  doi          = {10.1111/rssb.12158}
}

@article{CaiSzeidl2024,
  author       = {Cai, Jing and Szeidl, Adam},
  title        = {Indirect Effects of Access to Finance},
  journal      = {American Economic Review},
  volume       = {114},
  number       = {8},
  pages        = {2308--2351},
  year         = {2024},
  doi          = {10.1257/aer.20220711}
}

@techreport{ChernozhukovHansenKongWang2025PlausibleGMM,
  author       = {Chernozhukov, Victor and Hansen, Christian B. and Kong, Lingwei and Wang, Weining},
  title        = {Plausible {GMM}: A Quasi-Bayesian Approach},
  number       = {CWP07/26},
  institution  = {Centre for Microdata Methods and Practice},
  type         = {Working Paper},
  year         = {2026},
  month        = may,
  doi          = {10.47004/wp.cem.2026.0726},
  url          = {https://cemmap.ac.uk/publication/plausible-gmm-a-quasi-bayesian-approach-2/},
  eprint       = {2507.00555},
  archivePrefix = {arXiv},
  note         = {Published 8 May 2026; previous version CWP14/25, doi:10.47004/wp.cem.2025.1425}
}

@article{ChernozhukovHong2003,
  author       = {Chernozhukov, Victor and Hong, Han},
  title        = {An {MCMC} Approach to Classical Estimation},
  journal      = {Journal of Econometrics},
  volume       = {115},
  number       = {2},
  pages        = {293--346},
  year         = {2003},
  doi          = {10.1016/S0304-4076(03)00100-3}
}

@article{ChibShinSimoni2018MomentConditions,
  author       = {Chib, Siddhartha and Shin, Minchul and Simoni, Anna},
  title        = {{Bayesian} Estimation and Comparison of Moment Condition Models},
  journal      = {Journal of the American Statistical Association},
  volume       = {113},
  number       = {524},
  pages        = {1656--1668},
  year         = {2018},
  doi          = {10.1080/01621459.2017.1358172}
}

@article{Egami2023,
  author       = {Egami, Naoki and Hartman, Erin},
  title        = {Elements of External Validity: Framework, Design, and Analysis},
  journal      = {American Political Science Review},
  volume       = {117},
  number       = {3},
  pages        = {1070--1088},
  year         = {2023},
  doi          = {10.1017/S0003055422000880}
}

@article{EggerEtAl2022,
  author       = {Egger, Dennis and Haushofer, Johannes and Miguel, Edward and Niehaus, Paul and Walker, Michael},
  title        = {General Equilibrium Effects of Cash Transfers: Experimental Evidence from {Kenya}},
  journal      = {Econometrica},
  volume       = {90},
  number       = {6},
  pages        = {2603--2643},
  year         = {2022},
  doi          = {10.3982/ECTA17945}
}

@misc{FinanPouzo2026,
  author       = {Finan, Frederico and Pouzo, Demian},
  title        = {Learning about Treatment Effects with Prior Studies: A Bayesian Model Averaging Approach},
  year         = {2026},
  doi          = {10.48550/arXiv.2601.09888},
  url          = {https://arxiv.org/abs/2601.09888},
  eprint       = {2601.09888},
  archivePrefix = {arXiv},
  primaryClass  = {econ.EM},
  note         = {arXiv preprint, version 2, revised March 2026}
}

@article{FinkJackMasiye2020,
  author       = {Fink, G{\"u}nther and Jack, B. Kelsey and Masiye, Felix},
  title        = {Seasonal Liquidity, Rural Labor Markets, and Agricultural Production},
  journal      = {American Economic Review},
  volume       = {110},
  number       = {11},
  pages        = {3351--3392},
  year         = {2020},
  doi          = {10.1257/aer.20180607}
}

@book{GelmanHill2007,
  author       = {Gelman, Andrew and Hill, Jennifer},
  title        = {Data Analysis Using Regression and Multilevel/Hierarchical Models},
  publisher    = {Cambridge University Press},
  address      = {New York},
  year         = {2007},
  doi          = {10.1017/CBO9780511790942}
}

@article{HolmesWalker2017,
  author       = {Holmes, Chris C. and Walker, Stephen G.},
  title        = {Assigning a value to a power likelihood in a general {B}ayesian model},
  journal      = {Biometrika},
  volume       = {104},
  number       = {2},
  pages        = {497--503},
  year         = {2017},
  doi          = {10.1093/biomet/asx010}
}

@unpublished{Huang2023GroupedPanels,
  author       = {Huang, Jiaming},
  title        = {{Quasi-Bayesian Inference for Grouped Panels}},
  year         = {2023},
  month        = dec,
  url          = {https://jiaminghuang.net/files/qbc.pdf},
  note         = {Job market paper; latest public PDF dated December 30, 2023; author page last updated January 2026}
}

@incollection{JamesStein1961,
  author       = {James, William and Stein, Charles},
  editor       = {Neyman, Jerzy},
  title        = {Estimation with Quadratic Loss},
  booktitle    = {Proceedings of the Fourth Berkeley Symposium on Mathematical Statistics and Probability, Volume 1: Contributions to the Theory of Statistics},
  pages        = {361--379},
  publisher    = {University of California Press},
  address      = {Berkeley},
  year         = {1961},
  url          = {https://projecteuclid.org/ebooks/berkeley-symposium-on-mathematical-statistics-and-probability/Estimation-with-Quadratic-Loss/chapter/Estimation-with-Quadratic-Loss/bsmsp/1200512173}
}

@misc{Kankanala2025GeneralizedBayesCMR,
  author       = {Kankanala, Sid},
  title        = {Generalized {Bayes} in Conditional Moment Restriction Models},
  year         = {2025},
  doi          = {10.48550/arXiv.2510.01036},
  url          = {https://arxiv.org/abs/2510.01036},
  eprint       = {2510.01036},
  archivePrefix = {arXiv},
  primaryClass  = {econ.EM}
}

@article{LyddonHolmesWalker2019,
  author       = {Lyddon, Simon P. and Holmes, Chris C. and Walker, Stephen G.},
  title        = {General {Bayesian} Updating and the Loss-Likelihood Bootstrap},
  journal      = {Biometrika},
  volume       = {106},
  number       = {2},
  pages        = {465--478},
  year         = {2019},
  doi          = {10.1093/biomet/asz006}
}

@article{Meager2019,
  author       = {Meager, Rachael},
  title        = {Understanding the Average Impact of Microcredit Expansions: A {Bayesian} Hierarchical Analysis of Seven Randomized Experiments},
  journal      = {American Economic Journal: Applied Economics},
  volume       = {11},
  number       = {1},
  pages        = {57--91},
  year         = {2019},
  doi          = {10.1257/app.20170299}
}

@article{Moreira2003,
  author       = {Moreira, Marcelo J.},
  title        = {A Conditional Likelihood Ratio Test for Structural Models},
  journal      = {Econometrica},
  volume       = {71},
  number       = {4},
  pages        = {1027--1048},
  year         = {2003},
  doi          = {10.1111/1468-0262.00438}
}

@article{Schennach2005BayesianETEL,
  author       = {Schennach, Susanne M.},
  title        = {{Bayesian} Exponentially Tilted Empirical Likelihood},
  journal      = {Biometrika},
  volume       = {92},
  number       = {1},
  pages        = {31--46},
  year         = {2005},
  doi          = {10.1093/biomet/92.1.31}
}

@article{Slough2022,
  author       = {Slough, Tara and Tyson, Scott A.},
  title        = {External Validity and Meta-Analysis},
  journal      = {American Journal of Political Science},
  volume       = {67},
  number       = {2},
  pages        = {440--455},
  year         = {2023},
  doi          = {10.1111/ajps.12742},
  note         = {First published online 2022}
}

@article{StaigerStock1997,
  author       = {Staiger, Douglas and Stock, James H.},
  title        = {Instrumental Variables Regression with Weak Instruments},
  journal      = {Econometrica},
  volume       = {65},
  number       = {3},
  pages        = {557--586},
  year         = {1997},
  doi          = {10.2307/2171753}
}

@incollection{StockYogo2005,
  author       = {Stock, James H. and Yogo, Motohiro},
  editor       = {Andrews, Donald W. K. and Stock, James H.},
  title        = {Testing for Weak Instruments in Linear {IV} Regression},
  booktitle    = {Identification and Inference for Econometric Models: Essays in Honor of Thomas Rothenberg},
  pages        = {80--108},
  publisher    = {Cambridge University Press},
  address      = {Cambridge},
  year         = {2005},
  doi          = {10.1017/CBO9780511614491.006}
}

@article{Vivalt2020,
  author       = {Vivalt, Eva},
  title        = {How Much Can We Generalize from Impact Evaluations?},
  journal      = {Journal of the European Economic Association},
  volume       = {18},
  number       = {6},
  pages        = {3045--3089},
  year         = {2020},
  doi          = {10.1093/jeea/jvaa019}
}

@incollection{Walters2024EmpiricalBayes,
  author       = {Walters, Christopher R.},
  editor       = {Dustmann, Christian and Lemieux, Thomas},
  title        = {Empirical {Bayes} Methods in Labor Economics},
  booktitle    = {Handbook of Labor Economics},
  volume       = {5},
  pages        = {183--260},
  publisher    = {Elsevier},
  year         = {2024},
  doi          = {10.1016/bs.heslab.2024.11.001}
}

@misc{You2026,
  author       = {You, Zhiheng},
  title        = {Using Prior Studies to Design Experiments: An Empirical Bayes Approach},
  year         = {2026},
  doi          = {10.48550/arXiv.2602.20581},
  url          = {https://arxiv.org/abs/2602.20581},
  eprint       = {2602.20581},
  archivePrefix = {arXiv},
  primaryClass  = {econ.EM},
  note         = {arXiv preprint, version 1}
}

@article{attanasio2015impacts,
  author       = {Attanasio, Orazio and Augsburg, Britta and De Haas, Ralph and Fitzsimons, Emla and Harmgart, Heike},
  title        = {The Impacts of Microfinance: Evidence from Joint-Liability Lending in {Mongolia}},
  journal      = {American Economic Journal: Applied Economics},
  volume       = {7},
  number       = {1},
  pages        = {90--122},
  year         = {2015},
  doi          = {10.1257/app.20130489}
}

@article{augsburg2015impacts,
  author       = {Augsburg, Britta and De Haas, Ralph and Harmgart, Heike and Meghir, Costas},
  title        = {The Impacts of Microcredit: Evidence from {Bosnia and Herzegovina}},
  journal      = {American Economic Journal: Applied Economics},
  volume       = {7},
  number       = {1},
  pages        = {183--203},
  year         = {2015},
  doi          = {10.1257/app.20130272}
}

@techreport{banerjee2019can,
  author       = {Banerjee, Abhijit and Breza, Emily and Duflo, Esther and Kinnan, Cynthia},
  title        = {Can Microfinance Unlock a Poverty Trap for Some Entrepreneurs?},
  number       = {26346},
  institution  = {National Bureau of Economic Research},
  type         = {Working Paper},
  year         = {2019},
  month        = oct,
  doi          = {10.3386/w26346},
  url          = {https://www.nber.org/papers/w26346},
  note         = {Revised July 2024; author page lists 2025 revision requested at American Economic Journal: Applied Economics}
}

@article{bari2024asset,
  author       = {Bari, Faisal and Malik, Kashif and Meki, Muhammad and Quinn, Simon},
  title        = {Asset-Based Microfinance for Microenterprises: Evidence from {Pakistan}},
  journal      = {American Economic Review},
  volume       = {114},
  number       = {2},
  pages        = {534--574},
  year         = {2024},
  doi          = {10.1257/aer.20210169}
}

@article{crepon2015estimating,
  author       = {Cr{\'e}pon, Bruno and Devoto, Florencia and Duflo, Esther and Parient{\'e}, William},
  title        = {Estimating the Impact of Microcredit on Those Who Take It Up: Evidence from a Randomized Experiment in {Morocco}},
  journal      = {American Economic Journal: Applied Economics},
  volume       = {7},
  number       = {1},
  pages        = {123--150},
  year         = {2015},
  doi          = {10.1257/app.20130535}
}

@article{de2008returns,
  author       = {de Mel, Suresh and McKenzie, David and Woodruff, Christopher},
  title        = {Returns to Capital in Microenterprises: Evidence from a Field Experiment},
  journal      = {The Quarterly Journal of Economics},
  volume       = {123},
  number       = {4},
  pages        = {1329--1372},
  year         = {2008},
  doi          = {10.1162/qjec.2008.123.4.1329}
}

@article{fafchamps2014microenterprise,
  author       = {Fafchamps, Marcel and McKenzie, David and Quinn, Simon and Woodruff, Christopher},
  title        = {Microenterprise Growth and the Flypaper Effect: Evidence from a Randomized Experiment in {Ghana}},
  journal      = {Journal of Development Economics},
  volume       = {106},
  pages        = {211--226},
  year         = {2014},
  doi          = {10.1016/j.jdeveco.2013.09.010}
}

@article{kaji2021weakidentification,
  author       = {Kaji, Tetsuya},
  title        = {Theory of Weak Identification in Semiparametric Models},
  journal      = {Econometrica},
  volume       = {89},
  number       = {2},
  pages        = {733--763},
  year         = {2021},
  doi          = {10.3982/ECTA16413}
}

@article{karlan2019debt,
  author       = {Karlan, Dean and Mullainathan, Sendhil and Roth, Benjamin N.},
  title        = {Debt Traps? Market Vendors and Moneylender Debt in {India} and the {Philippines}},
  journal      = {American Economic Review: Insights},
  volume       = {1},
  number       = {1},
  pages        = {27--42},
  year         = {2019},
  doi          = {10.1257/aeri.20180030}
}

\end{document}